\documentclass[a4paper]{article}
\usepackage{enumitem}
\usepackage{calc}

\usepackage{lmodern}
\usepackage{graphicx} 

\usepackage{amsmath}
\usepackage{amssymb,tikz}
\usepackage{pict2e}
\usepackage{amsthm}
\usepackage{amscd}
\usepackage{authblk}
\usepackage{mathrsfs}
\usetikzlibrary{calc}

\usepackage{yfonts}

\usepackage{marvosym}
\usepackage[percent]{overpic}

\newtheorem{theorem}{Theorem} [section]
\newtheorem{proposition}[theorem]{Proposition}	
\newtheorem{corollary}[theorem]{Corollary}	
\newtheorem{lemma}[theorem]{Lemma}		

\newtheorem{assumptions}[theorem]{Assumptions}

\newtheorem{remark}[theorem]{Remark}
\theoremstyle{definition}

\newcommand{\C}{\mathbb{C}}
\newcommand{\R}{\mathbb{R}}

\newcommand{\re}{\text{\upshape Re\,}}
\newcommand{\im}{\text{\upshape Im\,}}

\newcommand{\Ai}{{\rm Ai}}

\usepackage{tikz}
\usetikzlibrary{arrows}
\usetikzlibrary{decorations.pathmorphing}
\usetikzlibrary{decorations.markings}
\usetikzlibrary{patterns}
\usetikzlibrary{automata}
\usetikzlibrary{positioning}
\usepackage{tikz-cd}
\tikzset{->-/.style={decoration={
				markings,
				mark=at position #1 with {\arrow{latex}}},postaction={decorate}}}
	
	\tikzset{-<-/.style={decoration={
				markings,
				mark=at position #1 with {\arrowreversed{latex}}},postaction={decorate}}}

\usetikzlibrary{shapes.misc}\tikzset{cross/.style={cross out, draw, 
         minimum size=2*(#1-\pgflinewidth), 
         inner sep=0pt, outer sep=0pt}}

\usepackage{pgfplots}

\numberwithin{equation}{section}

\def\bigO{{\cal O}}

\usepackage[colorlinks=true]{hyperref}
\hypersetup{urlcolor=blue, citecolor=red, linkcolor=blue}

\begin{document}
\title{Uniform tail asymptotics for Airy kernel determinant solutions to KdV and for the narrow wedge solution to KPZ}
\author[1]{Christophe Charlier}
\author[2]{Tom Claeys}
\author[3]{Giulio Ruzza}
\renewcommand\Affilfont{\small}
\affil[1]{\textit{Department of Mathematical Sciences, University of Copenhagen, Universitetsparken 5, \newline 2100 Copenhagen, Denmark;} \texttt{charlier@math.ku.dk}}
\affil[2,3]{\textit{Institut de Recherche en Math\'ematique et Physique,  UCLouvain, Chemin du Cyclotron 2, B-1348 Louvain-la-Neuve, Belgium;} \texttt{tom.claeys@uclouvain.be}, \texttt{giulio.ruzza@uclouvain.be}}
\date{}
\maketitle

\begin{abstract}
We obtain uniform asymptotics for deformed Airy kernel determinants, 
which arise in models of finite temperature free fermions and which characterize the narrow wedge solution of the Kardar--Parisi--Zhang equation.
The asymptotics for the determinants yield uniform initial data for an associated family of solutions to the Korteweg--de\thinspace Vries equation, and uniform lower tail asymptotics for the narrow wedge solution of the Kardar--Parisi--Zhang equation. 
\end{abstract}

\medskip



\medskip

\section{Introduction}
We study Fredholm determinants of deformed Airy kernels of the form
\begin{equation}
\label{def:Q1}
Q_{\sigma}(x,t) := \det (1-\mathbb{K}_{\sigma,x,t}^{\mathrm{Ai}})
\end{equation}
where $\mathbb{K}_{\sigma,x,t}^{\mathrm{Ai}}$ is the integral operator acting on $L^2(\mathbb R)$ with kernel
\begin{align}
\label{def:Q2}
K_{\sigma,x,t}^{\rm Ai}(u,v) =\sigma(r(u;x,t))K^{\mathrm{Ai}}(u,v),\quad r(u;x,t)= \frac{u}{t^{2/3}}+\frac{x}{t},\end{align}
where $K^{\mathrm{Ai}}$ is the Airy kernel
\[
K^{\mathrm{Ai}}(u,v) = \frac{\Ai (u)\Ai' (v)-\Ai' (u)\Ai (v)}{u-v},
\]
$\Ai$ is the Airy function, and $\sigma:\mathbb R\to[0,1]$ belongs to a class of functions that we will specify below, and which is in particular such that $\mathbb K_{\sigma,x,t}^{\rm Ai}$ is trace class for all $x\in\mathbb R$ and $t>0$.
Expressing the Fredholm determinant as a Fredholm series, we have 
\begin{equation}
Q_\sigma(x,t)=\sum_{k=0}^\infty \frac{(-1)^k}{k!}\int_{\mathbb R^k}\det\left(K^{\mathrm{Ai}}(u_i,u_j)\right)_{i,j=1}^k\prod_{j=1}^k \sigma(r(u_j;x,t)) du_j.
\end{equation}

\medskip

Fredholm determinants of this form appear in various contexts. 
If $\sigma=1_{(0,+\infty)}$ is the Heaviside function, $Q_\sigma(x,t)$ is equal to the Tracy--Widom distribution  $F^{\rm TW}(s)$ evaluated at $s=-xt^{-1/3}$ \cite{TW}. In this case, 
the Tracy--Widom formula \cite{TW} states that
\begin{equation}\label{TW1}
\partial_x^2\log Q_{1_{(0,+\infty)}}(x,t)=-t^{-2/3}y^2\left(-xt^{-1/3}\right),
\end{equation}
where $y$ is the Hastings--McLeod solution of the Painlev\'e~II equation $y''=\xi y+2y^3$,
uniquely characterized by the asymptotics
\begin{equation}
\label{eq:asHMminusinfty}
y(\xi)\sim {\rm Ai}(\xi) \qquad \mbox{ as $\xi\to +\infty$.}
\end{equation}
Moreover, this solution also satisfies the asymptotics
\begin{equation}
\label{eq:asHMplusinfty}
y(\xi)=\sqrt{\frac{-\xi}{2}}\left(1+\frac{1}{8\xi^3}+\mathcal O\bigl(\xi^{-6}\bigr)\right) \qquad \mbox{ as $\xi\to -\infty$.}
\end{equation}
Upper and lower tail asymptotics for $\log Q_{1_{(0,+\infty)}}(x,t)$ as $s=-xt^{-1/3}$ goes to $\pm\infty$  can be derived from the Tracy--Widom formula \eqref{TW1}, the fact that $\lim_{x\to +\infty}\partial_x\log Q_{1_{(0,+\infty)}}(x,t)=0$, and the asymptotics \eqref{eq:asHMminusinfty} and~\eqref{eq:asHMplusinfty} for $y$: as $s=-xt^{-1/3}\to -\infty$, we have 
\begin{equation}\label{eq:TWtail}
\log F^{\rm TW}(-xt^{-1/3})=\log Q_{1_{(0,+\infty)}}(x,t)={-\frac{x^3}{12 t}}{-\frac{1}{8}\log (xt^{-1/3})}+c+o(1),
\end{equation}
and the value of $c$ was proven to be $c={\frac{\log 2}{24}}+{\zeta'(-1)}$ in \cite{DIK, BBdF}. Moreover, it is straightforward to verify that
\begin{equation}
\label{def:u0}
u_{1_{(0,+\infty)}}(x,t):=\partial_x^2\log Q_{1_{(0,+\infty)}}(x,t)+\frac{x}{2t}
\end{equation}
is a solution to the Korteweg--de\thinspace Vries (KdV) equation
\begin{equation}
\label{def:KdV}
\partial_t u+ 2u\partial_x u + \frac{1}{6}\partial_x^3u=0.
\end{equation}

\medskip

We will be interested here in $Q_\sigma$ for a class of smooth functions $\sigma$, and the most important choice of $\sigma$ for our concerns will be
\begin{align*}
\sigma(r)=\sigma_{\rm KPZ}(r):=\frac{1}{1+e^{-r}}.
\end{align*}
Then $Q_{\sigma}(x,t)$ characterizes the probability distribution of the narrow wedge solution of the KPZ equation, see e.g.~\cite{AmirCorwinQuastel, SasamotoSpohn, CLDR, BorodinGorin}, and encodes information about the 
tail probabilities of the KPZ solution.
See also \cite{LDMRS} for the $t\to+\infty$ asymptotics of $Q_\sigma(x,t)$ for $\sigma=\sigma_{\rm KPZ}$.
The lower tail expansion is particularly challenging, and it is connected to the $x\to +\infty$ asymptotics for $Q_\sigma(x,t)$, as we will explain in more detail in Section \ref{section:KPZ}. Lower tail  expansions have been obtained as $x\to +\infty$ for fixed $t$ \cite{CorwinGhosal, Tsai, CaCl2019, KPZKPDoussal, KrajenbrinkLeDoussal}, but without uniformity for $t$ small. In this work, we will complete the description of the lower tail of the KPZ solution with narrow wedge initial data, by deriving all growing terms in the $x\to+\infty$ asymptotics, with uniformity for $t$ small.

\medskip

By algebraic manipulations of the Fredholm determinant (see e.g.~\cite[Proof of Proposition 5.1]{AmirCorwinQuastel}), one can derive an alternative expression for $Q_{\sigma_{\rm KPZ}}$,
\begin{equation}
\label{def:Q3}
Q_{\sigma_{\rm KPZ}}(x,t)=\det\left(1-1_{(-xt^{-1/3},\infty)}\mathbb L^{\rm Ai}_{{\sigma_{\rm KPZ}},x,t}\right),
\end{equation}
where $1_{(-xt^{-1/3},\infty)}$ is the projection operator onto $L^2(-xt^{-1/3},\infty)$, and $\mathbb L^{\rm Ai}_{\sigma,x,t}$ is the integral operator acting on $L^2(\mathbb R)$ corresponding to the deformed Airy kernel
\begin{equation}
\label{def:Q4}
L^{\rm Ai}_{\sigma,x,t}(u,v)=\int_{\mathbb R}
\sigma(st^{-2/3}){\rm Ai}(u+s){\rm Ai}(v+s)ds.
\end{equation}
This representation of $Q_{\sigma_{\rm KPZ}}$ is known as the finite temperature analogue of the Tracy--Widom distribution, arising in the physics of free fermion models at finite temperature \cite{DLMS} and as the largest particle distribution in the MNS model \cite{MNS,Johansson, LiechtyWang}.
In this context, $x\to +\infty$ asymptotics of $Q_{\sigma_{\rm KPZ}}(x,t)$ describe the large gap probability for the largest free fermion. 

\medskip

For a large class of functions $\sigma$, including all $\sigma$ satisfying Assumptions \ref{assumptions} below, we have
\begin{equation*}
		\partial_x^2 \log Q_\sigma\left(x,t\right) = -\frac{1}{t}\int_\mathbb R \phi_\sigma^2\left(r;x,t\right)d\sigma\left(r\right),
	\end{equation*}
where $\phi_\sigma$ satisfies the integro-differential Painlev\'e II equation
	\begin{equation}\label{integrodiffPII}
		\partial_x^2\phi_\sigma\left(z;x,t\right) = \left(z - \frac{x}t + \frac{2}t \int_{\mathbb R} \phi_\sigma^2\left(r;x,t\right) d\sigma\left(r\right) \right) \phi_\sigma\left(z;x,t\right),
	\end{equation}
see \cite{AmirCorwinQuastel} and \cite[Corollary 1.4]{CaClR2020}, and like in the case $\sigma=1_{(0,+\infty)}$, we have that 
$u_\sigma$ defined by 
\begin{equation}
\label{def:u}
u_{\sigma}(x,t):=\partial_x^2\log Q_{\sigma}(x,t)+\frac{x}{2t}
\end{equation}
solves the KdV equation \eqref{def:KdV}.
This result was proved in \cite[Theorem 1.3]{CaClR2020}, but the relation between Fredholm determinants of a general class of operators and the KdV equation was already understood much earlier by P\"oppe \cite{Poppe}; for other similar relations between Fredholm determinants and integrable systems (in particular, the KP equation) see also  \cite{PoppeSattinger} and the recent works \cite{BaikLiuSilva,BothnerCafassoTarricone,Krajenbrink, KPZKPDoussal,LNR,QuastelRemenik}.

\medskip

In the case $\sigma=1_{(0,+\infty)}$, the $xt^{-1/3}\to \pm\infty$ asymptotics for the KdV solution are obtained easily from \eqref{TW1}--\eqref{def:u0}: we have 
\begin{align*}
&u_{1_{(0,+\infty)}}(x,t)=\frac{x}{2t}+\mathcal O\bigg( \frac{e^{-\frac{4}{3}(-xt^{-1/3})^{3/2}}}{t^{2/3}(xt^{-1/3})} \bigg),&&\mbox{as $xt^{-1/3}\to -\infty$,} \\
&u_{1_{(0,+\infty)}}(x,t)=\frac{1}{8x^2}+\bigO\bigg(\frac{1}{t^{2/3}(xt^{-1/3})^{5}}\bigg),&&\mbox{as $xt^{-1/3}\to +\infty$.}
\end{align*}

In \cite[Theorem 1.8]{CaClR2020}, it was proved, again for a class of functions $\sigma$ including those satisfying Assumptions \ref{assumptions} below, that the $t\to 0$ asymptotics for $u_\sigma(x,t)$ are as follows, see also Figure \ref{figure: phase diagram}.

\begin{itemize}
\item[(i)] For any $t_0>0$, there exist $M, c>0$ such that we have uniformly for $x\leq-M t^{1/3}$ and for $0<t<t_0$ that
\begin{equation}
\label{estimatei}
u_\sigma\left(x,t\right)=\frac{x}{2t}+\bigO\left({\rm e}^{-c\frac{|x|}{t^{1/3}}}\right).
\end{equation}
\item[(ii)] There exists $\epsilon>0$ such that for any $M>0$, we have uniformly for $|x|\leq M t^{1/3}$ and for $0<t<\epsilon$ that 
\begin{equation}
u_\sigma\left(x,t\right)=\frac{x}{2t}-t^{-2/3}y^2\left(-xt^{-1/3}\right)+{\mathcal O\left(1\right)},
\end{equation}
with $y$ the Hastings--McLeod solution of Painlev\'e II.
\item[(iii)] There exist $\epsilon, M>0$ such that for any $K>0$, we have uniformly for $M t^{1/3}\leq x\leq K$ and for $0<t<\epsilon$ that
\begin{equation}
\label{eq:thmasymptoticsiii}
u_\sigma\left(x,t\right)=v_\sigma\left(x\right)\left(1+\mathcal O\left(x^{-1}t^{1/3}\right)\right),
\end{equation}
where $v_\sigma$ is a function of $x>0$, independent of $t$, with asymptotics
\begin{equation}
\label{as:v}
v_\sigma(x)=\frac{1}{8x^2}+\frac 12\int_{-\infty}^{+\infty}\left(1_{(0,+\infty)}(r)-\sigma(r)\right)d r+\mathcal O\left(x^2\right),\qquad \mbox{as $x\to 0_+$.}
\end{equation}
Moreover, $v_\sigma$ satisfies an integro-differential version of the Painlev\'e V equation \cite[Theorem 1.12]{CaClR2020}.
\end{itemize}

It follows from these asymptotics that the KdV solutions $u_\sigma$ are unbounded for all $t>0$ and have ill-defined initial data at $t=0$; in particular their direct and inverse scattering theory is non-standard (see also \cite{DubrovinMinakov,ItsSukhanov}).

The analogous result for $Q_\sigma(x,t)$ \cite[Theorem 1.14]{CaClR2020} is as follows.
\begin{itemize}
\item[(i)] For any $t_0>0$, there exist $M, c>0$ such that we have uniformly for $x\leq -M t^{1/3}$ and for $0<t<t_0$ that
\begin{equation}
\log Q_\sigma\left(x,t\right)=\mathcal O\left(e^{-c\frac{|x|}{t^{1/3}}}\right).
\end{equation}
\item[(ii)] There exists $\epsilon>0$ such that for any $M>0$, we have uniformly for $|x|\leq M t^{1/3}$ and for $0<t<\epsilon$ that
\begin{equation}
\log Q_\sigma\left(x,t\right)=\log F_{\rm TW}\left(-xt^{-1/3}\right)+\mathcal O\left(t^{1/3}\right),
\end{equation}
where $F_{\rm TW}$ is the Tracy--Widom distribution.
\item[(iii)] There exist $\epsilon, M>0$ such that for any $K>0$, we have uniformly for $M t^{1/3}\leq x\leq K$ and for $0<t<\epsilon$ that
\begin{equation}
\label{logQexpansion}
\log Q_\sigma\left(x,t\right)=-\frac{x^3}{12t}-\frac{1}{8}\log (xt^{-1/3}) +\frac{\log 2}{24} +\zeta'(-1)
+\int_0^x (x-\xi) \left(v_\sigma(\xi)-\frac{1}{8\xi^2}\right)d\xi+\mathcal O(x^{-1}t^{1/3}),
\end{equation}
where $\zeta$ is the Riemann zeta function.\end{itemize}
It is important to note that these asymptotics are uniform for $x$ large and negative, but not for $x$ large and positive. 
The above expansion \eqref{logQexpansion} should be compared with the tail expansion \eqref{eq:TWtail} for the Tracy--Widom distribution.

\medskip

Besides for $\sigma=1_{(0,+\infty)}$, precise asymptotics for $Q_\sigma(x,t)$ as $x\to +\infty$ are only known for  $\sigma=\sigma_{\rm KPZ}$ \cite[Theorem 1.1]{CaCl2019}; they hold for $t>0$ fixed and also as $t\to 0$ slow enough such that $t\geq \delta/x$, where $\delta>0$ is fixed (see also Figure \ref{figure: phase diagram}):\footnote{The parameters $s$ and $T$ of \cite{CaCl2019} correspond respectively to $xt^{-1/3}$ and $t^{-2}$ in this paper.}
\begin{equation}\label{eq:Qlowertail}
\log Q_{\sigma_{\rm KPZ}}(x,t)=-\frac{1}{\pi^6t^4}F_1(\pi^2xt){-\frac{1}6}\sqrt{1+{\pi^2xt}} 
+\mathcal O(\log^2 (xt^{-1/3}))+\mathcal O(t^{-2/3}),
\end{equation}
where
\begin{equation}\label{def:phi}
F_1(y) := \frac{4}{15}(1+y)^{5/2}-\frac{4}{15}-\frac{2}{3}y-\frac{1}{2}y^2.
\end{equation}
Observe that $F_1(y)\sim\frac{y^3}{12}$ as $y\to 0$, and  $F_1(y)\sim \frac{4}{15}y^{5/2}$ as $y\to\infty$.

\medskip

Our main objective is to derive $x\to +\infty$ asymptotics for $u_\sigma(x,t)$ and $Q_\sigma(x,t)$ that are uniform for $t$ sufficiently small (also for $t\leq\delta/x$), for a rather broad class of functions $\sigma$, in order to complete the description of the initial data for the KdV solutions $u_\sigma(x,t)$, and to describe the cross-over between the small $t$ asymptotics \eqref{logQexpansion} and the tail expasion \eqref{eq:Qlowertail} of $Q_\sigma(x,t)$. In the particular case $\sigma=\sigma_{\rm KPZ}$, the asymptotics for  $Q_\sigma(x,t)$  will allow us to derive uniform lower tail asymptotics for the KPZ solution with narrow wedge initial data, see Section \ref{section:KPZ}.
Our results are more general and also more precise than \eqref{eq:Qlowertail}: Theorem \ref{theorem:main} below, even when specialized to $\sigma=\sigma_{\mathrm{KPZ}}$ and to the subsector $t \geq \delta/x$, $t \leq t_{0}$, contains new results, see Remark \ref{remark:sigmaKPZ}, part 3.

\medskip

Let us now formulate our assumptions on $\sigma$.
\begin{assumptions}\label{assumptions}
The function $\sigma:\mathbb R\to [0,1)$ is such that $F(r):=\frac{1}{1-\sigma(r)}$ satisfies the following conditions.
\begin{enumerate}
\item[{\rm 1.}] $F$ can be extended to an entire function in the complex plane.
\item[{\rm 2.}] $F'(x)\geq 0$ and $(\log F)''(x)\geq 0$ for all $x\in\mathbb R$.
\item[{\rm 3.}] There exist $c_-, c_-'>0$ such that $F(r)=1+c_-'e^{-c_-|r|}(1+o(1))$ as $r\to -\infty$.
\item[{\rm 4.}] There exist $c_+, c_+',\epsilon>0$ such that $F(r)=c_+'e^{c_+r}(1+\mathcal O(e^{-\epsilon r}))$ as $r\to +\infty$,
and for $z\in\mathbb C$, we have
\begin{align*}
\left|F(z)\right|=\mathcal O(e^{c_+\re z}), \qquad \mbox{as } \re z\to +\infty.
\end{align*}
\end{enumerate}
\end{assumptions}

\begin{remark}
\label{remark:logFdecay}
\begin{enumerate}
\item[{\rm 1.}] Functions $\sigma$ with discontinuities like $\sigma=1_{(0,+\infty)}$ are not admissible.

\item[{\rm 2.}] A class of admissible functions $F$ consists of the Laplace transforms of measures of the form $\mu=\delta_0+a_-\delta_{c_-}+\nu +a_+\delta_{c_+}$, where $c_\pm, a_\pm>0$, $c_-\leq c_+$, where $\delta_x$ is the Dirac distribution at $x$, and where $\nu$ is a finite Borel measure whose support is a subset of $[c_-,c_+)$ without atom at $c_-$.
Then, it is indeed straightforward to verify that 
\[F(z)=\int_\R e^{tz}d\mu(t)=1+a_-e^{c_-z}+\int_\R e^{tz}d\nu(t)+a_+e^{c_+z}\]
satisfies the assumptions, with $c_\pm'=a_\pm$ if $c_-<c_+$, and $c_+'=c_-'=a_++a_-$ if $c_+=c_-$.
Observe that these functions are completely monotonic, i.e.\ $F$ and all its derivatives are positive.
In particular, our prototype example $\sigma=\sigma_{\rm KPZ}$ is admissible and belongs to this class, with
\[F(z)=1+e^{z},\qquad \mu=\delta_0+\delta_1,\qquad c_\pm=c_\pm'=1,\qquad \epsilon\leq 1.\]

\item[{\rm 3.}]
Assumptions \ref{assumptions} imply that
\begin{align*}
(\log F)'(r) \leq \frac{\log F(2r) - \log F(r)}{r} = c_{+} + \bigO(e^{-\epsilon r}) \leq (\log F)'(2r), \qquad r \to + \infty,
\end{align*}
from which we conclude that 
\begin{equation}
\label{eq:vplusinfty}
(\log F)'(r) = c_{+} + \bigO(e^{-\frac{\epsilon}{2} r}),\qquad r \to + \infty.
\end{equation}
Similarly, since
\begin{align*}
(\log F)'(2r) \leq \frac{\log F(2r) - \log F(r)}{r} = \bigO(e^{-c_{-} |r|}) \leq (\log F)'(r), \qquad r \to - \infty,
\end{align*}
we deduce that 
\begin{equation}
\label{eq:vminusinfty}
(\log F)'(r) = \bigO(e^{-\frac{c_{-}}{2} |r|}),\qquad r \to - \infty.
\end{equation}

\item[{\rm 4.}]
{In \cite{CaClR2020} a broader class of functions $\sigma$ is considered. Namely, it is assumed in \cite[Assumptions~1.1]{CaClR2020} that $\sigma$ is non-decreasing and $\mathscr C^\infty$ on the real line (except possibly at a finite number of points), that $r\mapsto r^2\sigma(r)\in L^2((-\infty,0),dr)$, and in~\cite[equations~(1.21) and (1.22)]{CaClR2020} that $\sigma$ is exponentially close to $1_{(0,+\infty)}$ and $\sigma'$ vanishes at least quadratically at $\pm\infty$.
It is straightforward to check that if $\sigma$ satisfies Assumptions~\ref{assumptions}, then it also satisfies all the assumptions made in \cite{CaClR2020} which we just recalled.
In particular, the operator $\mathbb K_{\sigma,x,t}$ is trace-class \cite[Remark~1.2]{CaClR2020} and all the above mentioned results from \cite{CaClR2020} hold.}
The stronger Assumptions~\ref{assumptions} are needed here for technical reasons in the asymptotic analysis.
\end{enumerate}
\end{remark}

\medskip

\begin{theorem}\label{theorem:main}
Let $\sigma$ satisfy Assumptions \ref{assumptions}, let $Q_\sigma(x,t)$ be as in \eqref{def:Q1}--\eqref{def:Q2}, and let $u_\sigma(x,t)$ be as in \eqref{def:u}.
For any $t_0>0$, there exists $K>0$ such that we have the asymptotics 
\begin{align}
\label{thm:asu}
u_\sigma(x,t)&=
\frac x{2t}a_0\left(\frac{\pi^2}{c_+^2}xt\right)+\frac 1{2\sqrt{xt}}a_1\left(\frac{\pi^2}{c_+^2}xt\right)+\frac{t^{1/2}}{2x^{3/2}}a_2\left(\frac{\pi^2}{c_+^2}xt\right)
+\bigO(x^{-2}),
\\
\label{thm:asQ}
\log Q_\sigma(x,t)&=-\frac{c_+^6}{\pi ^6t^4}F_1\left(\frac{\pi^2}{c_+^2}xt\right)-\frac{c_+^3\log c'_+}{\pi^4t^2}F_2\left(\frac{\pi^2}{c_+^2}xt\right)+F_3\left(\frac{\pi^2}{c_+^2}xt\right)
-C\log t
+\mathcal O(1),
\end{align}
with error terms that are uniform for $x\geq K$ and $0<t\leq t_0$.
Here 
\begin{align}
\label{eq:thm-a0a1}
a_0(y)&=\frac{\left(\sqrt{y+1}-1\right)^2}y,\qquad
a_1(y)=-\frac{\log c_+'}{\pi}\sqrt{\frac y{1+y}},
\\
\label{eq:thm-a2}a_2(y) &= \frac{y^{3/2}}{\sqrt{y+1}(\sqrt{y+1}-1)^{2}}\bigg( \frac{1}{4\pi c_{+}}\frac{1-2\sqrt{y+1}}{y+1}\log^{2} c_{+}' - j_{\sigma} \bigg)
,\\\label{eq:thm-F1F2}F_1(y) &=\frac{4}{15}\left(1+y\right)^{5/2}-\frac{4 }{15}-\frac{2}{3}y-\frac{1}{2}y^2,
\qquad F_2(y) =\frac 23\left(1+y\right)^{3/2}-\frac 23-y =F_1'(y),
\\
\label{eq:thm-F3}
F_3(y) &=\frac{2c_+ j_\sigma}\pi\sqrt{1+y}-\frac{1}{48}\log(1+y)+\left(\frac{2c_+ j_\sigma}\pi+\frac{\log^2c'_+}{2\pi^2}+\frac{1}{24}\right)\log\left(\sqrt{1+y}-1\right),
\end{align}
and 
\[
C=\frac{2c_+ j_\sigma}\pi+\frac{\log^2c'_+}{2\pi^2},\qquad
j_\sigma = \frac 1{2\pi}\int_{-\infty}^{+\infty}\left[\log(1-\sigma(r))+(c_+r+\log c_+')1_{(0,+\infty)}(r)\right]dr.\]
\end{theorem}
\begin{remark}\label{remark:sigmaKPZ}
\begin{enumerate}
\item[{\rm 1.}]
The form of the above expansions was predicted
on a non-rigorous basis in \cite{KPZKPDoussal} by checking consistency with the cylindrical KdV equation, which is a reduction of the KP equation and equivalent to the KdV equation through a change of variables.
\item[{\rm 2.}] {
Using \eqref{eq:thm-F1F2}--\eqref{eq:thm-F3}, we obtain
\begin{multline*}
-\frac{c_+^6}{\pi^6t^4}F_1 \left(\frac{\pi^2}{c_+^2}Kt\right) - \frac{c_+^3\log c_+'}{\pi^4t^2}F_2  \left(\frac{\pi^2}{c_+^2}Kt\right) + F_3  \left(\frac{\pi^2}{c_+^2}Kt\right) \\
=  - \frac{K^3}{12t} + \left(\frac{2c_+ j_\sigma}\pi+\frac{\log^2c'_+}{2\pi^2}+\frac{1}{24}\right)\log t + \bigO(1), \qquad \mbox{as } t \to 0.
\end{multline*}
Substituting the above in \eqref{thm:asQ} yields $\log Q_\sigma(K,t) = -\frac{K^3}{12t}+\frac{1}{24}\log t +\mathcal O(1)$ as $t \to 0$, which is consistent with \eqref{logQexpansion}. 
}
\item[{\rm 3.}]
Recall the relation $\partial_x^2\log Q_\sigma(x,t)+\tfrac x{2t} = u_\sigma(x,t)$, see \eqref{def:u}.
Even though we have no control on the second derivative of the error term in the asymptotic expansion of $\log Q_\sigma$, it is instructive to formally substitute the asymptotics for $\log Q_\sigma$ and $u_\sigma$ as $x\to\infty$ with $t>0$ fixed into this relation (note that $\partial_x =\frac yx \partial_y$ with $y=\pi^2xt/c_+^2$), and to verify order by order whether the result is consistent.
We then observe that the leading and sub-leading order in the relation are consistent because of the identities
\[
-\left(\frac yx\right)^2 \frac{c_+^6}{\pi^6t^4}F_1''(y)+\frac x{2t} = \frac x{2t} a_0(y),\qquad
-\left(\frac yx\right)^2 \frac{c_+^3\log c_+'}{\pi^4t^2}F_2''(y) = \frac 1{2\sqrt{xt}} a_1(y).
\]
For the next order, we should have
\[
\left(\frac yx\right)^2F_3''(y) = \frac {t^{1/2}}{2x^{3/2}} a_2(y)+\mathcal O(x^{-2}),
\]
and this is indeed the case, by the identity
\[
\left(\frac yx\right)^2F_3''(y) - \frac {t^{1/2}}{2x^{3/2}} a_2(y) = \frac{y^2 \left(-3 y+2 \sqrt{y+1}-3\right)}{96 x^2 (y+1)^{5/2} \left(\sqrt{y+1}-1\right)^2},
\]
where we note that the right-hand side is a uniformly bounded function of $y\in[0,+\infty)$ times $x^{-2}$.
\item[{\rm 4.}]
For $xt$ bounded, the term $\frac{t^{1/2}}{2x^{3/2}}a_2\left(\frac{\pi^2}{c_+^2}xt\right)$ in the expansion \eqref{thm:asu} of $u_\sigma$ can simply be absorbed by the error term. However, as $xt\to\infty$, this term is of order $\bigO(\frac{(xt)^{1/2}}{x^{2}})$.
\item[{\rm 5.}]
For $\sigma=\sigma_{\mathrm{KPZ}}$, we have $c_+=c'_+=1$, $j_{\sigma_{\rm KPZ}}=-\frac\pi{12}$, $C=-\frac 16$, and the asymptotics \eqref{thm:asQ} become
\begin{align}
\log Q_{\sigma_{\rm KPZ}}(x,t) = & -\frac{F_1(\pi^2xt)}{\pi^6t^4}-\frac{\sqrt{1+\pi^2xt}}{6} - \frac{\log(1+\pi^{2} x t)}{48} \nonumber \\
& - \frac{\log \big( \sqrt{1+\pi^{2}xt}-1 \big)}{8} + \frac{\log t}{6} + \bigO(1), \label{KPZasymp}
\end{align}
uniformly for $x \geq K$ and $0<t \leq t_{0}$. This result  is more precise and valid for a wider sector of the $(x,t)$-half-plane than \eqref{eq:Qlowertail}.
Up to the coefficient of the logarithmic term in $t$, this expansion was predicted  in \cite[equation (50)]{KPZKPDoussal}. We thus confirm the result of \cite{KPZKPDoussal} rigorously with uniform error terms, and refine it by computing the logarithmic term.
\end{enumerate}
\end{remark}

\begin{figure}[htbp]
\centering
\hspace{1cm}
\begin{tikzpicture}[scale=1]

\draw[->] (-5,0) -- (8,0) node[right] {$x$};
\draw[-] (0,0) -- (0,0.6);
\draw[->](0,1.5) -- (0,4) node[above] {$t$};
\draw[domain=0:2.1, smooth, variable=\x, dashed] plot ({\x}, {.2*\x*\x*\x}) node[above] {${\scriptstyle x=Mt^{1/3}}$};
\draw[domain=0:2.65, smooth, variable=\x, dashed] plot ({-\x}, {.2*\x*\x*\x}) node[above] {${\scriptstyle x=-Mt^{1/3}}$};
\draw[-, dashed] (-2.54,3.3) -- (-5,3.3) node[left] {${\scriptstyle t=t_0}$};
\draw[-, dashed] (5,3.7) -- (5,0) node[below] {${\scriptstyle x=K}$};
\draw[-, dashed] (-1.98,1.55) -- (5,1.55) node[right] {${\scriptstyle t=\epsilon}$};

\node at (-4,1.6) {\Large{$\frac{x}{2t}$}};
\node at (-4,2.2) {singular};
\node at (0,.9) {$\frac{x}{2t}-\frac 1{t^{2/3}}y^2\left(-\frac x{t^{1/3}}\right)$};
\node at (0,1.35) {Painlev\'e II};
\node at (3.3,.5) {$v_\sigma\left(x\right)$};
\node at (3.4,1.1) {$\begin{array}{c}\mbox{integro-differential}\\\mbox{Painlev\'e V}\end{array}$};

\draw[domain=5:8, smooth, variable=\x, dashed] plot ({\x}, {6/\x}) node[above] {${\scriptstyle xt=\delta}$};
\draw[-, dashed] (8,3.3) -- (5,3.3) node[left] {${\scriptstyle t=t_0}$};
\node at (6.5,2.2) {\Large{$\frac{x}{2t}$}};
\node at (6.5,.4) {\Large{$\frac{\pi^2x^2}{8c_+^2}$}};

\end{tikzpicture}
\caption{Phase diagram showing the different leading asymptotics for $u_\sigma(x,t)$ when $t$ is small or $x$ is large (or both, uniformly in the indicated regions); sub-leading corrections are omitted and can be found in Theorem \ref{theorem:main} and \cite[Theorem 1.8]{CaClR2020}.
When $x$ is large and positive, the leading asymptotics are given by $\frac {x}{2t}a_0(\pi^2 x t/c_+^2)$.
Since $a_0(y)\sim y/4$ as $y\to 0$ and $a_0(y)\sim 1$ as $y\to+\infty$, this asymptotic behavior exhibits a phase transition in the critical regime when $x\to+\infty$ and $t\to 0_+$ in such a way that $xt$ is bounded away from zero and infinity.
Outside of this critical regime we have the two asymptotic behaviors indicated in the diagram.
It is worth mentioning that the leading behavior $x/(2t)$, emerging when $x\to\pm\infty$, is itself a KdV solution.}
\label{figure: phase diagram}
\end{figure}
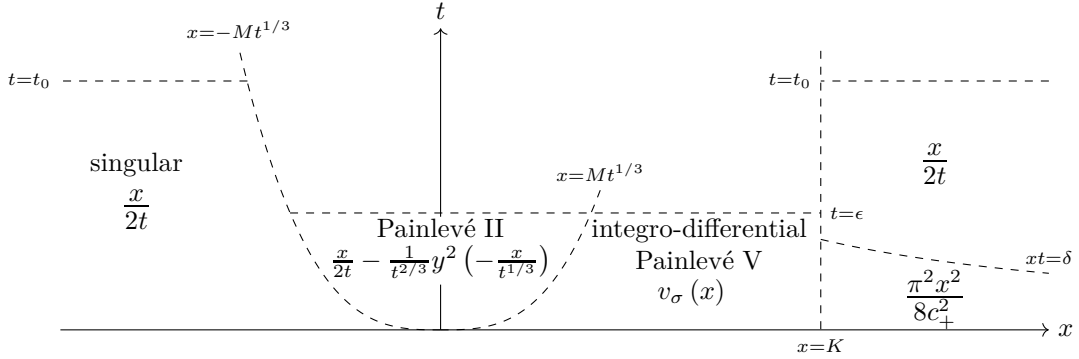

The starting point of our analysis is a characterization of $\log Q_\sigma(x,t)$ in terms of a Riemann--Hilbert (RH) problem, obtained in \cite{CaClR2020} by relying on the theory of integrable operators developed by Its, Izergin, Korepin, and Slavnov \cite{IIKS}: from \cite[(3.14) and the second equation after (5.3)]{CaClR2020}, we know that 
\begin{align}\label{xdiffid}
& \partial_{x} \log Q_{\sigma}(x,t) = p(x,t)-\frac{x^2}{4t}, 
\end{align}
and that
\begin{equation}\label{eq:upq}
u_\sigma(x,t)=\partial_x p(x,t)=-2q(x,t)-p(x,t)^2,
\end{equation}
where $p(x,t), q(x,t)$, which depend also on $\sigma$, are defined in terms of the unique solution $\Psi(z;x,t)$ to the following RH problem.
\subsubsection*{RH problem for $\Psi$}
\begin{itemize}
\item[(a)] $\Psi : \mathbb{C}\setminus \mathbb{R} \to \mathbb{C}^{2\times 2}$ is analytic. 
\item[(b)] The boundary values of $\Psi$ are continuous on $\mathbb{R}$ and are related by
\begin{align*}
\Psi_{+}(z;x,t) = \Psi_{-}(z;x,t) \begin{pmatrix}
1 & 1-\sigma(z) \\
0 & 1
\end{pmatrix}, \qquad z \in \mathbb{R}.
\end{align*}
The subscript $+$ (resp. $-$) indicates the boundary value from above (resp. below) the real axis.
\item[(c)] As $z \to \infty$, we have
\begin{align*}
\Psi(z;x,t) = \bigg( I + \frac{\Psi_{1}(x,t)}{z} + \bigO(z^{-2}) \bigg) z^{\frac{1}{4}\sigma_{3}}A^{-1}e^{(-\frac{2}{3}tz^{\frac 32}+xz^{\frac 12})\sigma_{3}} \begin{cases}
I, & |\arg z | < \pi -\delta, \\
\begin{pmatrix}
1 & 0 \\
\mp 1 & 1
\end{pmatrix}, & \pi - \delta < \pm \arg z < \pi,
\end{cases}
\end{align*}
for any $\delta \in (0,\frac{\pi}{2})$, where {we take the principal branches of $z^{\frac 14\sigma_3}${, $z^{\frac{3}{2}}$} and $z^{\frac 12}$, analytic in $\C\setminus[0,+\infty)$ and positive for $z>0$, and}
\begin{align}
\label{eq:A}
\Psi_1(x,t) = \begin{pmatrix}
q(x,t) & -i r(x,t) \\
i p(x,t) & -q(x,t)
\end{pmatrix}, \qquad \sigma_{3} = \begin{pmatrix}
1 & 0 \\
0 & -1
\end{pmatrix}, \qquad A = \frac{1}{\sqrt{2}} \begin{pmatrix}
1 & i \\ i & 1
\end{pmatrix},
\end{align}
\end{itemize}
Our main task will be to derive asymptotics for the RH solution $\Psi(z;x,t)$ as $x\to +\infty$, uniform in $t\leq t_0$ for any $t_0>0$ and uniform for large $z$; this is achieved via the Deift--Zhou nonlinear steepest descent method \cite{DZ}.
The large $z$ asymptotics in condition (c) of the RH problem for $\Psi$ will allow us to derive asymptotics for $\log Q_\sigma(x,t)$ and $u_\sigma(x,t)$ using the identities \eqref{xdiffid} and \eqref{eq:upq}.

\medskip

In order to obtain asymptotics for $\Psi$, we will need to consider two asymptotic regimes separately.
First, we will consider $K\leq x\leq \delta/t$ for a sufficiently large $K>0$ and a sufficiently small $\delta>0$. The asymptotic RH analysis in this regime requires novel ideas, in particular the manner in which the jump is decomposed into an
explicit part taken care of by the $g$-function and a $\sigma$-dependent
part taken care of by the global parametrix, and is the most challenging part of the asymptotic analysis; 
this is the subject of Section \ref{section:3}.
Secondly, we will consider the regime $x\geq \delta/t$, and in this part of the analysis $\delta>0$ will be arbitrary. Here, the asymptotic analysis is a generalization of the analysis done in \cite{CaCl2019} for $\sigma=\sigma_{\rm KPZ}$, which we also refine in order to cover general $\sigma$ and in order to compute all growing terms in the asymptotic expansions; this is done in Section \ref{section:4}. Finally, the results for the two regimes are combined in Section \ref{finalsec}, where the proof of Theorem \ref{theorem:main} is completed.

\section{Lower tail of the KPZ solution with narrow wedge initial data}\label{section:KPZ}

The KPZ equation was introduced by Kardar, Parisi, and Zhang in 1986 \cite{KPZ} as a universal model for interface growth, and is believed to model a variety of real-life phenomena such as bacterial colony growth, see e.g.\ the review article \cite{HHT}.
It is a stochastic PDE, given by
\begin{equation}
\label{eq:KPZ}
\partial_T\mathcal H(T,X)=\frac{1}{2}\partial_X^2\mathcal H(T,X)+\frac{1}{2}(\partial_X\mathcal H(T,X))^2+\xi(T,X),
\end{equation}
where the first term in the right-hand side stands for relaxation, the second term for non-linear slope-dependent growth, and the random forcing term $\xi(T,X)$ is space-time white noise. Mathematical sense can be given to this equation through the Hopf--Cole transformation $\mathcal H(T,X)=\log Z(T,X)$, which transforms the KPZ equation to the stochastic heat equation 
\[\partial_T Z(T,X)=\frac{1}{2}\partial_X^2Z(T,X)+Z(T,X)\xi(T,X),\]
a stochastic PDE for which the solution theory is classical. 
A more general notion of solution was constructed by Hairer \cite{Hairer}.
The narrow wedge solution of KPZ is one of the physically relevant solutions, characterized by Dirac initial data of the Hopf--Cole transform
$Z(0,X)=\delta_{X=0}.$ With this initial condition, typical KPZ solutions behave roughly like a narrowing parabola $\frac{X^2}{2T}$ as $T\to 0$.
We refer to \cite{Corwin} for more details about this solution and to \cite{Quastel} for a general discussion on Hopf--Cole solutions of KPZ.

\medskip

In what follows, $\mathcal H(T,X)$ denotes the Hopf--Cole solution of KPZ equation with narrow wedge initial condition. Amir, Corwin, and Quastel \cite{AmirCorwinQuastel} proved in 2010 that the probability distribution of $\mathcal H(T,X)$ can be characterized by the Fredholm determinant $Q_{\sigma_{\rm KPZ}}$. Such a characterization was obtained independently, but without rigorous justification of certain steps, by Sasamoto and Spohn \cite{SasamotoSpohn}, and predicted also by Dotsenko \cite{Dotsenko} and Calabrese, Le Doussal, and Rosso \cite{CLDR}.
Let us now describe this characterization.

\medskip
The probability distribution of $\mathcal H(T,X)-\frac{X^2}{2T}$ is independent of $X$  \cite[Theorem 1.1]{Corwin}), hence we can restrict ourselves to studying the probability distribution of $\mathcal H(T,0)$, or to the distribution of the suitably re-scaled random variable
\begin{equation}\label{Upsilon}\Upsilon_T=\frac{\mathcal H(2T,0)+\frac{T}{12}}{T^{1/3}}.\end{equation}
Borodin and Gorin \cite[Theorem 2.1]{BorodinGorin} proved that the results from \cite{AmirCorwinQuastel} imply the identity
\begin{equation}\label{eq:BorodinGorin}
\mathbb E\left[{\rm e}^{-{\rm e}^{T^{1/3}(\Upsilon_T+s)}}\right]=Q_{\sigma_{\rm KPZ}}\left(x=sT^{-1/6},t=T^{-1/2}\right),\end{equation} 
or, in other words, $Q_{\sigma_{\mathrm{KPZ}}}(x,t)$ is a Laplace-type transform of the solution to the stochastic heat equation.

\medskip

Corwin and Ghosal \cite[Theorem 1.1]{CorwinGhosal} obtained robust upper and lower bounds for the lower tail of the probability distribution of $\Upsilon_{T}$: given $\epsilon, \delta\in(0,\frac{1}{3})$ and $T_0>0$, there exist $S = S(\epsilon, \delta, T_0)$, $C=C(T_0)>0$, $K_1=K_1(\epsilon,\delta, T_0)>0$, and $K_2=K_2(T_0)>0$ such that for all $s\geq S$ and $T\geq T_0$,
\begin{align}\label{eq:MainUpperBound}
\mathbb{P}\big(\Upsilon_T\leq -s\big)\leq  e^{-\frac{4(1-C\epsilon) }{15\pi}T^{\frac{1}{3}} s^{5/2}} + e^{-K_1s^{3-\delta}-\epsilon T^{1/3}s} + e^{-\frac{(1-C\epsilon)}{12}s^3},
\end{align}
and
\begin{equation}\label{eq:MainLowerBound}
\mathbb{P}\big(\Upsilon_T\leq -s\big) \geq e^{-\frac{4(1+C\epsilon)}{15\pi}T^{\frac{1}{3}}s^{5/2}}+  e^{-K_2 s^{3}}.
\end{equation}
For different values of $s$ and $T$, different terms at the right can dominate the others.
In \cite{CaCl2019}, these bounds were refined in the region where $M^{-1}\leq T\leq Ms^{3/2}$, for any $M>0$: then for any $\epsilon>0$,
\[B(s+(3+\epsilon)T^{-1/3}\log s,T)\leq \log \mathbb P(\Upsilon_T<-s)\leq B(s,T),\]
with
\begin{equation}\label{eq:KPZlowertailCC}
B(s,T)=-\frac{T^{2}}{\pi^6}F_1\bigg(\frac{\pi^2s}{T^{2/3}}\bigg) {-\frac{1}6}\sqrt{1+\frac{\pi^2s}{T^{2/3}}}
+\mathcal O(\log^2 s)+\mathcal O(T^{1/3}),\qquad\mbox{as $s\to +\infty$,}
\end{equation}
where $F_1$ is defined in \eqref{eq:thm-F1F2}.

We extend these estimates uniformly to $s\geq s_0,T\geq T_0$ for any $s_0,T_0>0$. To this end, we first obtain the asymptotic properties of $\log Q_{\sigma_{\rm KPZ}}$ in this regime, see also Figure \ref{figure: Q phase diagram}.

\begin{theorem}
\label{thm:KPZ}
We have, uniformly for $s\geq s_0,T\geq T_0$,
\begin{equation}
\label{eq:asympQsigmaKPZG}
\log Q_{\sigma_{\rm KPZ}}(x=sT^{-1/6},t=T^{-1/2}) = -G(s,T) +\bigO(1),
\end{equation}
where
\begin{equation}
\label{eq:functionG}
G(s,T) := \frac{T^2}{\pi^6}F_1\bigg(\frac{\pi^2s}{T^{2/3}}\bigg)+\frac 16\sqrt{1+\frac{\pi^2s}{T^{2/3}}}
+\frac 1{48}\log\left(1+\frac{\pi^2s}{T^{2/3}}\right)+\frac 18\log\left(\sqrt{1+\frac{\pi^2s}{T^{2/3}}}-1\right)+\frac 1{12}\log T,
\end{equation}
with $F_1$ defined in \eqref{eq:thm-F1F2}.
\end{theorem}
\begin{proof}
The asymptotics \eqref{KPZasymp} imply
\begin{equation}
\label{eq:generalasympQ}
\log Q_{\sigma_{\rm KPZ}}(x=sT^{-1/6},t=T^{-1/2}) = -G(s,T) +\bigO(1),
\end{equation}
uniformly for $s\geq K T^{1/6}, T\geq T_0$.
The remaining region $s_0\leq s \leq KT^{1/6}$, $T\geq T_0$ is covered precisely by \cite[Theorem 1.14 (iii)]{CaClR2020}, see also \eqref{logQexpansion}; in terms of the variables $s,T$ we obtain
\begin{equation}
\log Q_{\sigma_{\rm KPZ}}(x=sT^{-1/6},t=T^{-1/2}) = -\frac {s^3}{12}-\frac 18\log s+\bigO(1),
\end{equation}
uniformly in $s_0\leq s\leq KT^{1/6}$, $T\geq T_0$.
We complete the proof by observing that
\begin{equation}
G(s,T) = \frac {s^3}{12}+\frac 18\log s+\bigO(1)
\end{equation}
uniformly in the same regime.
\end{proof}

\begin{figure}[htbp]
\centering
\hspace{1cm}
\begin{tikzpicture}[scale=1]

\draw[->] (0,0) -- (4,0) node[right] {$s$};
\draw[->](0,0) -- (0,4) node[left] {$T$};
\draw[domain=0.91:1.56, smooth, variable=\w, dashed] plot (1.2*\w*\w,\w*\w*\w) node[right] {$\begin{matrix} {\scriptstyle s=yT^{2/3}:}\\ \frac{T^2F_1(\pi^{2}y)}{\pi^6}\end{matrix}$};
\node at (2,3.5) {$\begin{matrix}{\scriptstyle s=o(T^{2/3}):}\\ \frac{s^3}{12}\end{matrix}$};
\node at (3.5,1.5) {$\begin{matrix}{\scriptstyle T=o(s^{3/2}):}\\ \frac{4s^{5/2}T^{1/3}}{15\pi}\end{matrix}$};

\draw[dashed] (1,3.8)--(1,0) node[below]{$s_0$};
\draw[dashed] (3.8,.6)--(0,.6) node[left]{$T_0$};

\end{tikzpicture}
\caption{Phase diagram showing the different leading asymptotics for $-\log Q_{\sigma_{\rm KPZ}}(sT^{-1/6},T^{-1/2})$ when $s,T\to+\infty$, uniform for $s\geq s_0,T\geq T_0$. This is given by the leading behavior for large $s,T$ of the function $G(s,T)$ in \eqref{eq:functionG}. A phase transition occurs in the critical regime when $s,T\to+\infty$ such that $sT^{-2/3}$ remains bounded away from zero and infinity.}
\label{figure: Q phase diagram}
\end{figure}
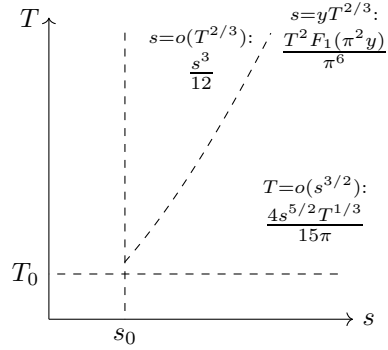

The following result is a direct consequence of Theorem \ref{thm:KPZ}, whose proof is almost identical to that in \cite[Section 3.1]{AmirCorwinQuastel} or \cite[Corollary 1.2]{CaCl2019}.

\begin{corollary}\label{corollary:KPZ}
Let $\Upsilon_T$ be as defined in \eqref{Upsilon} and $G(s,T)$ be as defined in \eqref{eq:functionG}.
\begin{enumerate}
\item[{\rm 1.}] For any $s_0,T_0>0$, there exists a real constant $D_+=D_+(s_0,T_0)$ such that the inequality
\begin{equation}
\label{eq:upperboundp}
\log \mathbb P(\Upsilon_T<-s)\leq p-G(s+T^{-1/3}\log p,T)+D_+
\end{equation}
holds for all $s\geq s_0,T\geq T_0$, and $p\geq 1$.
\item[{\rm 2.}] For any $\epsilon>0$ and for any $s_0,T_0>0$ sufficiently large, there exists a real constant $D_-=D_-(s_0,T_0)$ such that the inequality
\begin{equation}
\label{eq:lowerboundp}
\log \mathbb P(\Upsilon_T<-s)\geq -G(s+T^{-1/3}\log(s^{3+\epsilon}+T^\epsilon),T)+D_-
\end{equation}
holds for all $s\geq s_0,T\geq T_0$.
\end{enumerate}
\end{corollary}
\begin{proof}
Let $\Omega(s,T):=e^{-e^{T^{1/3}(\Upsilon_T+s)}}$ so that $\mathbb E[\Omega(s,T)] = Q_{\sigma_{\rm KPZ}}(sT^{-1/6},T^{-1/2})$ by \eqref{eq:BorodinGorin}. Moreover, for all $p>0$ we have
\[
\Omega(s,T)^p = e^{-pe^{T^{1/3}(\Upsilon_T+s)}} = e^{-e^{T^{1/3}(\Upsilon_T+s+T^{-1/3}\log p)}} = \Omega(s+T^{-1/3}\log p,T).
\]
\begin{enumerate}
\item Markov's inequality implies (for any $p>0$)
\[\mathbb P(\Upsilon_T\leq -s)=\mathbb P(\Omega(s,T)^p\geq e^{-p})\leq e^p \,\mathbb E[\Omega(s,T)^p]=e^p\, Q_{\sigma_{\rm KPZ}}\left((s+T^{-1/3}\log p)T^{-1/6},T^{-1/2}\right)\]
and taking logarithms it follows that 
\[
\log \mathbb P(\Upsilon_T\leq -s)\leq p+\log Q_{\sigma_{\rm KPZ}}\left((s+T^{-1/3}\log p)T^{-1/6},T^{-1/2}\right).
\]
Then it suffices to use \eqref{eq:asympQsigmaKPZG}.

\item For any $q>0$ we have
\[
\mathbb E\left[\Omega(s,T)^q\right]
=\mathbb E\left[1_{\{\Upsilon_T+s<0\}}\Omega(s,T)^q\right]+\mathbb E\left[1_{\{\Upsilon_T+s\geq 0\}}\Omega(s,T)^q\right]
\leq \mathbb P(\Upsilon_T+s<0)+e^{-q}
\]
where in the last step we used that $\Omega(s,T)\leq 1$ and $1_{\{\Upsilon_T+s\geq 0\}}\Omega(s,T)^q\leq e^{-q}$.
Therefore
\[
\mathbb P(\Upsilon_T+s<0) \geq Q_{\sigma_{\rm KPZ}}\left((s+T^{-1/3}\log q)T^{-1/6},T^{-1/2}\right)-e^{-q}.
\]
On the other hand, the uniform asymptotics \eqref{eq:asympQsigmaKPZG} imply that if $q = s^{3+\epsilon}+T^\epsilon$ (for any $\epsilon>0$) then $e^{-q}$ is negligible with respect to the other term.
Therefore, for this choice of $q$ and $s_0,T_0$ sufficiently large, we have
\[
Q_{\sigma_{\rm KPZ}}\left((s+T^{-1/3}\log q)T^{-1/6},T^{-1/2}\right)-e^{-q}\geq \frac 12 Q_{\sigma_{\rm KPZ}}\left((s+T^{-1/3}\log q)T^{-1/6},T^{-1/2}\right)
\]
for all $s\geq s_0,T\geq T_0$.
The proof is completed by combining the last two inequalities, taking logarithms, and using \eqref{eq:asympQsigmaKPZG}.
\end{enumerate}
\end{proof}

\begin{remark}
The bounds of Corollary \ref{corollary:KPZ} imply the following estimates in various regimes.
\begin{itemize}
\item[{\rm 1.}]
When $s=yT^{2/3}$ and $T$ sufficiently large with $y$ in any compact subset of $(0,+\infty)$, we set $p=1$ in the above result, and use
the expressions \eqref{eq:thm-F1F2} for $F_1$ and \eqref{eq:functionG} for $G$. After a straightforward computation, we then obtain that for any compact $K\subset (0,+\infty)$, there exist $s_0,T_0,C>0$ such that for $s\geq s_0, T\geq T_0$ and $y=sT^{-2/3}\in K$, we have
\[
-CT\log T
\leq \log\mathbb P(\Upsilon_T<-yT^{2/3})+\frac{T^2F_1(\pi^2 y)}{\pi^6}\leq \frac{1}{12}\log T+C.
\]
This confirms the large deviation function which was predicted using different approaches in \cite{SasorovMeersonProlhac,KLDP18,CGKLDT,KLD2} and  rigorously proved in \cite{Tsai}; our approach provides an alternative proof via the RH method and allows to quantify the error term.
\item[{\rm 2.}]
When  $s\to+\infty$ and $sT^{-2/3}\to +\infty$ (in particular when $s\to +\infty$ with $T>0$ fixed), we obtain a sharper bound by choosing $p = s^{3/2}$. We then get
\begin{align*}
\log \mathbb P(\Upsilon_{T}<-s)\leq & -\frac{4 s^{\frac 52}T^{\frac 13}}{15 \pi }+\frac{s^2 T^{\frac 23}}{2 \pi ^2}-\frac{2s^{\frac 32}T}{3\pi^3}
+\frac{2sT^{\frac 43}}{3\pi^4}
+\frac{s^{\frac 12}T^{\frac 53}}{2\pi^5}
 \\
 & +\frac{4T^2}{15\pi^6}
-\frac{s^{\frac 32}\log s}{\pi} + \bigO(T^{7/3}s^{-1/2} + s^{3/2} + T^{1/3}s \log s).
\end{align*}
Combining this with the lower bound \eqref{eq:lowerboundp}, we obtain that for any $T_0>0$, there exist $C,s_0,y_0>0$ such that for  $s\geq s_0, T\geq T_0, sT^{-2/3}\geq y_0$, we have the upper and lower bounds
\begin{align*}
& -C(T^{7/3}s^{-1/2} + T^{1/3}s \log s) \leq \bigg( \log \mathbb P(\Upsilon_{T}<-s)+\frac {4s^{\frac 52}T^{\frac 13}}{15\pi}-\frac{s^2T^{\frac 23}}{2\pi^2}+\frac{2s^{\frac 32}T}{3\pi^3} \\
& \qquad -
\frac{2sT^{\frac 43}}{3\pi^4} + \frac{s^{\frac 12}T^{\frac 53}}{2\pi^5}-\frac{4T^2}{15\pi^6}
+\frac {s^{\frac 32}\log s}{\pi} \bigg) \leq C(T^{7/3}s^{-1/2} + s^{3/2} + T^{1/3}s \log s).
\end{align*}
\item[{\rm 3.}]
When $T$ is sufficiently large and $sT^{-2/3}$ is sufficiently small, we take $p=1$, and obtain the following bounds. For any $\epsilon>0$, there exist $s_0, T_0, y_0, C$ such that for $s\geq s_0, T\geq T_0$, $sT^{-2/3}\leq y_0$, we have the bounds
\begin{align*}
-\frac{s^2\log\left(s^{3+\epsilon}+T^\epsilon\right)}{4T^{\frac 13}}-Cs^4T^{-\frac 23}-CsT^{-\frac 23}\log^2T-C&\leq \log \mathbb P(\Upsilon_{T}<-s)+\frac{s^3}{12}+\frac{1}{8}\log s
\\
&\leq C+Cs^4T^{-\frac 23}.
\end{align*}
\end{itemize}
\end{remark}
\section{Asymptotic analysis for $\Psi$ for $K\leq x\leq \delta/t$}\label{section:3}
\subsection{Rescaling}\label{subsection:rescaling section 3}
In this section, $x$ is large while $xt$ remains bounded. Since the exponential functions in the large $z$ asymptotics of $\Psi(z;x,t)$ are $e^{\pm (-\frac{2}{3}tz^{\frac{3}{2}}+xz^{\frac{1}{2}})}$, it is not convenient to directly work with $\Psi$. The change of variables $z=x^2\zeta$ is suited for this section as it transforms these exponential functions into $e^{\pm x^{2}(-\frac{2}{3}xt\zeta^{\frac 32}+\zeta^{\frac{1}{2}})}$. Therefore, we define
\begin{align*}
\widehat\Psi(\zeta;x,t)= x^{-\frac{\sigma_{3}}{2}}\Psi(x^{2}\zeta;x,t).
\end{align*}
The RH problem for $\Psi$ then transforms into the follow RH problem for $\widehat\Psi$.
\subsubsection*{RH problem for $\widehat{\Psi}$}
\begin{itemize}
\item[(a)] $\widehat{\Psi} : \mathbb{C}\setminus  \mathbb{R} \to \mathbb{C}^{2 \times 2}$ is analytic.
\item[(b)] $\widehat{\Psi}$ has the jumps
\begin{align*}
& \widehat{\Psi}_{+}(\zeta;x,t) = \widehat{\Psi}_{-}(\zeta;x,t)\begin{pmatrix}
1 & \frac{1}{F(x^{2}\zeta)} \\
0 & 1
\end{pmatrix}, & & \zeta \in \mathbb{R},
\end{align*}
where we recall that $F=\frac{1}{1-\sigma}$.
\item[(c)] As $\zeta \to \infty$, we have
\begin{multline*}
\widehat{\Psi}(\zeta;x,t) = \bigg( I + \frac{1}{\zeta}\widehat{\Psi}_{1}(x,t) + \bigO(\zeta^{-2}) \bigg) \zeta^{\frac{1}{4}\sigma_{3}}A^{-1}e^{{x^{2}}(-\frac{2}{3}{xt}\zeta^{\frac 32}{+}\zeta^{\frac{1}{2}})\sigma_{3}}
\\
\times
\begin{cases}
I, & |\arg \zeta | < \pi -\varepsilon, \\
\begin{pmatrix}
1 & 0 \\
\mp 1 & 1
\end{pmatrix}, & \pi - \varepsilon < \pm \arg \zeta < \pi,
\end{cases}
\end{multline*}
for any $\varepsilon \in (0,\frac{\pi}{2})$, where
\begin{align}\label{eq:hatPsi1}
\widehat{\Psi}_1(x,t) = 
\renewcommand*{\arraystretch}{1.4}
\begin{pmatrix}
\frac{q(x,t)}{x^{2}} & -\frac{i r(x,t)}{x^{3}} \\
\frac{i p(x,t)}{x} & -\frac{q(x,t)}{x^{2}}
\end{pmatrix}.
\end{align}
\end{itemize}
\subsection{Construction of the $g$-function}
By Assumptions \ref{assumptions}, we have, for $M>0$ sufficiently large, the estimates
\begin{equation}\label{eq:estimatelogsigma}
\log F(x^{2}\zeta) =
\begin{cases}
\mathcal O\left(e^{-c_-x^2|\zeta|}\right),& \mbox{for $\zeta<-M/x^2$,}\\
\mathcal O\left(1\right),& \mbox{for $-M/x^2\leq \zeta\leq M/x^2$,}\\
\log c_+'+c_+x^2\zeta+\mathcal O\left(e^{-\epsilon x^2\zeta}\right),& \mbox{for $M/x^2<\zeta$,}
\end{cases}
\end{equation}
which are valid as $x\to +\infty$, uniformly in $\zeta\in\mathbb R$.
Let us decompose $\frac{1}{F(x^2\zeta)}$ as follows:
\begin{align*}
\frac{1}{F(x^2\zeta)} = e^{-x^{2}V(\zeta)}e^{W(\zeta;x)},
\end{align*}
where
\begin{align}
\label{eq:defVW}
& V(\zeta) := \begin{cases}
0, & \zeta < 0, \\
c_+ \zeta, & \zeta > 0,
\end{cases} & & W(\zeta;x) :=- \log F(x^{2}\zeta) + x^{2}V(\zeta).
\end{align}
Whereas $e^{-x^2V(\zeta)}$ decays rapidly as soon as $\zeta>0$ moves away from $0$, $e^{W(\zeta;x)}$ is uniformly bounded and bounded away from $0$ for $\zeta\in\mathbb R$ and $x$ sufficiently large. More precisely, for sufficiently large $M>0$, we have the uniform in $\zeta$ large $x$ asymptotics
\begin{equation}
\label{eq:estimateW}
W(\zeta;x) =
\begin{cases}
\mathcal O\left(e^{-c_-x^2|\zeta|}\right),& \mbox{for $\zeta<-M/x^2$,}\\
\mathcal O\left(1\right),& \mbox{for $-M/x^2\leq \zeta\leq M/x^2$,}\\
-\log c_+'+\mathcal O\left(e^{-\epsilon x^2\zeta}\right),& \mbox{for $M/x^2<\zeta$.}
\end{cases}
\end{equation}
Note that $W$ and $V$ are both (Lipschitz) continuous at $\zeta=0$, but not differentiable. In view of our upcoming RH analysis, we need to construct a $g$-function, which will be the key towards a successful application of the Deift--Zhou nonlinear steepest descent method \cite{DZ}, and which will later turn out to determine the leading order asymptotic behavior of $Q_\sigma(x,t)$ as $t\to 0$ and $K\leq x\leq \delta/t$. 
We aim to construct $g$ such that
\subsubsection*{Conditions for $g$}
\begin{itemize}
\item[(a)] $g(\zeta;xt)$ is analytic in $\zeta\in\mathbb C\setminus(-\infty,{\alpha}(xt)]$,
\item[(b)] $g_+(\zeta;xt) + g_-(\zeta;xt) + \ell(xt) = V(\zeta)$ for $\zeta\in(-\infty,{\alpha}(xt))$,
\item[(c)] $g(\zeta;xt) = -\frac{2}{3}xt \zeta^{3/2} + \zeta^{1/2}+g_{0}(xt)+g_1(xt)\zeta^{-1/2}+ \bigO(\zeta^{-3/2})$ as $\zeta \to\infty$,
\item[(d)] $g_{+}(\zeta;xt)-g_{-}(\zeta;xt) = \bigO((\zeta-{\alpha}(xt))^{3/2})$ as $\zeta \to {\alpha}(xt)$, $\zeta<{\alpha}(xt)$,
\end{itemize}
Here, ${\alpha}(xt)>0$, $g_0(xt)$, $g_1(xt)$, and $\ell(xt)$ are to be determined. The derivative $g'$ then has to satisfy the following.
\subsubsection*{Conditions for $g'$}
\begin{itemize}
\item[(a)] $g'(\zeta;xt)$ is analytic in $\zeta\in\mathbb C\setminus(-\infty,{\alpha}(xt)]$,
\item[(b)] $g_+'(\zeta;xt)+g_-'(\zeta;xt)=V'(\zeta)$ for $\zeta\in(-\infty,{\alpha}(xt)){\setminus\lbrace 0\rbrace}$,
\item[(c)] $g'(\zeta;xt)=-xt\zeta^{1/2}+\frac{1}{2}\zeta^{-\frac{1}{2}}-\frac{1}{2}g_1(xt)\zeta^{-\frac{3}{2}}+\mathcal O(\zeta^{-5/2})$ as $\zeta\to\infty$,
\item[(d)] $g_{+}'(\zeta;xt)-g_{-}'(\zeta;xt) = \bigO((\zeta-{\alpha}(xt))^{1/2})$ as $\zeta \to {\alpha}(xt)$, $\zeta<{\alpha}(xt)$.
\end{itemize}
We can construct $g'$ as 
\begin{align}
\label{gprime 1}
g'(\zeta;xt)&=\frac{1}{(\zeta-{\alpha}(xt))^{1/2}}\bigg(  -xt \zeta + \frac{1+{\alpha}(xt) xt}{2}+\frac{c_+}{2\pi}\int_{0}^{{\alpha}(xt)}\sqrt{{\alpha}(xt)-s}\frac{ds}{s-\zeta} \bigg),
\end{align}
with $(\zeta-{\alpha}(xt))^{1/2}$ analytic except for $\zeta\leq{\alpha}(xt)$ and positive for $\zeta>{\alpha}(xt)$. Then, conditions (a), (b), and (c) are straightforward to verify. In order to have also (d), we need to impose that ${\alpha}(xt)$ is such that $g_{+}'(\zeta;xt)-g_{-}'(\zeta;xt)$ vanishes as $\zeta\to{\alpha}(xt)$.
This is the case if 
\begin{align}
\label{a def}
\frac{c_+}{2\pi}\int_{0}^{{\alpha}(xt)}\frac{ds}{\sqrt{{\alpha}(xt)-s}}=\frac{1-{\alpha}(xt)xt}{2}.
\end{align}
This equation admits two solutions, and we take the solution for ${\alpha}(xt)$ which remains bounded as $xt \to 0$, namely we take
\begin{equation}
\label{def:a}
{\alpha}(xt) := \frac{2c_+^{2}+\pi^{2}xt - 2c_+ \sqrt{c_+^{2}+\pi^{2}xt}}{\pi^{2}(xt)^{2}}.
\end{equation}
\begin{remark}
\label{rmk:uniformity}
The function ${\alpha}(xt)$ is a smooth monotonically decreasing function of $xt>0$, with ${\alpha}(xt)\to\frac{\pi^2}{4c_+^2}$ as $xt\to 0_+$ and ${\alpha}(xt)\to 0_+$ as $xt\to+\infty$. For later convenience we point out that for $xt<\delta$, ${\alpha}(xt)$ is positive and bounded away from $0$.
\end{remark}
\begin{remark}
Observe that $g'$ can equivalently be constructed as 
\begin{align}
g'(\zeta;xt) = \sqrt{\zeta-{\alpha}(xt)} \bigg( -xt- \frac{c_+}{2\pi}\int_{0}^{{\alpha}(xt)} \frac{1}{\sqrt{{\alpha}(xt)-s}} \frac{ds}{s-\zeta} \bigg), \label{gprime 2}
\end{align}
where ${\alpha}(xt)$ is found by requiring that $g'(\zeta;xt) = -xt\zeta^{1/2}+\frac{1}{2}\zeta^{-\frac{1}{2}} + \bigO(\zeta^{-3/2})$ as $\zeta \to \infty$. Here too, we see that ${\alpha}(xt)$ must satisfy \eqref{a def}. Using \eqref{a def}, we easily show that the right-hand sides of \eqref{gprime 1} and \eqref{gprime 2} are indeed equal. With this expression, we immediately obtain that $g'(\zeta;xt)$ remains bounded as $\zeta\to {\alpha}(xt)$, and that condition (d) holds.
\end{remark}
Using a direct primitive, we note that $g'$ can also be expressed more explicitly as
\begin{align}\label{explicit expression for g'}
g'(\zeta;xt) = -xt({\zeta-{\alpha}(xt)})^{1/2}+\frac{c_+}{2\pi i}\log \bigg( \frac{({\zeta-{\alpha}(xt)})^{1/2}+i \sqrt{{\alpha}(xt)}}{({\zeta-{\alpha}(xt)})^{1/2}-i\sqrt{{\alpha}(xt)}} \bigg),
\end{align}
where the principal branches for the log and the square roots are taken. Using this and \eqref{a def}, we also easily see that
\begin{equation}
\label{def:g1}
g_{1}(xt) = - \bigg( \frac{{\alpha}(xt)^{3/2}c_+}{3\pi} + \frac{{\alpha}(xt)^{2}xt}{4} \bigg) = {\alpha}(xt) \bigg( \frac{\sqrt{{\alpha}(xt)}c_+}{6\pi}-\frac{1}{4} \bigg).
\end{equation}
We can now construct $g$ from $g'$ by setting
\begin{align*}
g(\zeta;xt):=\int_{{\alpha}(xt)}^\zeta g'(s;xt)ds.
\end{align*}
The values of the remaining constants $\ell(xt), g_0(xt)$ now follow easily.
Since $g({\alpha}(xt))=0$, we must choose 
\begin{align}
\label{eq:lca}
\ell(xt) = V({\alpha}(xt)) = c_+{\alpha}(xt).
\end{align}
By construction, $g$ satisfies the required jumps.
Because $g_{+}(\zeta;xt)+g_{-}(\zeta;xt)+\ell(xt) = 0$ for all $\zeta \in (-\infty,0]$, we have
\begin{align*}
g_{0}(xt) = - \frac{\ell(xt)}{2} = - \frac{c_+{\alpha}(xt)}{2}.
\end{align*}
\subsection{Normalization of the RH problem}
Let us define 
\begin{align}
\nonumber
T(\zeta;x,t):={}& \begin{pmatrix}
1 & ig_{1}(xt)x^{2} \\
0 & 1
\end{pmatrix} \widehat{\Psi}(\zeta;x,t)e^{-x^{2} (g(\zeta;xt)-g_0(xt))\sigma_3},
\\
\label{def of phi}
\phi(\zeta;xt):={}&2g(\zeta;xt)+\ell(xt)-V(\zeta).
\end{align}
We define an analytic continuation of $V(\zeta)$ on $\mathbb C\setminus i\mathbb R$ by setting it equal to $0$ for $\re\zeta<0$, and to $c_+\zeta$ for $\re\zeta>0$, so that we can also extend $\phi$ to an analytic function in $\C\setminus\left(i\R\cup(-\infty,{\alpha}(xt)]\right)$.
By the jump condition for $g$, we have the relation
\begin{align}
\label{eq:phipm}
\phi_\pm(\zeta;xt)=\pm \left(g_+(\zeta;xt) - g_-(\zeta;xt)\right),\qquad \zeta\in(-\infty,{\alpha}(xt)),
\end{align}
from which RH conditions for $T$ follow in a straightforward way. 

\subsubsection*{RH problem for $T$}
\begin{itemize}
\item[(a)] $T:\mathbb{C}\setminus \mathbb{R} \to \mathbb{C}^{2\times 2}$ is analytic.
\item[(b)] $T$ has the jumps
\begin{align*}
& T_{+}(\zeta;x,t) = T_{-}(\zeta;x,t) \begin{pmatrix}
e^{-x^{2} \phi_{+}(\zeta;xt)} & e^{W(\zeta;x)} \\
0 & e^{-x^{2} \phi_{-}(\zeta;xt)}
\end{pmatrix}, & & \zeta \in (-\infty,\alpha(xt)),\\
& T_{+}(\zeta;x,t) = T_{-}(\zeta;x,t) \begin{pmatrix}
1 & e^{x^{2} \phi(\zeta;xt)}e^{W(\zeta;x)} \\
0 & 1
\end{pmatrix}, & & \zeta \in (\alpha(xt),+\infty).
\end{align*}
\item[(c)] We have
\begin{align}\label{asT}
T(\zeta;x,t) = \bigg( I + \frac{1}{\zeta}T_{1}(x,t) + \bigO(\zeta^{-2}) \bigg) \zeta^{\frac{1}{4}\sigma_{3}}A^{-1}, \qquad \mbox{as } \zeta \to \infty,
\end{align}
where
\begin{equation}
\label{eq:T1}
T_1(x,t) = \begin{pmatrix}
 -\frac{1}{2} g_1^2 x^4 - x g_1 p + \frac{q}{x^2} & \frac{i x^6 g_1^3 }3+ i x^3 g_1^2 p -i x^2 g_2-2 i g_1 q-\frac{i r}{x^3} \\
 i x^2 g_1 +\frac{i p}{x} & \frac{1}{2} g_1^2 x^4+ x g_1 p - \frac{q}{x^2} \\
\end{pmatrix}
\end{equation}
where $g_i=g_i(xt)$, $p=p(x,t)$, $q=q(x,t)$, and $r=r(x,t)$.
\item[(d)] As $\zeta \to 0$ and as $\zeta\to{\alpha}(xt)$, we have $T(\zeta;x,t) = \bigO(1)$.
\end{itemize}

\subsection{Opening of the lenses}
Note that the jump of $T$ for $\zeta \in (-\infty,{\alpha}(xt))$ can be written as
\begin{multline}\label{factorization of the jump}
\begin{pmatrix}
e^{-x^{2} \phi_{+}(\zeta;xt)} & e^{W(\zeta;x)} \\
0 & e^{-x^{2} \phi_{-}(\zeta;xt)}
\end{pmatrix} \\= \begin{pmatrix}
1 & 0 \\
e^{-W(\zeta;x)-x^{2} \phi_{-}(\zeta;xt)} & 1
\end{pmatrix}
\begin{pmatrix}
0 & e^{W(\zeta;x)} \\ -e^{-W(\zeta;x)} & 0
\end{pmatrix} \begin{pmatrix}
1 & 0 \\ e^{-W(\zeta;x)-x^{2} \phi_{+}(\zeta;xt)} & 1
\end{pmatrix}.
\end{multline}
Now we will use this factorization to open lenses around $(-\infty,{\alpha}(xt))$: we will split it into two parts, and deform the upper part into the upper half plane, and the lower part into the lower half plane.
Although the general principles of RH asymptotic analysis \cite{DKMVZ1} suggest to choose these curves close to the real line, we prefer to deform the upper and lower parts of $(-\infty,{\alpha}(xt))$ all the way to the vertical half-lines $\alpha(xt)\pm i\mathbb R^+$.
This is not essential, but it will be convenient later to prove that the jump matrices are small on these parts of the jump contour.
The function $\phi$ is analytic in $\C\setminus(i\R\cup(-\infty,{\alpha}(xt)])$ only, but the identity
\begin{align*}
x^{2}\phi(\zeta;xt)+W(\zeta;x) = 2x^{2}(g(\zeta;xt)-g_{0})-\log F(x^{2}\zeta)
\end{align*}
shows that the expressions
\begin{equation}
\exp\left(-x^{2}\phi_\pm(\zeta;xt)-W(\zeta;x)\right)=F(x^2\zeta){\exp\left(-2x^2(g_\pm(\zeta;xt)-g_0)\right)},\qquad \zeta\in\R,
\end{equation}
appearing in the factorization \eqref{factorization of the jump}, are boundary values of a function analytic in $\C\setminus(-\infty,{\alpha}(xt)]$, due to the conditions defining $g$ and Assumptions \ref{assumptions}.
The next transformation can then be defined by
\begin{align}
\label{eq:StoT}
S(\zeta;x,t) :=
\begin{cases}
T(\zeta;x,t)\begin{pmatrix}
1 & 0 \\
\mp e^{-x^{2} \phi(\zeta;xt)}e^{-W(\zeta;x)} & 1
\end{pmatrix}, & \mbox{if } \re\zeta < {\alpha}(xt) \mbox{ and } \pm\im \zeta >0, 
\\
T(\zeta;x,t), & \mbox{if } \re\zeta > {\alpha}(xt).
\end{cases}
\end{align}
Indeed, by the above discussion, $e^{-x^{2} \phi(\zeta;xt)}e^{-W(\zeta;x)} = F(x^2\zeta){\exp\left(-2x^2(g(\zeta;xt)-g_0)\right)}$ is an analytic function of $\zeta$ in the relevant regions.
RH conditions for $S$ follow from those for $T$ along with the identity \eqref{factorization of the jump}.

\subsubsection*{RH problem for $S$}
\begin{itemize}
\item[(a)] $S: \mathbb{C}\setminus \big( \mathbb{R} \cup (i\R+{\alpha}(xt) \big) \to \mathbb{C}^{2\times 2}$ is analytic.
\item[(b)] $S$ has the jumps
\begin{align*}
& S_{+}(\zeta;x,t) = S_{-}(\zeta;x,t)\begin{pmatrix}
1 & 0 \\
e^{-x^{2} \phi(\zeta;xt)}e^{-W(\zeta;x)} & 1
\end{pmatrix}, & & \zeta \in i\R+{\alpha}(xt),
\\
& S_{+}(\zeta;x,t) = S_{-}(\zeta;x,t) \begin{pmatrix}
0 & e^{W(\zeta;x)} \\
-e^{-W(\zeta;x)} & 0
\end{pmatrix}, & & \zeta \in (-\infty,{\alpha}(xt)),
\\
& S_{+}(\zeta;x,t) = S_{-}(\zeta;x,t) \begin{pmatrix}
1 & e^{x^{2} \phi(\zeta;xt)}e^{W(\zeta;x)} \\
0 & 1
\end{pmatrix}, & & \zeta \in ({\alpha}(xt),+\infty),
\end{align*}
where $\mathbb R$ is oriented from left to right, and the vertical half-lines $\alpha(xt)\pm i\mathbb R^+$ are oriented towards $\alpha(xt)$. According to this orientation, the subscript $+$ (resp. $-$) indicates the boundary value from the left (resp. right).
\item[(c)] We have
\begin{align*}
S(\zeta;x,t) = \bigg( I + \frac{1}{\zeta}T_{1}(x,t) + \bigO(\zeta^{-2}) \bigg) \zeta^{\frac{1}{4}\sigma_{3}}A^{-1}, \qquad \mbox{as } \zeta \to \infty,
\end{align*}
where $T_1(x,t)$ is the same as in \eqref{asT}.
\item[(d)] As $\zeta\to 0$ and as $\zeta\to{\alpha}(xt)$, we have $S(\zeta;x,t) = \bigO(1)$.
\end{itemize}

We now show that the jump matrix for $S$ is close to the identity except on $(-\infty,{\alpha}(xt))$ and near ${\alpha}(xt)$.

\begin{proposition}\label{prop:lenses}
{\it 1.}
For any $\rho>0$, there exists $\eta>0$ such that the inequality
\[\phi({\alpha}(xt)+v;xt)<-\eta v\]
holds for $v>\rho$ and $xt\leq \delta$.

\noindent {\it 2.}
For any $\rho>0$, there exists $\eta>0$ such that the inequalities
\[\re\phi({\alpha}(xt)\pm iv;xt)>\eta\sqrt v\]
hold for $v>\rho$ and $xt\leq \delta$.
\end{proposition}

\begin{proof}
{\it 1.}
Let us rewrite \eqref{explicit expression for g'} for $\zeta={\alpha}(xt)+v$, $v>0$, as
\begin{align}
g'({\alpha}(xt)+v)=-xt\sqrt v + \frac{c_+}\pi\arctan\sqrt{\frac {{\alpha}(xt)}v}.
\end{align}
Integration of this expression, using $g({\alpha}(xt))=0$, yields for $v>0$,
\begin{align*}
g({\alpha}(xt)+v)
=-\frac 23xt v^{3/2} + \frac{c_+}\pi\left(\sqrt{{\alpha}(xt)v}-\frac{{\alpha}(xt)\pi}2+({\alpha}(xt)+v)\arctan\sqrt{\frac{{\alpha}(xt)}v}\right).
\end{align*}
By omitting the negative term $-\frac 23xt v^{3/2}$ we obtain, using \eqref{eq:lca} and \eqref{def of phi},
\begin{align*}
\phi({\alpha}(xt)+v) = 2g({\alpha}(xt)+v) - c_+v
< \frac{2c_+}\pi\left(\sqrt{{\alpha}(xt)v}-({\alpha}(xt)+v)\arctan\sqrt{\frac v{{\alpha}(xt)}}\right).
\end{align*}
The function $\mathcal{F}({\alpha},v):=\frac{2c_+}\pi\left(\sqrt{{\alpha} v}-({\alpha}+v)\arctan\sqrt{\frac v{{\alpha}}}\right)$ (for ${\alpha},v>0$) is concave in $v$, $\mathcal{F}({\alpha},0)=0$, and $\mathcal{F}({\alpha},v)<0$ for $v>0$; therefore, for any $\rho>0$ we have $\mathcal{F}({\alpha},v)<-\eta v$ for $v>\rho$, where $\eta({\alpha}):=-\mathcal{F}({\alpha},\rho)/\rho>0$.
To see that we can take $\eta$ in the statement of the proposition independent of $x,t>0$, note that $\eta({\alpha})$ is decreasing in ${\alpha}$ and ${\alpha}(xt)\leq \pi^2/(4c_+^2)$, see Remark \ref{rmk:uniformity}, so that the proof is complete by taking $\eta = \eta( \pi^2/(4c_+^2))$.

\noindent{\it 2.}
Let us rewrite \eqref{explicit expression for g'} for $\zeta={\alpha}(xt) \pm i v$, $v>0$, as
\begin{align*}
g'({\alpha}(xt)\pm iv)=-xt e^{\pm i\pi/4}\sqrt v+\frac{c_+}{2\pi i}\log\left(\frac{e^{\pm i\pi/4}\sqrt v+i\sqrt{{\alpha}(xt)}}{e^{\pm i\pi/4}\sqrt v-i\sqrt{{\alpha}(xt)}}\right).
\end{align*}
For $v>0$ we have $\re\phi({\alpha}(xt)\pm i v)=2\re g({\alpha}(xt)\pm i v)$, see \eqref{eq:lca} and \eqref{def of phi}, so that
\begin{align*}
\re \phi({\alpha}(xt)\pm i v) &=\pm 2\,\re i\int_0^v g'({\alpha}(xt)\pm is) ds 
\\
&=
\int_0^v \left(xt\sqrt{2s}\pm
\frac{c_+}{\pi}\log\left|
\frac{e^{\pm i\pi/4}\sqrt s+i\sqrt{{\alpha}(xt)}}{e^{\pm i\pi/4}\sqrt s-i\sqrt{{\alpha}(xt)}}
\right|\right)ds
\\
&=\frac {2\sqrt 2}3 xt v^{3/2} + \frac{c_+}{2\pi}\int_0^v\log\left(1+\frac{4}{\sqrt{\frac{2{\alpha}(xt)}s}+\sqrt{\frac{2s}{{\alpha}(xt)}}-2}\right)ds
\\
&\geq \frac{\sqrt{2}c_+}{\pi}\int_0^v\frac{ds}{\sqrt{\frac s{{\alpha}(xt)}}+\sqrt{\frac {{\alpha}(xt)}s}} 
= \frac{2\sqrt{2}c_+}\pi\left (\sqrt{{\alpha}(xt)v} - {\alpha}(xt)\arctan\sqrt{\frac v{{\alpha}(xt)}}\right),
\end{align*}
where we use the inequality $\log\left(1+4(\sqrt 2(\xi+\xi^{-1})-2)^{-1}\right)\geq 2\sqrt 2\left(\xi+\xi^{-1}\right)^{-1}$, valid for all $\xi>0$.
The function $\mathcal H({\alpha},v):=\frac{2\sqrt{2}c_+}\pi\left(\sqrt{{\alpha} v}-{\alpha}\arctan\sqrt{\frac v{{\alpha}}}\right)$ (for ${\alpha},v>0$) is convex with respect to the variable $\sqrt v$, $\mathcal H({\alpha},0)=0$, and $\mathcal H({\alpha},v)>0$ for $v>0$; therefore, for any $\rho>0$ we have $\mathcal H({\alpha},v)>\eta\sqrt v$ for $v>\rho$, where $\eta({\alpha}):=\mathcal H({\alpha},\rho)/\sqrt\rho>0$.
To see that we can take $\eta$ in the statement of the proposition independent of $x,t>0$, note that $\eta({\alpha})$ is bounded away from $0$ provided ${\alpha}$ does not tend to $0$, and that $\alpha(xt)$ does not tend to $0$ for $xt\leq\delta$ by Remark \ref{rmk:uniformity}.
\end{proof}

\begin{corollary}
\label{corollary:smalljumpS}
For any sufficiently large $K$ and any $\rho>0$, there exists $\eta>0$ such that for $\zeta \in (\alpha(xt)+i\R)\cup ({\alpha}(xt),+\infty)$ and satisfying $|\zeta-{\alpha}(xt)|>\rho$, the jump matrix for $S$ is $I+\mathcal O\left(\exp(-\eta x^2\sqrt{|\zeta-{\alpha}(xt)|})\right)$, uniformly in $x,t>0$ such that $xt\leq \delta$, $x \geq K$.
\end{corollary}

\begin{proof}
Equation \eqref{eq:estimateW} implies that $|e^{W(\zeta;x)}|$ remains bounded for $\zeta \in ({\alpha}(xt),+\infty)$.
Also, for $\zeta \in {\alpha}(xt)+i\R$, by part 4 of Assumptions \ref{assumptions} and the definition \eqref{eq:defVW} of $W$, we have $|e^{-W(\zeta;x)}| = \big|\frac{F(x^{2}\zeta)}{e^{x^{2}c_{+} \zeta}}\big| = \bigO(1)$.
The claim now directly follows from Proposition \ref{prop:lenses} and the definition of the jump matrix of $S$.
\end{proof}

\subsection{Global parametrix}

Let us for a moment ignore the jumps on $(\alpha(xt),+\infty)\cup(i\mathbb{R}+\alpha(xt))$ in the RH problem for $S$.
We then obtain the following RH problem for $P^{(\infty)}$. 

\subsubsection*{RH problem for $P^{(\infty)}$}
\begin{itemize}
\item[(a)] $P^{(\infty)}: \mathbb{C}\setminus (-\infty,{\alpha}(xt)] \to \mathbb{C}^{2\times 2}$ is analytic.
\item[(b)] $P^{(\infty)}$ has the jumps
\begin{align*}
& P^{(\infty)}_{+}(\zeta;x,t) = P^{(\infty)}_{-}(\zeta;x,t) \begin{pmatrix}
0 & e^{W(\zeta;x)} \\
-e^{-W(\zeta;x)} & 0
\end{pmatrix}, & & \zeta \in (-\infty,{\alpha}(xt)).
\end{align*}
\item[(c)] We have
\begin{align}\label{asymp for Pinf}
P^{(\infty)}(\zeta;x,t) = \bigg( I + \frac{1}{\zeta}P^{(\infty)}_{1}(x,t) + \bigO(\zeta^{-2}) \bigg) \zeta^{\frac{1}{4}\sigma_{3}}A^{-1}, \qquad \mbox{as } \zeta \to \infty.
\end{align}
\end{itemize}
If we impose in addition that we have the condition 
\begin{itemize}
\item[(d)] As $\zeta \to {\alpha}(xt)$, $P^{(\infty)}(\zeta;x,t) = \bigO(|\zeta-{\alpha}(xt)|^{-\frac{1}{4}})$,
\end{itemize}
there is a unique solution to this RH problem, and it will turn out to be a good approximation to $S$ for large $x$, for $\zeta$ not too close to $0$ and ${\alpha}(xt)$.

The unique solution to this RH problem is given by
\begin{align*}
P^{(\infty)}(\zeta;x,t)=\begin{pmatrix}
1 & i d_{1}(x,t) \\
0 & 1
\end{pmatrix} (\zeta-{\alpha}(xt))^{\frac{1}{4}\sigma_{3}}A^{-1}e^{-D(\zeta;x,t)\sigma_{3}},
\end{align*}
with
\begin{align*}
D(\zeta;x,t)=\frac{\sqrt{\zeta-{\alpha}(xt)}}{2\pi} \int_{-\infty}^{{\alpha}(xt)} \frac{W(s;x)}{\sqrt{{\alpha}(xt)-s}}\frac{ds}{\zeta-s}.
\end{align*}
Indeed, the function $D$ satisfies $D_{+}(\zeta;x,t)+D_{-}(\zeta;x,t)=W(\zeta;x)$ for $\zeta\in(-\infty,{\alpha}(xt))$, from which the jump condition for $P^{(\infty)}$ follows.
Moreover, we have the asymptotics
\begin{align}
\label{eq:di}
D(\zeta;x,t)=\zeta^{-\frac 12}d_1(x,t)+\bigO(\zeta^{-3/2})\mbox{  as } \zeta \to \infty,\qquad d_1(x,t) =  \frac 1{2\pi} \int_{-\infty}^{\alpha(xt)}\frac{W(s;x)}{\sqrt{\alpha(xt)-s}}ds,
\end{align}
from which \eqref{asymp for Pinf} follows.
In particular, let us record the asymptotic behavior of $d_1(x,t)$ for later use.
\begin{proposition}\label{prop:d1}
As $x\to +\infty$, we have uniformly in $xt\leq\delta$ for any $\delta>0$, that
\begin{align}
\nonumber
d_1(x,t)&=-\frac{\sqrt{{\alpha}(xt)}}{\pi}
\log c_+'+\frac{j_\sigma}{x^{2} \sqrt{{\alpha}(xt)}}+\bigO(x^{-4}),
\\ 
\label{eq:Isigma}
j_\sigma & := -\frac 1{2\pi}\int_{-\infty}^{+\infty}\left[\log F(r)-(c_+r+\log c_+')1_{(0,+\infty)}(r)\right]dr.
\end{align}
\end{proposition}
\begin{proof}
Let us write, using \eqref{eq:defVW}, $d_1(x,t)=\frac{1}{2\pi}(I_1+I_2+I_3+I_4)$, where
\begin{align*}
I_1:={}&-\int_{-\infty}^{-\frac{{\alpha}(xt)}{2}}\frac{\log F(x^2 s)}{\sqrt{{\alpha}(xt)-s}}d s,\quad
I_2:=-\int_{-\frac{{\alpha}(xt)}{2}}^{\frac{{\alpha}(xt)}{2}}\frac{\log F(x^2 s)-\left(c_+x^2s+\log c'_+\right)1_{(0,+\infty)}(s)}{\sqrt{{\alpha}(xt)-s}}d s,
\\
I_3:={}&-\int_{\frac{{\alpha}(xt)}{2}}^{{\alpha}(xt)}\frac{\log F(x^2 s)-\left(c_+x^2s+\log c'_+\right)}{\sqrt{{\alpha}(xt)-s}}d s,\quad
I_4:=-\int_{0}^{{\alpha}(xt)}\frac{\log c'_+}{\sqrt{{\alpha}(xt)-s}}d s=-2\sqrt{{\alpha}(xt)}\log c'_+.
\end{align*}
The terms $I_1,I_3$ are exponentially small as $x\to+\infty$ by \eqref{eq:estimateW};
\begin{align}
I_1=\bigO\left(\frac{e^{-c_-x^2{\alpha}(xt)/2}}{\sqrt{{\alpha}(xt)}}\right),\qquad I_3=\bigO\left(e^{-\epsilon x^2{\alpha}(xt)/2}\sqrt{{\alpha}(xt)}\right).
\end{align}
The estimate is uniform for $xt\leq \delta$, since ${\alpha}(xt)$ is positive and bounded away from zero (see Remark \ref{rmk:uniformity}).
To estimate $I_2$ we perform the change variable $x^2s=r$ and obtain
\begin{align*}
I_2&=-\frac 1{x^2}\int_{-x^2\frac{{\alpha}(xt)}{2}}^{x^2\frac{{\alpha}(xt)}{2}}\frac{\log F(r)-(c_+r+\log c_+')1_{(0,+\infty)}(r)}{\sqrt{{\alpha}(xt)-r/x^2}}d r
\\
&=\frac{-1}{x^2\sqrt{{\alpha}(xt)}}\sum_{k\geq 0}\frac{(2k-1)!!}{(2k)!!\left({\alpha}(xt)x^2\right)^k}\int_{-x^2\frac{{\alpha}(xt)}{2}}^{x^2\frac{{\alpha}(xt)}{2}}r^k\left[\log F(r)-(c_+r+\log c_+')1_{(0,+\infty)}(r)\right]d r
\\
&=\frac{-1}{x^2\sqrt{{\alpha}(xt)}}\int_\R\left[\log F(r)-(c_+r+\log c_+')1_{(0,+\infty)}(r)\right]d r+\mathcal O(x^{-4}).
\end{align*}
In the last step we use again \eqref{eq:estimatelogsigma} to show that
\begin{multline*}
\int_{-x^2\frac{{\alpha}(xt)}{2}}^{x^2\frac{{\alpha}(xt)}{2}}\left[\log F(r)-(c_+r+\log c_+')1_{(0,+\infty)}(r)\right]d r
\\
=\int_\R\left[\log F(r)-(c_+r+\log c_+')1_{(0,+\infty)}(r)\right]d r+\mathcal O\left(e^{-\min\{\epsilon,c_-\} x^2{\alpha}(xt)/2}\right).
\end{multline*}
\end{proof}
Recall that $W$ is Lipschitz continuous (but not differentiable) at $\zeta=0$. Hence $D$ admits well-defined boundary values at $\zeta=0$. Moreover, expanding $P^{(\infty)}(\zeta)$ as $\zeta\to\infty$, we obtain
\begin{align}\label{P1inf}
P_{1}^{(\infty)}(x,t) = \begin{pmatrix}
-\frac{{\alpha}(xt)}{4}-\frac{d_{1}(x,t)^{2}}{2} & \frac{i}{6}(3{\alpha}(xt)d_{1}(x,t)+2d_{1}(x,t)^{3}-6d_{2}(x,t)) \\ id_{1}(x,t) & \frac{{\alpha}(xt)}{4}+\frac{d_{1}(x,t)^{2}}{2}
\end{pmatrix}
\end{align}
in \eqref{asymp for Pinf}.
Finally, let us discuss the behavior of $P^{(\infty)}$ as $\zeta\to{\alpha}(xt)$; to this end write
\begin{align}
D(\zeta;x,t)&=\frac{W(\zeta;x)}2-\frac{(\zeta-{\alpha}(xt))^{1/2}}{2\pi}\int_{-\infty}^{{\alpha}(xt)}\frac{W(\zeta;x)-W(s;x)}{(\zeta-s)\sqrt{{\alpha}(xt)-s}}ds.
\end{align}
Since $W(\zeta;x)$ is analytic near $\zeta={\alpha}(xt)$, we have the Taylor series
\begin{align}
\frac{W(\zeta;x)-W(s;x)}{\zeta-s} = \sum_{\ell = 0}^{+\infty}\frac{(\zeta-{\alpha}(xt))^\ell}{(s-{\alpha}(xt))^{\ell+1}}\left[W(s;x)-\sum_{j=0}^{\ell}W^{(j)}({\alpha}(xt);x)\frac{(s-{\alpha}(xt))^j}{j!}\right]
\end{align}
and so we have the following Poincar\'e asymptotic series, as $\zeta\to{\alpha}(xt)$ away from $(-\infty,{\alpha}(xt))$,
\begin{align}
\int_{-\infty}^{{\alpha}(xt)}\frac{W(\zeta;x)-W(s;x)}{(\zeta-s)\sqrt{{\alpha}(xt)-s}}ds
&\sim\sum_{\ell=0}^{+\infty}\left(\zeta-{\alpha}(xt)\right)^\ell\int_{-\infty}^{{\alpha}(xt)}\frac{W(s;x)-\sum_{j=0}^{\ell}W^{(j)}({\alpha}(xt);x)\frac{(s-{\alpha}(xt))^j}{j!}}{(s-\alpha(xt))^{\ell+1}\,\sqrt{\alpha(xt)-s}}ds.
\end{align}
In particular, as $\zeta\to {\alpha}(xt)$,
\begin{align}
D(\zeta;x,t)&=\frac{W(\zeta;x)}2-\chi(x,t)(\zeta-{\alpha}(xt))^{1/2}+\bigO((\zeta-{\alpha}(xt))^{3/2}),
\\\label{eq:chi}
\chi(x,t)&=\frac 1{2\pi}\int_{-\infty}^{{\alpha}(xt)}\frac{W({\alpha}(xt);x)-W(s;x)}{({\alpha}(xt)-s)^{3/2}}ds.
\end{align}
\begin{proposition}
\label{prop:chi}
As $x\to +\infty$, we have uniformly for $xt\leq\delta$ for any $\delta>0$, that
\begin{align}
\chi(x,t)=-\frac{\log c_+'}{\pi\sqrt{{\alpha}(xt)}}+\mathcal O(x^{-2}).
\end{align}
\end{proposition}
\begin{proof}
Write $\chi=\chi_1+\chi_2-\frac {\log c'_+}{2\pi}\int_{-\infty}^0({\alpha}(xt)-s)^{-3/2}ds=\chi_1+\chi_2-\log c_+'/(\pi \sqrt{{\alpha}(xt)})$, where
\begin{align*}
\chi_1:=\int_{-\infty}^0\frac{W({\alpha}(xt);x)+\log c_+'-W(s;x)}{({\alpha}(xt)-s)^{3/2}}\frac{ds}{2\pi},\qquad
\chi_2:={}&\int_0^{{\alpha}(xt)}\frac{W({\alpha}(xt);x)-W(s;x)}{({\alpha}(xt)-s)^{3/2}}\frac{ds}{2\pi}.
\end{align*}
As $x\to+\infty$, by using \eqref{eq:estimateW} and the change of variables $r=x^2s$, we have
\begin{align*}
\chi_1&=(W({\alpha}(xt);x)+\log c_+')\int_{-\infty}^0\frac{ds}{2\pi({\alpha}(xt)-s)^{3/2}}
+\int_{-\infty}^0\frac{\log F(x^2s)}{({\alpha}(xt)-s)^{3/2}}\frac{ds}{2\pi}
\\
&=\bigO(e^{-\epsilon x^2 {\alpha}(xt)})+\int_{-\infty}^0\frac{\log F(r)}{({\alpha}(xt)-x^{-2}r)^{3/2}}\frac{dr}{2\pi x^2}
\\
&=\frac{1}{2\pi x^2{\alpha}(xt)^{3/2}}\int_{-\infty}^0\log F(r)dr+\bigO(x^{-4}) = \mathcal O(x^{-2}),
\end{align*}
and these estimates are uniform in $xt\leq\delta$ for any $\delta>0$, see Remark \ref{rmk:uniformity}.
Similarly, let $x\to+\infty$ and $\beta\to 0_+$ (in a manner to be specified below) so that
\begin{align*}
\chi_2&=(W({\alpha}(xt);x)+\log c_+')\int_0^{{\alpha}(xt)-\beta}\frac{ds}{2\pi({\alpha}(xt)-s)^{3/2}}
\\
&\qquad
+\int_0^{{\alpha}(xt)-\beta}\frac{\log F(x^2s)-c_+x^2s-\log c'_+}{({\alpha}(xt)-s)^{3/2}}\frac{ds}{2\pi}+\int_{{\alpha}(xt)-\beta}^{{\alpha}(xt)}\frac{W({\alpha}(xt);x)-W(s;x)}{({\alpha}(xt)-s)^{3/2}}\frac{ds}{2\pi}
\\
&=\bigO\left(\frac{e^{-\epsilon x^2 {\alpha}(xt)}}{\sqrt\beta}\right)
+\int_0^{x^2({\alpha}(xt)-\beta)}\frac{\log F(r)-c_+r-\log c'_+}{({\alpha}(xt)-x^{-2}r)^{3/2}}\frac{dr}{2\pi x^2}+\bigO\left(\sqrt{\beta}W'({\alpha}(xt);x)\right).
\end{align*}
By Assumptions \ref{assumptions}, we obtain
\begin{align*}
W'({\alpha}(xt);x) = x^2\left(c_+-(\log F)'(x^2\alpha(xt))\right)
=
\bigO(x^2e^{-\frac \epsilon 2 x^2\alpha(xt)}),
\end{align*}
where we use \eqref{eq:vplusinfty},
and so, by choosing for instance $\beta=\bigO(e^{-\epsilon x^2\alpha(xt)})$, we get
\begin{align*}
\chi_2=\frac 1{2\pi x^2{\alpha}(xt)^{3/2}}\int_0^{+\infty}(\log F(r)-c_+r-\log c'_+)dr+\bigO(x^{-4})=\mathcal O(x^{-2}).
\end{align*}
These estimates are uniform in $xt\leq \delta$, see again Remark \ref{rmk:uniformity}.
\end{proof}

\subsection{Local Airy parametrix near $\alpha(xt)$}\label{subsection:Airy local param}

We now construct a local parametrix in a neighborhood of $\alpha(xt)$, with exactly the same jumps as $S$.
Let us introduce
\begin{align*}
P^{({\alpha})}(\zeta;x,t) := E(\zeta;x,t) \Phi_{\mathrm{Ai}}^{[j]}(x^{\frac{4}{3}}f(\zeta;xt))e^{-\frac{x^{2}}{2}\phi(\zeta;xt)\sigma_{3}}e^{-\frac{W(\zeta;x)}{2}\sigma_{3}},
\end{align*}
where $j=1$ for $0<\arg(\zeta-\alpha(xt))<\pi/2$, $j=2$ for $\pi/2<\arg(\zeta-\alpha(xt))<\pi$, $j=3$ for $-\pi<\arg(\zeta-\alpha(xt))<-\pi/2$, and $j=4$ for $-\pi/2<\arg(\zeta-\alpha(xt))<0$, with
$\Phi_{\mathrm{Ai}}^{[1]},\ldots,\Phi_{\mathrm{Ai}}^{[4]}$ the entire functions defined in Appendix \ref{app:airy} which together form the solution to the Airy model RH problem.
The functions $f$ and $E$ are given by
\begin{align}\label{def of f and E}
f(\zeta;xt) := \left( -\frac{3}{4}\phi(\zeta;xt) \right)^{2/3}, \quad 
E(\zeta;x,t) := P^{(\infty)}(\zeta;x,t) e^{\frac{W(\zeta;x)}{2}\sigma_{3}} A^{-1} \left( x^{\frac{4}{3}}f(\zeta;xt) \right)^{\frac{\sigma_{3}}{4}}.
\end{align}
We now contend that there exists a sufficiently small $\rho>0$ such that $f(\zeta;xt)$ is a conformal transformation of $\zeta\in\mathcal D :=\{\zeta\in\mathbb C:\ |\zeta-{\alpha}(xt)|<\rho\}$ onto a neighborhood of $0$, satisfying $f'(\alpha(xt);xt)>0$, and that $E(\zeta;x,t)$ is a holomorphic function of $\zeta\in\mathcal D$.

For the first statement, we note from \eqref{explicit expression for g'} that, as $\zeta \to {\alpha}(xt)$, we have
\begin{align}
g'(\zeta;xt)=\frac{c_+}2-\left(\frac{c_+}{\pi \alpha(xt)^{1/2}}+xt\right)\left(\zeta-{\alpha}(xt)\right)^{1/2}+
\frac{c_+}{\pi}\sum_{\ell=1}^{+\infty}\frac{(-1)^{\ell+1}}{2\ell+1}\frac{(\zeta-{\alpha}(xt))^{\ell+1/2}}{{\alpha}(xt)^{\ell+1/2}}.
\end{align}
Since $\phi(\zeta;xt)=2g(\zeta;xt)-c_+(\zeta-{\alpha}(xt))$ for $\re\zeta>0$, see {\eqref{eq:defVW},} \eqref{eq:lca} and \eqref{def of phi}, integrating this relation (using $g({\alpha}(xt))=0$) we conclude that, as $\zeta\to{\alpha}(xt)$, we have 
\begin{align}
\label{eq:seriesphi}
\phi(\zeta;xt)=\left(\zeta-{\alpha}(xt)\right)^{3/2}\left[-\frac 43\left(\frac{c_+}{\pi {\alpha}(xt)^{1/2}}+xt\right)+\frac{4c_+}{\pi}\sum_{\ell=1}^{+\infty}\frac{(-1)^{\ell+1}}{(2\ell+3)(2\ell+1)}\frac{(\zeta-{\alpha}(xt))^{\ell}}{{\alpha}(xt)^{\ell+1/2}}\right].
\end{align}
Hence $f(\zeta;xt)$ defined in \eqref{def of f and E} is holomorphic in a neighborhood of ${\alpha}(xt)$ and $f'({\alpha}(xt);xt)>0$, proving the claim about $f$. For later reference we compute the first terms in the Taylor series at $\zeta={\alpha}(xt)$;
\begin{align}
\nonumber
f(\zeta;xt)&=f_1(xt)(\zeta-{\alpha}(xt))+f_2(xt)(\zeta-{\alpha}(xt))^2+\bigO\left((\zeta-{\alpha}(xt))^3\right),\qquad \zeta\to {\alpha}(xt),
\\
\label{eq:fTaylor}
f_1(xt)&=\left(\frac{c_+}{\pi {\alpha}(xt)^{1/2}}+xt\right)^{2/3},\quad
f_2(xt)=-\frac{2c_+}{15\pi {\alpha}(xt)^{3/2}}\left(\frac{c_+}{\pi {\alpha}(xt)^{1/2}}+xt\right)^{-1/3}.
\end{align}
We now prove that $E(\zeta;x,t)$ is holomorphic for $\zeta\in\mathcal D$.
To this end, first we note that $f(\zeta;xt)$ is analytic in $\mathcal D$ vanishing linearly at $\zeta={\alpha}(xt)$ so that
\begin{align}
\left( f(\zeta;xt)^{\sigma_3/4} \right)_+ = \begin{pmatrix}i & 0 \\ 0 & -i \end{pmatrix}\left( f(\zeta;xt)^{\sigma_3/4}\right)_-.
\end{align}
On the other hand, directly from the jump condition for $P^{(\infty)}$ we get
\begin{align}
\left(P^{(\infty)}(\zeta;x,t)e^{\frac {W(\zeta;x)}2\sigma_3}A^{-1}\right)_+
=
\left(P^{(\infty)}(\zeta;x,t)e^{\frac {W(\zeta;x)}2\sigma_3}A^{-1}\right)_-\begin{pmatrix}-i & 0 \\ 0 & i \end{pmatrix}.
\end{align}
Therefore, $E(\zeta;x,t)$ satisfies $E_+(\zeta;x,t)=E_-(\zeta;x,t)$ for ${\alpha}(xt)-\rho<\zeta<{\alpha}(xt)$ and so it has at worst an isolated singularity at $\zeta={\alpha}(xt)$; however $f(\zeta;xt)^{-1/4},P^{(\infty)}(\zeta;x,t)$ are both $\bigO\left((\zeta-{\alpha}(xt))^{-1/4}\right)$ as $\zeta\to {\alpha}(xt)$ so that $E(\zeta;x,t)=\bigO\left((\zeta-{\alpha}(xt))^{-1/2}\right)$ as $\zeta\to {\alpha}(xt)$.
Hence ${\alpha}(xt)$ is a removable singularity and $E(\zeta;x,t)$ is analytic in the whole disk $\mathcal D$, as claimed.

\subsection{Small norm RH problem}\label{subsection:small norm with xt small}

Define
\begin{align}\label{def of R section 3}
R(\zeta;x,t) := \begin{cases}
S(\zeta;x,t)P^{(\infty)}(\zeta;x,t)^{-1}, & \zeta \in \mathbb{C}\setminus \overline{\mathcal{D}}, \\
S(\zeta;x,t)P^{({\alpha})}(\zeta;x,t)^{-1}, & \zeta \in \mathcal{D}.
\end{cases}
\end{align}
The RH conditions for $R$ are detailed below and follow directly from those of $S$, $P^{(\infty)}$, and $P^{(\alpha)}$. 
Let us denote $\Gamma_R:=\partial \mathcal D \cup \Gamma_{-}\cup\Gamma_0\cup\Gamma_{+}$ where $\Gamma_0:=({\alpha}(xt)+\rho,+\infty)$, $\Gamma_\pm := ({\alpha}(xt)\pm i\rho,{\alpha}(xt)\pm i\infty)$; we consider $\Gamma_{0}$ oriented from left to right, $\Gamma_\pm$ oriented towards $\alpha(xt)$, and $\partial \mathcal D$ oriented clockwise (see Figure \ref{figcontourR}).
As usual, boundary values are labeled by $+$ (resp. $-$) when we approach the oriented contour from the left (resp. right).

\subsubsection*{RH problem for $R$}
\begin{itemize}
\item[(a)] $R: \mathbb{C}\setminus\Gamma_R \to \mathbb{C}^{2\times 2}$ is analytic.
\item[(b)] $R$ has the jumps
\begin{align}
&\label{eq:jumpRconditions1}
R_{+}(\zeta;x,t) = R_{-}(\zeta;x,t)P^{(\infty)}(\zeta;x,t)\begin{pmatrix}
1 & 0 \\
e^{-x^{2} \phi(\zeta;xt)}e^{-W(\zeta;x)} & 1
\end{pmatrix}P^{(\infty)}(\zeta;x,t)^{-1}, & & \zeta \in \Gamma_{\pm},
\\
& R_{+}(\zeta;x,t) = R_{-}(\zeta;x,t) P^{(\infty)}(\zeta;x,t)\begin{pmatrix}
1 & e^{x^{2} \phi(\zeta;xt)}e^{W(\zeta;x)} \\
0 & 1
\end{pmatrix}P^{(\infty)}(\zeta;x,t)^{-1}, & & \zeta \in \Gamma_0,
\\
& R_{+}(\zeta;x,t) = R_{-}(\zeta;x,t)P^{({\alpha})}(\zeta;x,t)P^{(\infty)}(\zeta;x,t)^{-1}, & & \zeta \in \partial\mathcal D.
\label{eq:jumpRconditions}
\end{align}
\item[(c)] We have
\begin{align}\label{asymp for R}
R(\zeta;x,t) = I + \frac{1}{\zeta}R_{1}(x,t) + \bigO(\zeta^{-2}), \qquad \mbox{as } \zeta \to \infty.
\end{align}
\item[(d)] As $\zeta\to\alpha(xt)+\rho$ and as $\zeta\to\alpha(xt)\pm i\rho$, $R(\zeta;x,t)=\bigO(1)$.
\end{itemize}
\begin{figure}[htbp]
\centering
\hspace{1cm}
\begin{tikzpicture}[scale=2]
\draw[very thick,postaction={decorate,decoration={markings,mark=at position .5 with {\arrow[line width=1pt ]{>}}}}] (1.05,0) -- (2.5,0);
\draw[very thick,postaction={decorate,decoration={markings,mark=at position .5 with {\arrow[line width=1pt ]{<}}}}] (.7,.35) -- (.7,1.3);
\draw[very thick,postaction={decorate,decoration={markings,mark=at position .5 with {\arrow[line width=1pt ]{<}}}}] (.7,-.35) -- (.7,-1.3);
\draw [very thick,fill=white,inner sep=1pt,postaction={decorate,decoration={markings,mark=at position .15  with {\arrow[line width=1pt ]{<}},markings,mark=at position .6  with {\arrow[line width=1pt ]{<}},markings,mark=at position .9  with {\arrow[line width=1pt ]{<}}}}] (.7,0) circle (.35);
\node at (.9,1.2) {$\Gamma_+$};
\node at (.9,-1.2) {$\Gamma_-$};
\node at (2.4,.15) {$\Gamma_0$};
\node at (.3,.3) {$\partial \mathcal D$};
\draw[->] (-1,0) -- (2.8,0);
\draw[->] (0,-1.5) -- (0,1.5);
\draw[fill=black] (.7,0) circle (.02) node[above]{\footnotesize{${\alpha}(xt)$}};
\end{tikzpicture}
\caption{Jump contour $\Gamma_R=\partial \mathcal D \cup \Gamma_{-}\cup\Gamma_0\cup\Gamma_{+}$ for $R$.}
\label{figcontourR}
\end{figure}
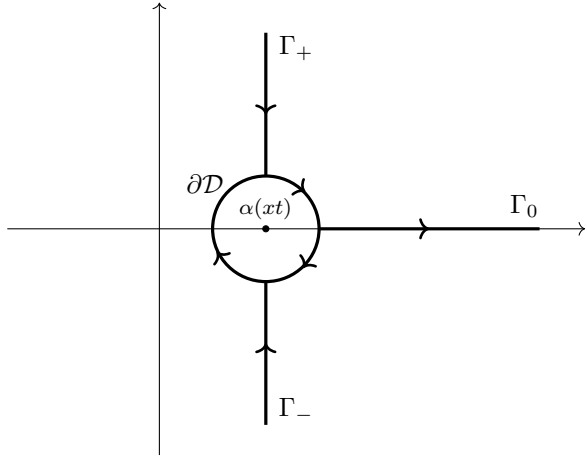

In particular, by construction of the global parametrix $P^{(\infty)}$ there is no jump on $(-\infty,{\alpha}(xt))$, and by construction of the local parametrix $P^{({\alpha})}$ there is no jump inside $\mathcal D$.
Furthermore, since both $S$ and $P^{(\alpha)}$ remain bounded near $\alpha(xt)$, $R$ has no pole at $\alpha(xt)$.

We now show that the jump for $R$ is close to the identity in the appropriate norms, namely that $R$ satisfies a \emph{small norm} RH problem.
To this end we introduce the matrix function $J_R:\Gamma_R\to\C^{2\times 2}$, piece-wise defined according to \eqref{eq:jumpRconditions1}--\eqref{eq:jumpRconditions} so that $R_+=R_-J_R$ on $\Gamma_R$.

\begin{lemma}
We have $\|J_R-I\|_{p} = \mathcal O(x^{-2})$ for $p=1,2,\infty$, uniformly in $K\leq x\leq\delta/t$; here $\|\cdot\|_p$ denotes the norm in $L^p(\Gamma_R,\C^{2\times 2})$ with respect to any matrix norm on $\C^{2\times 2}$.
\end{lemma}
\begin{proof}
By construction of the global parametrix $P^{(\infty)}$ we see that $P^{(\infty)}(\zeta;x,t)$ and $P^{(\infty)}(\zeta;x,t)^{-1}$ are $\bigO(\sqrt[4]{|\zeta|+1})$ for $|\zeta-{\alpha}(xt)|>\rho$, uniformly in $0<xt\leq \delta$.
It follows from this fact and from Corollary~\ref{corollary:smalljumpS} that there exists $\eta>0$ such that $J_R(\zeta)=I+\bigO(\sqrt{|\zeta|+1}\,e^{-\eta x^2\sqrt{|\zeta-{\alpha}(xt)|}})$ for $|\zeta-{\alpha}(xt)|>\rho$, uniformly in $0<xt\leq \delta$.
Up to possibly choosing a smaller $\eta>0$, it follows that the $L^p(\Gamma_R,\C^{2\times 2})$-norms of $J_R-I$ for $p=1,2,\infty$ are $\bigO(e^{-\eta x^2\sqrt{\rho}})$ uniformly in $0<xt\leq \delta$.
On the remaining part of the contour we have, as $x\to+\infty$,
\begin{align}
\nonumber
J_R(\zeta;x,t)-I&=P^{({\alpha})}(\zeta;x,t)P^{(\infty)}(\zeta;x,t)^{-1} -I
\\
\label{eq:estimatejumpRsec3}
&= \frac{P^{(\infty)}(\zeta;x,t)e^{\frac{W(\zeta;x)}{2}\sigma_{3}} \Phi_{\mathrm{Ai},1} e^{-\frac{W(\zeta;x)}{2}\sigma_{3}}P^{(\infty)}(\zeta;x,t)^{-1} }{f(\zeta;xt)^{3/2}}x^{-2}+ \bigO(x^{-4}),
\end{align}
uniformly for $\zeta \in \partial \mathcal D$.
Here we use \eqref{Asymptotics Airy}; for instance $\Phi_{\mathrm{Ai},1}=\frac 18\begin{pmatrix}
\frac 16 & i \\ i & -\frac 16
\end{pmatrix}$, see Appendix \ref{app:airy}.
It remains to note that \eqref{eq:estimatejumpRsec3} is uniform for $x$ large and $xt\leq \delta$.
First, $f$ is bounded away from zero, uniformly for {the relevant} values of $x,t$, {and $\zeta$}; this is seen by inspection of the series \eqref{eq:seriesphi} and by the definition of $f$ \eqref{def of f and E}.
Similarly, by Propositions \ref{prop:d1} and \ref{prop:chi} we can show that $P^{(\infty)}(\zeta;x,t)e^{\frac{W(\zeta;x)}{2}\sigma_{3}} $ is bounded for $\zeta\in\partial\mathcal D$, uniformly for large $x$ and $xt\leq \delta$.
\end{proof}
It should be noted that the contour $\Gamma_R$ depends on {$xt$} through $\alpha(xt)$.
This is not so convenient to apply 
the standard theory of small-norm RH problems; however one observes easily that the jump contour is independent of $x$ and $t$ after the uniformly bounded shift $\widetilde\zeta=\zeta-{\alpha}(xt)$, so that the standard theory applies to the RH problem in that variable. Back in our variable $\zeta$,
we obtain that, as $x \to + \infty$,
\begin{align}\label{asymp for R easy sector}
R(\zeta;x,t) = I + \frac{1}{x^{2}}R^{(1)}(\zeta;x,t) + \bigO(x^{-4}),
\qquad R^{(1)}(\zeta;x,t) = \frac{1}{2\pi i}\int_{\partial \mathcal D} \frac{J_{R}^{(1)}(s;x,t)}{s-\zeta}ds
\end{align}
uniformly in $\zeta$, where (see \eqref{eq:estimatejumpRsec3})
\begin{align}
J_{R}^{(1)}(\zeta;x,t) := \frac{1}{f(\zeta;xt)^{3/2}}P^{(\infty)}(\zeta;x,t)e^{\frac{W(\zeta;x)}{2}\sigma_{3}} \Phi_{\mathrm{Ai},1}e^{-\frac{W(\zeta;x)}{2}\sigma_{3}} P^{(\infty)}(\zeta;x,t)^{-1}.
\end{align}
By the relations
\begin{align*}
&(f^{3/2}(\zeta;xt))_+=-(f^{3/2}(\zeta;xt))_-,\qquad
\begin{pmatrix} 0 & 1 \\ -1 & 0 \end{pmatrix}\Phi_{\mathrm{Ai},1}\begin{pmatrix} 0 & -1 \\ 1 & 0 \end{pmatrix}=-\Phi_{\mathrm{Ai},1},
\\
&(P^{(\infty)}(\zeta;x,t)e^{\frac{W(\zeta;x)}{2}\sigma_{3}})_+=(P^{(\infty)}(\zeta;x,t)e^{\frac{W(\zeta;x)}{2}\sigma_{3}})_-\begin{pmatrix} 0 & 1 \\ -1 & 0 \end{pmatrix},
\end{align*}
we infer that $J_{R}^{(1)}(\zeta;x,t)$ has no jump for ${\alpha}(xt)-\rho<\zeta<{\alpha}(xt)$, hence $J_R^{(1)}(\zeta;x,t)$ is meromorphic for $\zeta\in\mathcal D$ with a pole at $\zeta={\alpha}(xt)$ only, which is readily seen to be of second order.
Thus we can compute $R^{(1)}$ in \eqref{asymp for R easy sector} for $\zeta\in\C\setminus\overline{\mathcal D}$ by extracting the polar part of $J^{(1)}_R$ at $\zeta={\alpha}(xt)$ (note that $\partial\mathcal D$ is oriented clockwise); namely, if $|\zeta-\alpha(xt)|>\rho$,
\begin{align}
\nonumber
& R^{(1)}(\zeta;x,t) ={} \frac{5}{48f_1^{3/2}} \begin{pmatrix}
-d_{1} & id_{1}^{2} \\ i & d_{1}
\end{pmatrix}\frac 1{(\zeta-{\alpha}(xt))^2} 
\\
\label{eq:Rp1p}
&+ \frac{1}{32 f_{1}^{5/2}}\begin{pmatrix}
-8\chi f_1(1+d_{1}\chi)+5d_{1}f_{2} & \frac{i}{3}\big( 14f_1+48 d_{1}\chi f_1 + 24d_{1}^{2}\chi^{2}f_1-15d_1^2f_{2} \big) \\
i \big( 8\chi^{2}f_1-5f_{2} \big) & 8\chi f_1(1+d_{1}\chi)-5d_{1}f_{2}
\end{pmatrix}\frac 1{\zeta-{\alpha}(xt)},
\end{align}
where $f_1=f_1(xt),f_2=f_2(xt)$ are given in \eqref{eq:fTaylor}, $d_1=d_1(x,t)$ in \eqref{eq:di}, and $\chi=\chi(x,t)$ in \eqref{eq:chi}.

Since \eqref{asymp for R easy sector} is uniform in $\zeta$ we also obtain
\begin{equation}
\label{R1inf}
R_1(x,t) =\frac{1}{32 x^2 f_{1}^{3/2}}\begin{pmatrix}
-8\chi f_1(1+d_{1}\chi)+5d_{1}f_{2} & \frac{i}{3}\big( 14f_1+48 d_{1}\chi f_1 + 24d_{1}^{2}\chi^{2}f_1-15d_1^2f_{2} \big) \\
i \big( 8\chi^{2}f_1-5f_{2} \big) & 8\chi f_1(1+d_{1}\chi)-5d_{1}f_{2}
\end{pmatrix}+\bigO(x^{-4}).
\end{equation}

\subsection{Asymptotics for $u_\sigma(x,t)$ and $\partial_x\log Q_\sigma(x,t)$}

\begin{proposition}
\label{prop:asympusect3}
As $x\to+\infty$ we have, uniformly in $0<xt\leq \delta$ for any $\delta>0$,
\begin{align*}
u_\sigma(x,t) = \frac{x^2\alpha(xt)}{2}-\frac{\log c_+'}{2\pi x t\sqrt{\alpha(xt)}+{2}c_+}+\bigO(x^{-2}),
\end{align*}
or equivalently
\begin{align*}
u_\sigma(x,t) = \frac{x}{2t}a_0\left(\frac{\pi^2}{c_+^2}xt\right)+\frac{1}{2\sqrt{xt}}a_1\left(\frac{\pi^2}{c_+^2}xt\right)+\bigO(x^{-2}),
\end{align*}
with $a_0$ and $a_1$ as in Theorem \ref{theorem:main}.
\end{proposition}
\begin{proof}
Combining \eqref{eq:upq} {and \eqref{eq:T1}}, we obtain the identity
\begin{equation}
u_\sigma(x,t) = x^2(T_{1,21}(x,t)^2-2T_{1,11}(x,t)).
\end{equation}
For $\zeta$ satisfying $\re\zeta>\alpha(xt)+\rho$, using \eqref{eq:StoT} {and \eqref{def of R section 3}}, we have
\begin{align}
T(\zeta;x,t)=S(\zeta;x,t)=R(\zeta;x,t)P^{(\infty)}(\zeta;x,t).
\end{align}
Substituting the large $\zeta$ asymptotics \eqref{asT} for $T$, \eqref{asymp for R} for $R$, and \eqref{asymp for Pinf} for $P^{(\infty)}$, we get
\begin{align}
\label{eq:TRP}
T_1(x,t)=P_1^{(\infty)}(x,t)+R_1(x,t).
\end{align}
We then complete the proof using \eqref{P1inf}, \eqref{R1inf}, and Proposition \ref{prop:chi}.
\end{proof}

\begin{proposition}
We have the identity
\begin{align}
\label{eq:finaldiffid}
\partial_x\log Q_\sigma(x,t)=-\frac{x^2}{4t}-x^3g_1(xt)+xd_1(x,t)-ix R_{1,21}(x,t),
\end{align}
where $g_1(xt)$ is given in \eqref{def:g1}, $d_1(x,t)$ in \eqref{eq:di}, and $R_1(x,t)$ in \eqref{asymp for R}.
\end{proposition}
\begin{proof}
Combining \eqref{xdiffid} and \eqref{eq:T1}, we obtain the identity
\begin{align}
\partial_x\log Q_\sigma(x,t)=-\frac{x^2}{4t}-x^3g_1(xt)-ixT_{1,21}(x,t).
\end{align}
By \eqref{eq:TRP} we get
\begin{align}
\partial_x\log Q_\sigma(x,t)=-\frac{x^2}{4t}-x^3g_1(xt)-ixP_{1,21}^{(\infty)}(x,t)-ixR_{1,21}(x,t),
\end{align}
and it suffices to use \eqref{P1inf} to obtain the result.
\end{proof}

We now compute {the} large $x$ asymptotics for the above differential identity, and we write the asymptotics as a derivative, so that we will be able to easily integrate it later.

\begin{corollary}\label{corollary:Qsec3}
As $x\to+\infty$ we have, uniformly in $0<xt\leq \delta$ for any $\delta>0$,
\begin{equation}
\label{eq:finalsec2}
\partial_x\log Q_\sigma(x,t)=\partial_x\left[-\frac{c_+^6}{\pi ^6t^4}F_1\left(\frac{\pi^2}{c_+^2}xt\right)-\frac{c_+^3\log c'_+}{\pi^4t^2}F_2\left(\frac{\pi^2}{c_+^2}xt\right)+F_3\left(\frac{\pi^2}{c_+^2}xt\right)\right]+\bigO(x^{-3})
\end{equation}
{with $F_{1}$, $F_{2}$, and $F_{3}$ as in Theorem \ref{theorem:main}.}
\end{corollary}
\begin{proof}
By \eqref{def:g1} and \eqref{def:a} it is straightforward to check that
\begin{align}
\frac{x^2}{4t}+x^3g_1(xt) = \partial_x\left[\frac{c_+^6}{\pi ^6t^4}F_1\left(\frac{\pi^2}{c_+^2}xt\right)\right].
\end{align}
Moreover, by Proposition \ref{prop:d1} and an explicit integration, as $x\to+\infty$ uniformly in $0<xt\leq\delta$, we have
\begin{align}
xd_1(x,t) = \partial_x\left[-\frac{c_+^3\log c'_+}{\pi^4t^2}F_2\left(\frac{\pi^2}{c_+^2}xt\right)+\frac{2c_+j_\sigma}\pi\left(\log \left(\sqrt{1+\frac{\pi^2}{c_+^2}xt}-1\right)+\sqrt{1+\frac{\pi^2}{c_+^2}xt}\right)\right]+\bigO(x^{-3}). 
\end{align}
Next, using \eqref{R1inf} we get
\begin{align}
\label{eq:R1chif1f2}
-i x R_{1,21}(x,t) =\frac {\chi^2(x,t)}{4xf_1(xt)^{3/2}}
-\frac{5f_2(xt)}{32xf_1(xt)^{5/2}} {+\bigO(x^{-3})}.
\end{align}
It follows from Proposition \ref{prop:chi} that, as $x\to+\infty$, we have
\begin{align}
-i x R_{1,21}(x,t)=\frac {\log^2c_+'}{4\pi^2x{\alpha}(xt)f_1(xt)^{3/2}}
-\frac{5f_2(xt)}{32xf_1(xt)^{5/2}}+\bigO(x^{-3}),
\end{align}
which can be explicitly integrated as
\begin{align}
-i x R_{1,21}(x,t) = \partial_x\left[\left(\frac{\log ^2c'_+}{2\pi^2}+\frac 1{24}\right) \log\left(\sqrt{1+\frac{\pi^2}{c_+^2}xt}-1\right)-\frac 1{48}\log\left(1+\frac{\pi^2}{c_+^2}xt\right)\right]+\bigO(x^{-3}).
\end{align}
The proof is complete by \eqref{eq:finaldiffid}.
\end{proof}

\section{Asymptotic analysis for $\Psi$ if $x \geq K, \, x\geq \delta/t${, $t \leq t_{0}$}}\label{section:4}
In this section, we analyze the solution to the RH problem for $\Psi$ asymptotically as $x\to\infty$ and $xt$ bounded from below. This the same regime as the one studied in \cite{CaCl2019} for $\sigma(r)=\frac{1}{1+e^{-r}}$. Here, we extend the analysis of~\cite{CaCl2019} to any $\sigma$ satisfying Assumptions \ref{assumptions}, and we compute further sub-leading terms in the asymptotic expansion.

Throughout this section, we will construct a number of auxiliary functions which will carry the same names $g,\phi, T, S, R$ as their natural counterparts in Section \ref{section:3}.
We emphasize however that these functions are not identical to the ones constructed in Section \ref{section:3}.

\medskip

\subsection{Rescaling}
In this section, both $x$ and $xt$ are large while $t$ remains bounded, and we therefore need another change of variables than in Subsection \ref{subsection:rescaling section 3}.
Instead of working in the $z$-plane, it will be more convenient to make the change of variable $w=\frac{t}{x} z$ and to work in the $w$-plane. {Indeed, this will transform the functions $e^{\pm (-\frac{2}{3}tz^{\frac{3}{2}}+xz^{\frac{1}{2}})}$, which appear in the large $z$ asymptotics of $\Psi(z;x,t)$, into $e^{\pm x^{3/2}t^{-1/2}(-\frac{2}{3}w^{\frac{3}{2}}+w^{\frac{1}{2}})}$.}
We define
\begin{align*}
\widetilde{\Psi}(w;x,t) := \bigg(\frac{t}{x}\bigg)^{\frac{\sigma_{3}}{4}} \Psi(z=w \frac{x}{t}; x,t).
\end{align*}
Then, the RH problem for $\Psi$ transforms into the following conditions for $\widetilde\Psi$.
\subsubsection*{RH problem for $\widetilde{\Psi}$}
\begin{itemize}
\item[(a)] $\widetilde{\Psi} : \mathbb{C}\setminus \mathbb{R} \to \mathbb{C}^{2\times 2}$ is analytic. 
\item[(b)] $\widetilde{\Psi}$ has continuous boundary values on $\mathbb{R}$, and they satisfy the jump relation
\begin{align*}
\widetilde{\Psi}_{+}(w;x,t) = \widetilde{\Psi}_{-}(w;x,t) \begin{pmatrix}
1 & \frac{1}{F( \tfrac{x}{t} w)} \\
0 & 1
\end{pmatrix}, \qquad w \in \mathbb{R},
\end{align*}
where we recall that $F=\frac{1}{1-\sigma}$.
\item[(c)] As $w \to \infty$, we have
\begin{multline*}
\widetilde{\Psi}(w;x,t) = \bigg( I + \frac{1}{w}\widetilde{\Psi}_{1}(x,t) + \bigO(w^{-2}) \bigg) w^{\frac{1}{4}\sigma_{3}}A^{-1}e^{x^{3/2}t^{-1/2}(-\frac{2}{3}w^{\frac{3}{2}}+w^{\frac{1}{2}})\sigma_{3}}
\\
{\times} \begin{cases}
I, & |\arg w | < \pi -\varepsilon, \\
\begin{pmatrix}
1 & 0 \\
\mp 1 & 1
\end{pmatrix}, & \pi - \varepsilon < \pm \arg w < \pi,
\end{cases}
\end{multline*}
for any $\varepsilon \in (0,\frac{\pi}{2})$, where 
\begin{align}
\widetilde{\Psi}_{1}(x,t) = \frac{t}{x} \begin{pmatrix}
q(x,t) & -i t^{1/2}x^{-1/2}r(x,t) \\
i t^{-1/2}x^{1/2} p(x,t) & -q(x,t)
\end{pmatrix}.\label{eq:Psi1tildepq}
\end{align}
\end{itemize}
\subsection{Construction of the $g$-function}
Similarly as in \cite{CaCl2019}, the key step in the asymptotic analysis of the RH problem is the construction of a $g$-function, which satisfies the following conditions.

\subsubsection*{Conditions for $g$}
\begin{itemize}
\item[(a)] $g(w;x,t)$ is analytic in $w\in\mathbb C\setminus(-\infty,a(x,t)]$,
\item[(b)] $g_+(w;x,t)+g_-(w;x,t)+\frac{t^{1/2}}{x^{3/2}}V_0(x,t)=\frac{t^{1/2}}{x^{3/2}}\log F(\frac{x}{t} w)$ for $w\in(-\infty,a(x,t))$,
\item[(c)] $g(w;x,t)=-\frac{2}{3}w^{\frac{3}{2}}+w^{\frac{1}{2}}+g_{0}(x,t)+g_1(x,t) w^{-1/2}+\mathcal O(w^{-3/2})$ as $w\to\infty$,
\item[(d)] $g_{+}(w;x,t)-g_{-}(w;x,t) = \bigO((w-a(x,t))^{3/2})$ as $w \to a(x,t)$, $w<a(x,t)$.
\end{itemize}
Here, $a(x,t)>0$, $g_0(x,t)$, $g_1(x,t)$ and $V_{0}=V_0(x,t)$ are to be determined.
Similarly as in Section \ref{section:3}, it will be convenient to first construct $g'$, which satisfies the conditions below, and then to integrate.
 \subsubsection*{Conditions for $g'$}
\begin{itemize}
\item[(a)] $g'(w;x,t)$ is analytic in $w\in\mathbb C\setminus(-\infty,a(x,t)]$,
\item[(b)]\label{bforgprime} $g_+'(w;x,t)+g_-'(w;x,t)=\frac{1}{\sqrt{xt}} (\log F)'(\frac{x}{t} w)$ for $w\in(-\infty,a(x,t))$,
\item[(c)] $g'(w;x,t)=-w^{1/2}+\frac{1}{2}w^{-\frac{1}{2}}-\frac 12 g_1(x,t)w^{-\frac{3}{2}}+\mathcal O(w^{-5/2})$ as $w\to\infty$, 
\item[(d)] $g_{+}'(w;x,t)-g_{-}'(w;x,t) = \bigO((w-a(x,t))^{1/2})$ as $w \to a(x,t)$, $w<a(x,t)$.
\end{itemize}
We can construct $g'$ directly in terms of a Cauchy integral: set
\begin{align}
g'(w;x,t)&=\frac{a(x,t)+1-2w}{2(w-a(x,t))^{1/2}}+\frac{1}{2\pi (w-a(x,t))^{1/2}} \frac{1}{\sqrt{xt}} \int_{-\infty}^{a(x,t)} (\log F)'(\tfrac{x}{t}\zeta) \sqrt{a(x,t)-\zeta}\frac{d\zeta}{\zeta-w}, \label{gprime 1 CC}
\end{align}
with $(w-a(x,t))^{1/2}$ analytic in $\mathbb C\setminus(-\infty,a(x,t)]$ and positive for $w>a(x,t)$, and where $\sqrt{a(x,t)-\zeta}$ denotes the positive square root.
Then, it is straightforward to verify that $g'$ satisfies conditions (a), (b), and (c) {(with $g_{1}(x,t)$ given by \eqref{eq:identityg1} below)}.
In order to achieve also the endpoint condition (d), we need to take $a(x,t)$ such that
\begin{align}\label{a def CC}
\int_{-\infty}^{a(x,t)}(\log F)'(\tfrac{x}{t} \zeta)\frac{d\zeta}{\sqrt{a(x,t)-\zeta}}=\pi \sqrt{xt} (1-a(x,t)).
\end{align}
For later convenience we study in detail the behavior of the endpoint $a(x,t)$.
\begin{proposition}
\label{prop:endpointCC}
For any $\sigma$ satisfying Assumptions \ref{assumptions} and for any $x,t>0$ there exists a unique real solution $a(x,t)$ to the equation \eqref{a def CC}.
There exists $K>0$, depending on $\sigma$ only, such that for all $\delta>0$ there exist $a_{\rm min},a_{\rm max}>0$, depending on $\delta$ and $\sigma$ only, such that $a_{\rm min}<a(x,t)<a_{\rm max}$ for all $x,t$ satisfying $x\geq K,xt\geq\delta$.
Moreover, for any $\delta>0$ {and $t_{0}>0$} we have
\begin{align}
\label{def:ainfty}
a(x,t)=a_0\left(\frac{\pi^2}{c_+^2}xt\right)
+\frac{t^{1/2}}{x^{3/2}}a_1\left(\frac{\pi^2}{c_+^2}xt\right)
+\frac{t^{3/2}}{x^{5/2}}a_2\left(\frac{\pi^2}{c_+^2}xt\right)
+\bigO(t^{5/2}x^{-7/2}),
\end{align}
as $x\to +\infty$, uniformly in $xt\geq\delta$ {and $t \leq t_{0}$}, where {$a_{0}$, $a_{1}$ and $a_{2}$ are as in Theorem \ref{theorem:main}.}
\end{proposition}
\begin{remark}
\label{rem:aalpha}
The leading order in \eqref{def:ainfty} satisfies $a_0(\frac{\pi^2}{c_+^2}xt)=xt\alpha(xt)$, where $\alpha(xt)$ is the endpoint for the RH analysis of Section \ref{section:3}, see \eqref{def:a}.
\end{remark}
\begin{proof}
Let us write the endpoint equation \eqref{a def CC} as
\begin{equation}
h(a;x/t)= \pi \sqrt{xt}(1-a),\qquad 
h(a;\lambda) := \int_{-\infty}^{a}(\log F)'(\lambda \zeta)\frac{d \zeta}{\sqrt{a-\zeta}}.
\end{equation}
Since $F$ is log-convex (see Assumptions \ref{assumptions}), the function $h=h(a;x/t)$ is increasing in $a$, while {$a \mapsto \pi \sqrt{xt}(1-a)$} is a decreasing affine function of $a$, thus there exists exactly one real solution $a=a(x,t)$, proving the first statement.
Next, let us define $v:=(\log F)'-c_+1_{(0,+\infty)}$.
By Assumptions \ref{assumptions}, see in particular \eqref{eq:vplusinfty} and \eqref{eq:vminusinfty}, $v\in L^1(\R)\cap L^\infty(\R)$ so that for all $a>0$ we have
\begin{align*}
\left|h(a;{x/t})-2c_+\sqrt{a}\right| = 
\left|\int_{-\infty}^a v\left(\frac{x}{t} \zeta\right)\frac{d\zeta}{\sqrt{a-\zeta}}\right|
&\leq\sqrt{\frac{t}{x}}\int_{-\infty}^{\frac{x}{t}a} |v(r)|\frac{dr}{\sqrt{\frac{x}{t}a-r}}\\
&\leq \sqrt{\frac{t}{x}}\int_{-\infty}^{\frac{x}{t}a-1}|v(r)|dr +\sqrt{\frac{t}{x}}\,\|v\|_\infty\int_{-1}^0\frac{ds}{\sqrt{-s}}\\
&\leq (\|v\|_1+2\|v\|_\infty)\sqrt{\frac tx}.
\end{align*}
Setting $M:=\|v\|_1+2\|v\|_\infty$ it then follows by the inequalities $2c_+\sqrt{a}-M\sqrt{t/x} \leq h(a;{x/t}) \leq 2c_+\sqrt{a}+M\sqrt{t/x}$, which we have just shown, that $a_-(x,t) \leq a(x,t) \leq a_+(x,t)$ where $a_\pm(x,t)$ are the solutions to
\begin{equation*}
2c_+\sqrt{a_\pm}\mp M\sqrt{\frac tx}=\pi\sqrt{xt}(1-a_\pm)
\ \Rightarrow\
a_\pm(x,t) = 1+\frac{2 c_+^2}{\pi ^2 t x}\pm \frac{M}{\pi  x}-\frac{2 c_+ \sqrt{c_+^2+\pi  t (\pi  x\pm M)}}{\pi ^2 t x}.
\end{equation*}
Here we assume $x>M/\pi$.
It is straightforward to check that $a_-(x,y/x)$ is an increasing function of both $x$ and $y$ provided that $x>M/\pi$, and that $a_-(M/\pi,t)=0$ for all $t>0$. Fixing $K>M/\pi$, for all $x,t>0$ such that $x\geq K,xt\geq\delta$ we have $a(x,t) \geq a_{\rm min}:=a_-(K,\delta/K)$.
A completely analogous argument proves the existence of $a_{\rm max}$ such that $a(x,t) \leq a_{\rm max}$ for all $x,t>0$ such that $x\geq K,xt\geq\delta$. (In fact, we can set $a_{\rm max}:=1+\frac M{\pi K}$.)

\def\wh{\widehat}

Next, we claim that $h(a;\lambda)$ has an expansion
\begin{equation}
\label{eq:claim}
h(a;\lambda) = h_0(a)+\lambda^{-1}h_1(a)+\lambda^{-2}h_2(a)+\bigO(\lambda^{-3}),\qquad\lambda\to+\infty,
\end{equation}
uniform for $a$ in compact subsets of $(0,+\infty)$, where
\begin{equation}
h_0(a)=2c_+\sqrt a,\quad h_1(a)=\frac{\log c_+'}{\sqrt a},\quad h_2(a)=\frac {\pi j_\sigma}{a^{3/2}},
\end{equation}
where $j_\sigma$ is defined in \eqref{eq:Isigma}.
Postponing for a moment the proof of \eqref{eq:claim}, let us note that the Implicit Function Theorem then implies an expansion of the form
\begin{equation}
a(x,t) = \wh a_0(xt)+\frac 1{\sqrt{xt}}\left(\frac tx \wh a_1(xt) + \frac {t^2}{x^2}\wh a_2(xt)+\bigO(t^3x^{-3})\right)
\end{equation}
as $x\to+\infty$, where $\wh a_i(xt)$ are found from the equations
\begin{align*}
h_0(\wh a_0(xt)) &=\pi\sqrt{xt}(1-\wh a_0(xt))
\\
h_1(\wh a_0(xt))+\frac{\wh a_1(xt)}{\sqrt{xt}}h_0'(\wh a_0(xt))&= -\pi \wh a_1(xt)
\\
h_2(\wh a_0(xt))+\frac{\wh a_1(xt)}{\sqrt {xt}}h_1'(\wh a_0(xt))+\frac 12\frac{\wh a_1(xt)^2}{xt}h_0''(\wh a_0(xt))+\frac{\wh a_2(xt)}{\sqrt{xt}} h_0'(\wh a_0(xt))&= -\pi \wh a_2(xt)
\end{align*}
and they are in particular bounded for $xt\geq\delta$, and the $\bigO(t^{5/2}x^{-7/2})$ is also easily checked to be uniform for $xt\geq\delta$ and $t \leq t_{0}$ for any $\delta>0$ and $t_{0}>0$.
These equation{s} imply
\begin{align}
\nonumber
\wh a_0(xt)&=\left(\sqrt{1+\frac{c_+^2}{\pi^2 xt}}-\frac{c_+}{\pi\sqrt{xt}}\right)^2,\qquad
\wh a_1(xt)=-\frac{\log c_+'}{\pi\sqrt{\wh a_0(xt)}+\frac{c_+}{\sqrt{xt}}},\\
\nonumber
\wh a_2(xt)&=\frac{c_+\wh a_1(xt)^2+2 \wh a_1(xt) {\sqrt{xt}} \log c_+'-4 \pi {xt} j_\sigma}{4 \wh a_0(xt)\left(\pi\sqrt{\wh a_0(xt)}+ \frac{c_+}{\sqrt{xt}}\right){xt}},
\end{align}
and then \eqref{def:ainfty} is obtained by simple algebra. 

It only remains to prove \eqref{eq:claim}; to this end we write
\begin{equation}
h(a;\lambda) - 2c_+\sqrt a=\int_{-\infty}^a v(\lambda\zeta)\frac{d\zeta}{\sqrt{a-\zeta}} = \frac 1\lambda \left[\int_{-\infty}^{-a\lambda /2}\frac{v(r)\,dr}{\sqrt{a-\frac r\lambda}}+\int_{-a\lambda/2}^{a\lambda/2}\frac{v(r)\,dr}{\sqrt{a-\frac r\lambda}}+\int_{a\lambda/2}^{a\lambda}\frac{v(r)\,dr}{\sqrt{a-\frac r\lambda}} \right]
\end{equation}
and we reason for each term separately.
For the first one we have
\begin{equation}
\left|\int_{-\infty}^{-a\lambda/2}\frac{v(r)\,dr}{\sqrt{a-\frac r\lambda}}\right| \leq \sqrt{\frac 2{3a}}\left(\sup_{r<-{\frac{a \lambda}{2}}}|v(r)|^{1/2}\right)\int_{-\infty}^{-a\lambda/2}|v(r)|^{1/2}\,d\zeta
\end{equation}
which is exponentially small as $\lambda\to+\infty$ by Assumptions \ref{assumptions}, see \eqref{eq:vminusinfty}.
For the second term we use the expansion
\begin{equation}
\frac{1}{\sqrt{a-\frac r\lambda}} = \frac 1{\sqrt a}+\lambda^{-1}\frac r{2a^{3/2}}+\bigO (r^2\lambda^{-2}),\qquad\lambda\to+\infty,
\end{equation}
which is uniform for $|r|<a\lambda/2$, to obtain
\begin{equation}
\int_{-a\lambda/2}^{a\lambda/2}\frac{v(r)\,dr}{\sqrt{a-\frac r\lambda}}
=\frac 1{\sqrt a}\int_{-a\lambda/2}^{a\lambda/2}v(r)\,dr+\lambda^{-1}\frac 1{2a^{3/2}}\int_{-a\lambda/2}^{a\lambda/2}rv(r)\,dr+\bigO(\lambda^{-2})
\end{equation}
and so, since $v(r)$ has exponential decay as $r \to \pm \infty$, see \eqref{eq:vplusinfty} and \eqref{eq:vminusinfty}, we have
\begin{equation}
\int_{-a\lambda/2}^{a\lambda/2}\frac{v(r)\,dr}{\sqrt{a-\frac r\lambda}}
=\frac 1{\sqrt a}\int_{-\infty}^{+\infty}v(r)\,dr+\lambda^{-1}\frac 1{2a^{3/2}}\int_{-\infty}^{+\infty}rv(r)\,dr+\bigO(\lambda^{-2}).
\end{equation}
For the third term we estimate
\begin{equation}
\int_{a\lambda/2}^{a\lambda}\frac{v(r)\,dr}{\sqrt{a-\frac r\lambda}} \leq\lambda\left(\sup_{r>{\frac{a \lambda}{2}}}|v(r)|\right)\int_{a/2}^a\frac{d\zeta}{\sqrt{a-\zeta}}=\lambda\sqrt{2a}\left(\sup_{r>{\frac{a \lambda}{2}}}|v(r)|\right)
\end{equation}
which is exponentially small as $\lambda \to + \infty$ by \eqref{eq:vplusinfty}.
Uniformity for $a$ in compact subsets of $(0,+\infty)$ is clear in all these estimates.
The proof is complete by the following integrations by parts;
\begin{align*}
\int_{-\infty}^{+\infty}v(r)dr&=\int_{-\infty}^0(\log F)'(r)dr+\int_{{0}}^{{+\infty}}(\log F(r)-c_+r)'dr = \lim_{r\to+\infty}(\log F(r)-c_+r) = \log c_+'
\\
\int_{-\infty}^{+\infty}rv(r)dr&=\int_{-\infty}^0r\,(\log F)'(r)dr+\int_{{0}}^{{+\infty}}r\,(\log F(r)-c_+r-\log c_+')'dr 
\\
&= -\int_{-\infty}^{+\infty}\left[\log F(r)-(c_+r+\log c_+')1_{(0,+\infty)}(r)\right]\,dr = 2\pi j_\sigma
\end{align*}
where in the last equality we use \eqref{eq:Isigma}.
\end{proof}

\begin{remark}
We can construct $g'$ equivalently  as 
\begin{align}
g'(w;x,t) = -\sqrt{w-a(x,t)} \bigg(1 + \frac{1}{2\pi\sqrt{xt}}\int_{-\infty}^{a(x,t)}  (\log F)'(\tfrac{x}{t} \zeta) \frac{1}{\sqrt{a(x,t)-\zeta}} \frac{d\zeta}{\zeta-w} \bigg), \label{gprime 2 CC}
\end{align}
where $a(x,t)$ is found by requiring that $g'(w;x,t) = -w^{1/2}+\frac{1}{2}w^{-\frac{1}{2}} + \bigO(w^{-\frac{3}{2}})$ as $w \to \infty$. Here too, we see that $a(x,t)$ must satisfy \eqref{a def CC}. Using \eqref{a def CC}, we easily show that the right-hand sides of \eqref{gprime 1 CC} and \eqref{gprime 2 CC} are indeed equal.
\end{remark}
Let us define
\begin{align*}
g(w;x,t)=\int_{a(x,t)}^wg'(s;x,t)ds.
\end{align*}
Since $g(a(x,t);x,t)=0$, we must choose 
\begin{align*}
V_0(x,t)=-\log(1-\sigma(\tfrac{x}{t}a(x,t)))= \log F(\tfrac{x}{t}a(x,t)).
\end{align*}
By \eqref{gprime 1 CC}, since $(\log F)'$ decays rapidly at $-\infty$, it is straightforward to check that the coefficient $g_{1}$ in the large $w$ expansion of $g$ and $g'$ is given by
\begin{align}
g_1(x,t)&=\frac{a(x,t)^2}4-\frac{a(x,t)}2+\frac{1}{\pi \sqrt{xt}} \int_{-\infty}^{a(x,t)}(\log F)'(\tfrac{x}{t}\zeta)\sqrt{a(x,t)-\zeta}\,d\zeta
\label{eq:identityg1}
\end{align}
Also, in view of condition (b) for $g$, $g_{0}(x,t)$ is given by
\begin{align*}
g_{0}(x,t) = - \frac{t^{1/2}}{2x^{3/2}}V_0(x,t).
\end{align*}

\subsection{Normalization of the RH problem}
Let us define 
\begin{align*}
T(w;x,t) := \begin{pmatrix}
1 & ig_{1}(x,t)x^{3/2}t^{-1/2} \\
0 & 1
\end{pmatrix} \widetilde{\Psi}(w;x,t)e^{-x^{3/2}t^{-1/2}(g(w;x,t)-g_{0}(x,t))\sigma_3},
\end{align*}
and
\begin{align}\label{def of phi CC}
\phi(w;x,t):=2g(w;x,t)+\frac{t^{1/2}}{x^{3/2}}V_0(x,t)-\frac{t^{1/2}}{x^{3/2}}\log F(\tfrac{x}{t}w).
\end{align}
This function satisfies
\begin{align*}
\phi_\pm(w;x,t)=\pm \left(g_+(w;x,t)- g_-(w;x,t)\right),\qquad w\in(-\infty,a(x,t)).
\end{align*}
Then, $T$ satisfies the following RH problem. 
\subsubsection*{RH problem for $T$}
\begin{itemize}
\item[(a)] $T: \mathbb{C}\setminus  \mathbb{R} \to \mathbb{C}^{2\times 2}$ is analytic.
\item[(b)] $T$ has the jumps
\begin{align*}
& T_{+}(w;x,t) = T_{-}(w;x,t) \begin{pmatrix}
e^{-x^{3/2}t^{-1/2}\phi_+(w;x,t)} & 1 \\
0 & e^{-x^{3/2}t^{-1/2}\phi_-(w;x,t)}
\end{pmatrix}, & & w \in (-\infty,a(x,t)), \\
& T_{+}(w;x,t) = T_{-}(w;x,t) \begin{pmatrix}
1 & e^{x^{3/2}t^{-1/2} \phi(w;x,t)} \\
0 & 1
\end{pmatrix}, & & w\in (a(x,t),+\infty).
\end{align*}
\item[(c)] We have
\begin{equation}
\label{eq:asTCC}
T(w;x,t) = \bigg( I + \frac{1}{w}T_{1}(x,t) + \bigO(w^{-2}) \bigg) w^{\frac{1}{4}\sigma_{3}}A^{-1}, \qquad \mbox{as } w \to \infty,
\end{equation}
where
\begin{equation}
\label{eq:T1CC}
T_1(x,t) = 
\begin{pmatrix}
 -\frac{x^3}{2 t}g_1^2 - xg_1 p+\frac{t}{x}q & 
 i\frac{x^{5/2}}{t^{1/2}} g_1^2 p-2 i (xt)^{1/2} g_1 q+i\frac{x^{9/2}}{3 t^{3/2}} g_1^3-i\frac{x^{3/2}}{t^{1/2}} g_2-i\left(\frac{t}{x}\right)^{3/2}r\\
 i \frac{x^{3/2} }{t^{1/2}}g_1+i \frac{t^{1/2}}{x^{1/2}} p &  \frac{x^3}{2 t}g_1^2 + xg_1 p-\frac{t}{x}q  \\
\end{pmatrix},
\end{equation}
where $g_i=g_i(x,t)$, $p=p(x,t)$, $q=q(x,t)$, and $r=r(x,t)$.
\item[(d)] As $w\to a(x,t)$, $T(w;x,t)=\bigO(1)$.
\end{itemize}

\subsection{Opening of the lenses}
For $w \in (-\infty,a(x,t))$, note that the jump matrix can be factorized as follows:
\begin{multline}\label{factorization of the jump CC}
\begin{pmatrix}
e^{-x^{3/2}t^{-1/2} \phi_{+}(w;x,t)} & 1\\
0 & e^{-x^{3/2}t^{-1/2} \phi_{-}(w;x,t)}
\end{pmatrix} = \begin{pmatrix}
1 & 0 \\
e^{-x^{3/2}t^{-1/2} \phi_{-}(w;x,t)} & 1
\end{pmatrix}  \\
{\times}\begin{pmatrix}
0 & 1 \\ -1 & 0
\end{pmatrix} \begin{pmatrix}
1 & 0 \\ e^{-x^{3/2}t^{-1/2} \phi_{+}(w;x,t)} & 1
\end{pmatrix}.
\end{multline}
This factorization will allow us to split the jump for $T$ on $(-\infty,a(x,t))$ into three jumps for a new matrix function $S$ with jump contour $\mathbb R\cup (i\mathbb R+a(x,t))$, exactly as in Section \ref{section:3}.
To this end, we define $S$ as
\begin{align}
\label{eq:StoTCC}
S(w;x,t) := \begin{cases}
T(w;x,t)\begin{pmatrix}
1 & 0 \\
\mp e^{-x^{3/2}t^{-1/2}\phi(w;x,t)} & 1
\end{pmatrix}, & \mbox{if } \re w<a(x,t)\mbox{ and }\pm\im w >0, \\
T(w;x,t), & \mbox{if } \re w>a(x,t).
\end{cases}
\end{align}
The following RH conditions are direct consequences of this definition and of the RH problem for $T$.
\subsubsection*{RH problem for $S$}
\begin{itemize}
\item[(a)] $S: \mathbb{C}\setminus \big( \mathbb{R} \cup (i\mathbb R+a(x,t))\big) \to \mathbb{C}^{2\times 2}$ is analytic.
\item[(b)] $S$ has the jumps
\begin{align*}
& S_{+}(w;x,t) = S_{-}(w;x,t)\begin{pmatrix}
1 & 0 \\
e^{-x^{3/2}t^{-1/2} \phi(w;x,t)} & 1
\end{pmatrix}, & & w \in i\R+a(x,t), \\
& S_{+}(w;x,t) = S_{-}(w;x,t) \begin{pmatrix}
0 & 1 \\
-1 & 0
\end{pmatrix}, & & w \in (-\infty,a(x,t)), \\
& S_{+}(w;x,t) = S_{-}(w;x,t) \begin{pmatrix}
1 & e^{x^{3/2}t^{-1/2} \phi(w;x,t)} \\
0 & 1
\end{pmatrix}, & & w \in (a(x,t),+\infty),
\end{align*}
where $\R$ is oriented from left to right, and the vertical half-lines $a(x,t)\pm i\R^+$ are oriented towards $a(x,t)$.
As usual, $+$ and $-$ denote the boundary values from the left and right sides of the contour, respectively. 
\item[(c)] We have
\begin{align*}
S(w;x,t) = \bigg( I + \frac{1}{w}T_{1}(x,t) + \bigO(w^{-2}) \bigg) w^{\frac{1}{4}\sigma_{3}}A^{-1}, \qquad \mbox{as } w\to \infty,
\end{align*}
where $T_1(x,t)$ is the same as in \eqref{eq:T1CC}.
\item[(d)] As $w \to a(x,t)$, $S(w;x,t) = \bigO(1)$.
\end{itemize}

{The next proposition shows} that the jump matrix for $S$ is close to the identity except on $(-\infty,a(x,t))$ and near $a(x,t)$.

\begin{proposition}\label{prop:lenses2}
Let $\sigma$ satisfy Assumptions \ref{assumptions}. Then, we have the inequalities
\begin{align*}
\exp\left(x^{3/2}t^{-1/2}\phi(w;x,t)\right)&\leq\exp\left(-\frac{4}{3}x^{3/2}t^{-1/2}(w-a(x,t))^{3/2}\right),&&\mbox{for $w>a(x,t)$,}\\
\exp\left(-x^{3/2}t^{-1/2}\phi(w;x,t)\right)&\leq \exp\left(-\frac{2\sqrt{2}}{3}x^{3/2}t^{-1/2}|w-a(x,t)|^{3/2}\right),&&\mbox{for $w\in a(x,t)+i\mathbb R$.}
\end{align*}
\end{proposition}
\begin{proof}
The proof is identical to that of \cite[Proposition 3.5]{CaCl2019}, and relies on Assumptions \ref{assumptions}, part 4.
\end{proof}

\subsection{Global parametrix}
Ignoring the exponentially small jumps of $S$ on the vertical line and a small, but fixed, disk centered at $a(x,t)$, we obtain a RH problem with jumps only on $(-\infty,a(x,t))$.

\subsubsection*{RH problem for $P^{(\infty)}$}
\begin{itemize}
\item[(a)] $P^{(\infty)}: \mathbb{C}\setminus (-\infty,a(x,t)) \to \mathbb{C}^{2\times 2}$ is analytic.
\item[(b)] $P^{(\infty)}$ has the jumps
\begin{align}
\label{eq:jumpPCC}
& P^{(\infty)}_{+}(w;x,t) = P^{(\infty)}_{-}(w;x,t) \begin{pmatrix}
0 & 1 \\
-1 & 0
\end{pmatrix}, & & w \in (-\infty,a(x,t)).
\end{align}
\item[(c)] There exists a matrix $P^{(\infty)}_1(x,t)$ such that
\begin{align}\label{asymp for Pinf CC}
P^{(\infty)}(w;x,t) = \bigg( I + \frac{1}{w}P^{(\infty)}_1(x,t) + \bigO(w^{-2}) \bigg) w^{\frac{1}{4}\sigma_{3}}A^{-1}, \qquad \mbox{as } w \to \infty.
\end{align}
\end{itemize}
If we impose in addition the condition
\begin{itemize}
\item[(d)] As $w\to a(x,t)$, $P^{(\infty)}(w;x,t) = \bigO(|w-a(x,t)|^{-\frac 14})$,
\end{itemize}
there is a unique solution to this RH {problem}, given by
\begin{equation}
\label{PinfCC}
P^{(\infty)}(w;x,t) = (w-a(x,t))^{\frac{1}{4}\sigma_{3}}A^{-1},
\end{equation}
and then the matrix $P^{(\infty)}_1(x,t)$ is given by
\begin{align}\label{P1inf CC}
P_{1}^{(\infty)}(x,t) = -\frac{1}{4}a(x,t)\sigma_3.
\end{align}

\subsection{Local Airy parametrix near $a(x,t)$}
As in {Section \ref{subsection:Airy local param}}, we now want to construct a local parametrix $P^{(a)}$ which has exactly the same jumps as $S$ in a disk $|w-a(x,t)|<\rho$. The local parametrix is constructed again in terms of the Airy model RH problem (see Appendix \ref{app:airy}), and it is given by
\begin{align*}
P^{(a)}(w;x,t) := E(w;x,t) \Phi_{\mathrm{Ai}}^{[j]}(xt^{-1/3}f(w;x,t))e^{-\frac{x^{3/2}t^{-1/2}}{2}\phi(w;x,t)\sigma_{3}},
\end{align*}
where $j=1$ for $0<\arg(w-a(x,t))<\pi/2$, $j=2$ for $\pi/2<\arg(w-a(x,t))<\pi$, $j=3$ for $-\pi<\arg(w-a(x,t))<-\pi/2$, and $j=4$ for $-\pi/2<\arg(w-a(x,t))<0$, with
$\Phi_{\mathrm{Ai}}^{[1]},\ldots,{\Phi}_{\mathrm{Ai}}^{[4]}$ the entire functions defined in Appendix \ref{app:airy} which together form the solution to the Airy model RH problem.
The functions $f$ and $E$ are given by
\begin{align}\label{def of f and E CC}
f(w;x,t) := \bigg( -\frac{3}{4}\phi(w;x,t) \bigg)^{2/3},\qquad
E(w;x,t) := P^{(\infty)}(w;x,t) A^{-1} \big( xt^{-1/3}f(w;x,t) \big)^{\frac{\sigma_{3}}{4}}.
\end{align}
Similarly as in {Section \ref{subsection:Airy local param}}, we will now show that there exists a sufficiently small $\rho>0$ such that $f(w;x,t)$ is a conformal transformation of $w\in\mathcal D :=\{w\in\mathbb C:\ |w-a(x,t)|<\rho\}$ onto a neighborhood of $0$, satisfying $f'(a(x,t);x,t)>0$, and that $E(w;x,t)$ is a holomorphic function of $w\in\mathcal D$.
For the first statement we start by rewriting \eqref{gprime 2 CC} as
\begin{align*}
g'(w;x,t) = & \; -\sqrt{w-a(x,t)}\left(1+\frac{1}{2\pi\sqrt{xt}}\int_{-\infty}^{a(x,t)}\frac{(\log F)'(\tfrac xt\zeta)-(\log F)'(\tfrac xtw)}{\zeta-w}\frac{d\zeta}{\sqrt{a(x,t)-\zeta}}\right) \\
& \; +\frac 1{2\sqrt{xt}}(\log F)'(\tfrac xt w).
\end{align*}
The Taylor series at $w=a(x,t)$
\begin{align*}
\frac{(\log F)'(\tfrac xt\zeta)-(\log F)'(\tfrac xt w)}{\zeta-w} &= \sum_{\ell = 0}^{+\infty}\frac{\varphi_\ell(\zeta;x,t)}{(\zeta-a(x,t))^{\ell+1}}(w-a(x,t))^\ell,
\\
\varphi_\ell(\zeta;x,t)&=(\log F)'\left(\frac xt\zeta\right)-\sum_{j=0}^{\ell}\left(\frac xt\right)^j(\log F)^{(j+1)}\left(\frac xt a(x,t)\right)\frac{(\zeta-a(x,t))^j}{j!},
\end{align*}
implies the identity
\begin{align}
\label{tobeintegrated}
g'(w;x,t)=& -\sqrt{w-a(x,t)}+\frac {(\log F)'(\tfrac xt w)}{2\sqrt{xt}} \nonumber \\
& {-}\frac 1{2\pi\sqrt{xt}}\sum_{\ell=0}^{+\infty}(w-a(x,t))^{\ell+\frac 12}\int_{-\infty}^{a(x,t)}\frac{\varphi_\ell(\zeta;x,t)\,d\zeta}{(\zeta-a(x,t))^{\ell+1}\sqrt{a(x,t)-\zeta}},
\end{align}
where the last term is a Puiseux series as $w\to a(x,t)$ away from the cut $(-\infty,a(x,t))$.

To simplify the coefficients let us integrate by parts
\begin{align*}
(-1)^{\ell+1}\int_{-\infty}^{a(x,t)}\frac{\varphi_\ell(\zeta;x,t)\,d\zeta}{(a(x,t)-\zeta)^{\ell+\frac 32}}
&=\frac{2^{\ell+1}}{(2\ell+1)!!}\int_{-\infty}^{a(x,t)}\varphi_\ell(\zeta;x,t)\left[(-\partial_\zeta)^{\ell+1}\frac 1{\sqrt{a(x,t)-\zeta}}\right]\,d\zeta
\\
&=\frac{2^{\ell+1}}{(2\ell+1)!!}\left(\frac xt\right)^{\ell+1}\int_{-\infty}^{a(x,t)}(\log F)^{(\ell+2)}\left(\frac xt\zeta\right)\frac{d\zeta}{\sqrt{a(x,t)-\zeta}} .
\end{align*}
Using this relation, integrating \eqref{tobeintegrated} in $w$ (by construction $g(a(x,t);x,t)=0$) and applying the definition \eqref{def of phi CC} of $\phi(w;x,t)$ we finally obtain
\begin{multline}
\label{eq:seriesphiCC}
\phi(w;x,t)=-\frac 43(w-a(x,t))^{3/2} \\
{\times}\left[1+\frac{3}{2\pi\sqrt{xt}}\sum_{\ell=0}^{+\infty}\left(\frac{2x}t\right)^{\ell+1}\frac{(w-a(x,t))^\ell}{(2\ell+3)!!}\int_{-\infty}^{a(x,t)}(\log F)^{(\ell+2)}\left(\frac xt\zeta\right)\frac{d\zeta}{\sqrt{a(x,t)-\zeta}}\right].
\end{multline}
Consequently, $f$ defined in \eqref{def of f and E CC} is holomorphic in a neighborhood of $a(x,t)$, maps $a(x,t)$ to the origin, and $f'(a(x,t);x,t)>0$.
The first few terms in the Taylor expansion of $f$ near $w=a(x,t)$ are
\begin{align}
\label{eq:f1f2CC}
f(w;x,t)&=f_1(x,t)(w-a(x,t))+f_2(x,t)(w-a(x,t))^2+\bigO((w-a(x,t))^3),
\\
\nonumber
f_1(x,t)&=\left(1+\frac {x^{1/2}}{\pi t^{3/2}}\int_{-\infty}^{a(x,t)}(\log F)''(\tfrac xt \zeta)\frac{d\zeta}{\sqrt{a(x,t)-\zeta}}\right)^{2/3},
\\
\nonumber
f_2(x,t)&=\frac {4x^{3/2}}{15\pi t^{5/2}f_1(x,t)^{1/2}}\int_{-\infty}^{a(x,t)}(\log F)'''(\tfrac xt \zeta)\frac{d\zeta}{\sqrt{a(x,t)-\zeta}}.
\end{align}
Finally, the proof that $E$ is holomorphic in $\mathcal D$ is a straightforward adaptation of the argument in Section \ref{subsection:Airy local param}.

\subsection{Small norm RH problem}

Define
\begin{align}\label{def of R large xt}
R(w;x,t) := \begin{cases}
S(w;x,t)P^{(\infty)}(w;x,t)^{-1}, & w\in \mathbb{C}\setminus \overline{\mathcal{D}}, \\
S(w;x,t)P^{(a)}(w;x,t)^{-1}, & w\in \mathcal{D}.
\end{cases}
\end{align}
The RH conditions for $R$ are detailed below and follow directly from those of $S$, $P^{(\infty)}$, and $P^{(a)}$. 
Let us denote $\Gamma_R:=\partial \mathcal D \cup \Gamma_{-}\cup\Gamma_0\cup\Gamma_{+}$ where $\Gamma_0:=(a(x,t)+\rho,+\infty)$, $\Gamma_\pm := (a(x,t)\pm i\rho,a(x,t)\pm i\infty)$; we consider $\Gamma_{0}$ oriented from left to right, $\Gamma_\pm$ oriented towards $a(x,t)$, and $\partial \mathcal D$ oriented clockwise (see Figure \ref{figcontourR} for the analogous contour of Section \ref{section:3}).
As usual, boundary values are labeled by $+$ (resp. $-$) when we approach the oriented contour from the left (resp. right).

\subsubsection*{RH problem for $R$}
\begin{itemize}
\item[(a)] $R: \mathbb{C}\setminus\Gamma_R \to \mathbb{C}^{2\times 2}$ is analytic.
\item[(b)] $R$ has the jumps
\begin{align}
&\label{eq:jumpRconditions1CC}
R_{+}(w;x,t) = R_{-}(w;x,t)P^{(\infty)}(w;x,t)\begin{pmatrix}
1 & 0\\e^{-x^{3/2}t^{-1/2} \phi(w;x,t)} & 1
\end{pmatrix}P^{(\infty)}(w;x,t)^{-1}, & & w\in \Gamma_{\pm},
\\
& R_{+}(w;x,t) = R_{-}(w;x,t) P^{(\infty)}(w;x,t)\begin{pmatrix}
1 & e^{x^{3/2}t^{-1/2} \phi(w;x,t)} \\
0 & 1
\end{pmatrix}P^{(\infty)}(w;x,t)^{-1}, & & w \in \Gamma_0,
\\
& R_{+}(w;x,t) = R_{-}(w;x,t)P^{(a)}(w;x,t)P^{(\infty)}(w;x,t)^{-1}, & & w\in \partial\mathcal D.
\label{eq:jumpRconditionsCC}
\end{align}
\item[(c)] We have
\begin{align}\label{asymp for R CC}
R(w;x,t) = I + \frac{1}{w}R_{1}(x,t) + \bigO(w^{-2}), \qquad \mbox{as } w \to \infty.
\end{align}
\item[(d)] As $w\to a(x,t)+\rho$ and as $w\to a(x,t)\pm i\rho$, $R(w;x,t)=\bigO(1)$.
\end{itemize}

In particular, by construction of the global parametrix $P^{(\infty)}$ there is no jump on $(-\infty,a(x,t))$, and by construction of the local parametrix $P^{(a)}$ there is no jump inside $\mathcal D$. Furthermore, since both $S$ and $P^{(a)}$ remain bounded near $a(x,t)$, $R$ has no pole at $a(x,t)$.

We now show that the jump for $R$ satisfies a small norm RH problem; to this end we introduce the matrix function $J_R:\Gamma_R\to\C^{2\times 2}$, piece-wise defined according to \eqref{eq:jumpRconditions1CC}--\eqref{eq:jumpRconditionsCC} so that $R_+=R_-J_R$ on $\Gamma_R$.

\begin{lemma}
We have $\|J_R-I\|_{p} = \mathcal O(x^{-3/2}t^{1/2})$ for $p=1,2,\infty$, uniformly in $x\geq K$, $xt\geq\delta${, $t \leq t_{0}$}; here $\|\cdot\|_p$ denotes the norm in $L^p(\Gamma_R,\C^{2\times 2})$ with respect to any matrix norm on $\C^{2\times 2}$.
\end{lemma}
\begin{proof}
By construction of the global parametrix $P^{(\infty)}$ we see that $P^{(\infty)}(w;x,t)$ and $P^{(\infty)}(w;x,t)^{-1}$ are $\bigO(\sqrt[4]{|w|+1})$ for $|w-a(x,t)|>\rho$, uniformly in $xt\geq \delta${, $x\geq K$, $t \leq t_{0}$}.
From this and Proposition \ref{prop:lenses2}, it follows that $J_R(w)=I+\bigO(\sqrt{|w|+1}\,e^{-\frac{2\sqrt 2}3 x^{3/2}t^{-1/2}|w-a(x,t)|^{3/2}})$ for $|w-a(x,t)|>\rho$, uniformly in $xt\geq \delta${, $x\geq K$, $t \leq t_{0}$}.
It follows that the $L^p(\Gamma_R,\C^{2\times 2})$-norms of $J_R-I$ for $p=1,2,\infty$ are $\bigO(e^{- x^{3/2}t^{-1/2}\rho^{3/2}})$ uniformly in $xt\geq\delta${, $x\geq K$, $t \leq t_{0}$}.
On the remaining part of the contour we have, as $x\to+\infty$,
\begin{align}
\nonumber
J_R(w;x,t)-I&=P^{(a)}(w;x,t)P^{(\infty)}(w;x,t)^{-1} -I
\\
\nonumber
&= \frac{1}{x^{3/2}t^{-1/2} f(w;x,t)^{3/2}}P^{(\infty)}(w;x,t) \Phi_{\mathrm{Ai},1} P^{(\infty)}(w;x,t)^{-1} + \bigO(x^{-3}t)
\\
\label{eq:estimatejumpRsec4}
&=\frac{i}{48x^{3/2}t^{-1/2} f(w;x,t)^{3/2}}
\begin{pmatrix}
0&7(w-a(x,t))^{1/2}\\
5(w-a(x,t))^{-1/2}&0
\end{pmatrix}
+\bigO(x^{-3}t),
\end{align}
uniformly for $w\in \partial \mathcal D$, where we use \eqref{P1inf CC} and $\Phi_{\mathrm{Ai},1}=\frac 18\begin{pmatrix}
\frac 16 & i \\ i & -\frac 16
\end{pmatrix}$, see Appendix \ref{app:airy}.
Moreover this asymptotic relation is uniform for $x\geq K${, $t \leq t_{0}$} and $xt\geq \delta$; indeed, it is clear by \eqref{eq:seriesphiCC} and \eqref{def of f and E CC} that $f(w;x,t)$ is bounded away from zero for $w\in\partial\mathcal D$, uniformly for the relevant values of $x,t$.
\end{proof}
Similarly as in Section \ref{subsection:small norm with xt small}, it should be noted that the contour $\Gamma_R$ depends on $x$ and $t$ through $a(x,t)$.
This is not so convenient to apply  the standard theory of small-norm RH problems; however one observes easily that the jump contour is independent of $x$ and $t$ after the shift $\widetilde w=w-a(x,t)$ which is uniformly bounded by Proposition \ref{prop:endpointCC}, so that the standard theory applies to the RH problem in that variable.
Back in our variable $w$, we obtain that, as $x \to + \infty$,
\begin{align}
\label{asymp for R easy sector CC}
R(w;x,t) = I + \frac{1}{x^{3/2}t^{-1/2}}R^{(1)}(w;x,t) + \bigO(x^{-3}t),
\qquad R^{(1)}(w;x,t) = \frac{1}{2\pi i}\int_{\partial \mathcal D} \frac{J_{R}^{(1)}(s;x,t)}{s-w}ds
\end{align}
uniformly in $w$, where (see \eqref{eq:estimatejumpRsec4})
\begin{align}
J_{R}^{(1)}(w;x,t) := \frac{i}{48f(w;x,t)^{3/2}}
\begin{pmatrix}
0&7(w-a(x,t))^{1/2}\\
5(w-a(x,t))^{-1/2}&0
\end{pmatrix}.
\end{align}
It is straightforward to check, by using the definition of $f(w;x,t)$ in \eqref{def of f and E CC}, that $J_{R}^{(1)}(w;x,t)$ is meromorphic for $w\in\mathcal D$ with a double pole at $w=a(x,t)$ only.
Thus we can compute $R^{(1)}$ in \eqref{asymp for R easy sector CC} for $w\in\C\setminus\overline{\mathcal D}$ by extracting the polar part of $J^{(1)}_R$ at $w=a(x,t)$; namely, if $|w-a(x,t)|>\rho$,
\begin{align}
\label{eq:Rp1pCC}
R^{(1)}(w;x,t) &= \frac{i}{16 f_1^{5/2}}\begin{pmatrix}
0 & \frac{7}{3}f_1 \\
- \frac{5}{2}f_2 & 0
\end{pmatrix} \frac{1}{w-a(x,t)} + \frac{5i}{48 f_1^{3/2}} \begin{pmatrix}
0 & 0 \\ 1 & 0
\end{pmatrix}\frac{1}{(w-a(x,t))^{2}},
\end{align}
where $f_1=f_1(x,t)$ and $f_2=f_2(x,t)$ are given in \eqref{eq:f1f2CC}.
Since \eqref{asymp for R easy sector CC} is uniform in $w$ we also obtain
\begin{equation}
\label{R1infCC}
R_1(x,t) =\frac{i t^{1/2}}{16 x^{3/2} f_{1}(x,t)^{5/2}}\begin{pmatrix}
0 & \frac{7}{3}f_1(x,t) \\
- \frac{5}{2}f_2(x,t) & 0
\end{pmatrix} +\bigO(x^{-3}t).
\end{equation}

\subsection{Asymptotics for $u_\sigma(x,t)$ and $\partial_x\log Q_\sigma(x,t)$}

\begin{proposition}
\label{prop:asympusect4}
As $x\to+\infty$ we have, uniformly in $xt\geq\delta${, $t \leq t_{0}$} for any $\delta>0$ {and $t_{0}>0$},
\begin{align}
u_\sigma(x,t) = \frac x{2t}a(x,t)+\bigO(x^{-2}).
\end{align}
where $a(x,t)$ is the unique solution to \eqref{a def CC}.
In particular we have
\begin{align}
u_\sigma(x,t) = \frac x{2t}a_0\left(\frac{\pi^2}{c_+^2}xt\right)+\frac 1{2\sqrt{xt}}a_1\left(\frac{\pi^2}{c_+^2}xt\right)+\frac{t^{1/2}}{2x^{3/2}}a_2\left(\frac{\pi^2}{c_+^2}xt\right)+\bigO(x^{-2})
\end{align}
as $x\to+\infty$, uniformly in $xt\geq \delta${, $t \leq t_{0}$}, where $a_i$ are {as in \eqref{def:ainfty} and in Theorem \ref{theorem:main}}.
\end{proposition}
\begin{proof}
By combining \eqref{eq:upq}, {and \eqref{eq:T1CC}} we obtain the identity
\begin{equation*}
u_\sigma(x,t) = \frac xt(T_{1,21}(x,t)^2-2T_{1,11}(x,t)).
\end{equation*}
For $w$ satisfying $\re w>a(x,t)+\rho$, using \eqref{eq:StoTCC} {and \eqref{def of R large xt}} we have
\begin{align*}
T(\zeta;x,t)=S(w;x,t)=R(w;x,t)P^{(\infty)}(w;x,t).
\end{align*}
Substituting the large $w$ asymptotics \eqref{eq:asTCC} for $T$, \eqref{asymp for R CC} for $R$, and \eqref{asymp for Pinf CC} for $P^{(\infty)}$, we get
\begin{align}
\label{eq:TRPCC}
T_1(x,t)=P_1^{(\infty)}(x,t)+R_1(x,t).
\end{align}
The proof is complete by \eqref{P1inf CC} and \eqref{R1infCC}.
\end{proof}

\begin{proposition}
{As $x\to+\infty$ we have, uniformly in $xt\geq\delta$, $t \leq t_{0}$ for any $\delta>0$ and $t_{0}>0$,}
\begin{align}
\label{eq:finaldiffidCC}
\partial_x\log Q_\sigma(x,t)=-\frac{x^2}{t}\left(\frac 14+g_1(x,t)\right)-\frac{5}{32x}\frac{f_2(x,t)}{f_1(x,t)^{5/2}}+\bigO(x^{-5/2}t^{1/2}),
\end{align}
where $g_1(x,t)$ is given in \eqref{eq:identityg1} and {$f_1(x,t), f_2(x,t)$} in \eqref{eq:f1f2CC}.
\end{proposition}
\begin{proof}
By combining \eqref{xdiffid} and \eqref{eq:T1CC} we obtain the identity
\begin{align}
\partial_x\log Q_\sigma(x,t)=-\frac{x^2}{4t}-\frac{x^2}tg_1(x,t)-i\sqrt{\frac xt}T_{1,21}(x,t).
\end{align}
By \eqref{eq:TRPCC} we get
\begin{align}
\partial_x\log Q_\sigma(x,t)=-\frac{x^2}{4t}-\frac{x^2}tg_1(x,t)-i\sqrt{\frac xt}P_{1,21}^{(\infty)}(x,t)-i\sqrt{\frac xt}R_{1,21}(x,t),
\end{align}
and it suffices to use \eqref{P1inf CC} and \eqref{R1infCC} to obtain the result.
\end{proof}

We now study the asymptotics of the two terms on the right-hand side of \eqref{eq:finaldiffidCC}; this is done in the next two lemmas.

\begin{lemma}
As $x\to+\infty$ we have, uniformly in $xt\geq\delta$ {and $t \leq t_{0}$},
\begin{align}
\nonumber
\frac{x^2}t\left(\frac 14+g_1(x,t)\right) &=
(a_0(y)-1)^2\frac{x^2}{4 t}
+\frac{2c_+}{3\pi}a_0(y)^{3/2}\frac{x^{3/2}}{t^{3/2}}
+\left( \frac{(a_0(y)-1) a_1(y)}2+\frac{\log c_+'}\pi\sqrt{a_0(y)}\right)\frac{x^{1/2}}{t^{1/2}}
\\
\nonumber
&\qquad
+\frac{c_+}{\pi}a_1(y)\sqrt{a_0(y)}\frac 1t
+\left(\frac{(a_0(y)-1)a_2(y)}2-\frac{j_\sigma}{\sqrt{a_0(y)}}\right)\frac{t^{1/2}}{x^{1/2}}
\\
\nonumber
&\qquad
+\left(\frac{a_1(y)^2}4+\frac{c_+}{\pi}a_2(y)\sqrt{a_0(y)}+\frac{\log c_+'}{2\pi\sqrt{a_0(y)}}a_1(y)\right)\frac 1x
\\
\label{x2g/t}
&\qquad
+\frac{c_+}{4\pi}\frac{a_1(y)^2}{\sqrt{a_0(y)}}\frac 1{t^{1/2}x^{3/2}}+\bigO(t^{3/2}x^{-3/2})
\end{align}
where $y=\pi^2xt/c_+^2$, {$a_0(y),a_1(y),a_2(y)$} are {as in \eqref{def:ainfty} and in Theorem \ref{theorem:main}}, and $j_\sigma$ is as in \eqref{eq:Isigma}.
\end{lemma}
\begin{proof}
We start by rewriting the expression \eqref{eq:identityg1} as
\begin{align}
\nonumber
g_1(x,t)&=
\frac{a(x,t)^2}4-\frac{a(x,t)}2+\frac{1}{\pi \sqrt{xt}} \int_{-\infty}^{a(x,t)}(\log F)'(\tfrac{x}{t}\zeta)\sqrt{a(x,t)-\zeta}\,d\zeta
\\
\nonumber
&=\frac{a(x,t)^2}4-\frac{a(x,t)}2+\frac{1}{2\pi}\sqrt{\frac t{x^3}}\int_{-\infty}^{a(x,t)}\frac{\log F(\tfrac xt\zeta)}{\sqrt{a(x,t)-\zeta}}d\zeta
\\
\nonumber
&=\frac{a(x,t)^2}4-\frac{a(x,t)}2
+\frac{2c_+}{3\pi\sqrt{xt}} a(x,t)^{3/2}+\frac{ \log c_+'}{\pi}\sqrt{\frac t{x^3}a(x,t)}
\\
&\qquad\qquad\qquad
+\frac{1}{2\pi}\frac {t^{3/2}}{x^{5/2}}\int_{-\infty}^{xa(x,t)/t}\frac{\log F(r)-(c_+r+\log c_+')1_{(0,+\infty)}(r)}{\sqrt{a(x,t)-\tfrac tx r}}dr
\end{align}
where in the first step we integrate by parts, and in the last step we also change integration variable $r=\tfrac xt \zeta$.
Similar arguments to those used previously {in the proof of Proposition \ref{prop:endpointCC}} allow to estimate the last integral as
\begin{equation}
\frac 1{2\pi}\int_{-\infty}^{xa(x,t)/t}\frac{\log F(r)-(c_+r+\log c_+')1_{(0,+\infty)}(r)}{\sqrt{a(x,t)-\tfrac tx r}}dr =- \frac{j_\sigma}{\sqrt{a(x,t)}}+\bigO(t/x).
\end{equation}
It then suffices to apply \eqref{def:ainfty}.
\end{proof}

\begin{lemma}
We have the following asymptotics as $x\to+\infty$, uniformly for $xt\geq\delta$ and $t \leq t_{0}$:
\begin{align}
\label{eq:R121asympCC}
-\frac{5f_2(x,t)}{32f_1(x,t)^{5/2}} = \frac{1+\sqrt{1+y}}{48(1+y)}+\bigO(t^{1/2}x^{-{3}/2}),
\end{align}
{where $y=\pi^2xt/c_+^2$.}
\end{lemma}
\begin{proof}
For any $a>0$ we have
\begin{align*}
\int_{-\infty}^{a}(\log F)''(\tfrac xt\zeta)\frac{d\zeta}{\sqrt{a-\zeta}}= {\frac{t}{x}} \frac 1{\sqrt a}\int_{-\infty}^{+\infty}(\log F)''(r)\,dr+\bigO(t{^{2}}/x{^{2}})
\end{align*}
as $x\to+\infty$, where we use similar estimates as in the proof of Proposition \ref{prop:endpointCC}.
Moreover
\begin{align*}
\int_{-\infty}^{+\infty}(\log F)''(r)\,dr = \lim_{r_1,r_2\to+\infty} \left[(\log F)'(r_1)-(\log F)'(-r_2)\right] = c_+.
\end{align*}
Similarly, for any $a>0$ we have
\begin{align*}
\int_{-\infty}^{a}(\log F)'''(\tfrac xt\zeta)\frac{d\zeta}{\sqrt{a-\zeta}}=\frac 1{2a^{3/2}}\frac{t^2}{x^2}\int_{-\infty}^{+\infty}r\,(\log F)'''(r)\,dr+\bigO(t^{3}/x^{3})
\end{align*}
as $x\to+\infty$.
Moreover we can integrate by parts to obtain
\begin{align*}
\int_{-\infty}^{+\infty}r\,(\log F)'''(r)\,dr = -\int_{-\infty}^{+\infty}(\log F)''(r)\,dr=-c_+.
\end{align*}
The proof is complete by the explicit formulae \eqref{eq:f1f2CC} for $f_1,f_2$ and the expansion \eqref{def:ainfty} for $a(x,t)$.
\end{proof}

\begin{corollary}\label{corollary:Qsec4}
We have
\begin{align*}
\partial_x\log Q_\sigma(x,t) =
\partial_x\left[-\frac{c_+^6}{\pi^6t^4}F_1\left(\frac{\pi^2}{c_+^2}xt\right)-\frac{c_+^3\log c_+'}{\pi^4t^2}F_2\left(\frac{\pi^2}{c_+^2}xt\right)+F_3\left(\frac{\pi^2}{c_+^2}xt\right)\right]
+\bigO(t^{3/2}x^{-3/2}),
\end{align*}
where $F_1, F_{2}, F_{3}$ are defined in Theorem \ref{theorem:main}.
\end{corollary}
\begin{proof}
It is a direct verification based on \eqref{eq:finaldiffidCC} and the last two lemmas.
\end{proof}

\subsection{Proof of Theorem \ref{theorem:main}}\label{finalsec}
The asymptotics \eqref{thm:asu} for $u_\sigma$ 
follow from Proposition \ref{prop:asympusect3} if $0<xt\leq \delta$. Indeed, in this regime, it is easy to see that the term
$\frac{t^{1/2}}{2x^{3/2}}a_2\left(\frac{\pi^2}{c_+^2}xt\right)$ is $\mathcal O(x^{-2})$, since $a_2(y)=\mathcal O(y^{-1/2})$ as $y\to 0$ by \eqref{eq:thm-a2}.
If $xt\geq \delta$, \eqref{thm:asu} follows immediately from Proposition \ref{prop:asympusect4}.

\medskip

For $Q_\sigma$, we first observe that Corollary 
\ref{corollary:Qsec3} and Corollary \ref{corollary:Qsec4} together imply that
\[\partial_x\log Q_\sigma(x,t)=
\partial_x\left[-\frac{c_+^6}{\pi^6t^4}F_1\left(\frac{\pi^2}{c_+^2}xt\right)-\frac{c_+^3\log c_+'}{\pi^4t^2}F_2\left(\frac{\pi^2}{c_+^2}xt\right)+F_3\left(\frac{\pi^2}{c_+^2}xt\right)\right]
+\bigO(t^{3/2}x^{-3/2})+\mathcal O(x^{-3}),
\]
uniformly for $x\geq K$, $t\leq t_0$.
Let us now integrate this expression in $x$, starting at the point $x=K$. Since the error terms are integrable in $x$ between $x=K$ and $x=\infty$, we then obtain
\begin{multline}\log Q_\sigma(x,t)=\log Q_\sigma(K,t)
+\frac{c_+^6}{\pi^6t^4}F_1\left(\frac{\pi^2}{c_+^2}Kt\right)+\frac{c_+^3\log c_+'}{\pi^4t^2}F_2\left(\frac{\pi^2}{c_+^2}Kt\right)-F_3\left(\frac{\pi^2}{c_+^2}Kt\right)\\
-\frac{c_+^6}{\pi^6t^4}F_1\left(\frac{\pi^2}{c_+^2}xt\right)-\frac{c_+^3\log c_+'}{\pi^4t^2}F_2\left(\frac{\pi^2}{c_+^2}xt\right)+F_3\left(\frac{\pi^2}{c_+^2}xt\right)
+\bigO(1).
\end{multline}
Then, we distinguish two cases. First, if $t$ does not tend to $0$, then $\log Q_\sigma(K,t)$ remains bounded as well as the other terms that do not depend on $x$, and the result is proven. Secondly, if $t\to 0$, we know from \cite[Theorem 1.14 (iii)]{CaClR2020}, see \eqref{logQexpansion}, that
\begin{align*}
\log Q_\sigma(K,t) & = -\frac{K^3}{12t}-\frac{1}{8}\log (Kt^{-1/3}) + \bigO(1) = -\frac{K^3}{12t}+\frac{1}{24}\log t + \bigO(1), \qquad \mbox{as } t \to 0.
\end{align*}
Furthermore, it is straightforward to see from \eqref{eq:thm-F1F2}--\eqref{eq:thm-F3}, particular from the asymptotics as $y\to 0$, that
\begin{align*}
& \frac{c_+^6}{\pi^6t^4}F_1 \hspace{-0.07cm}\left(\frac{\pi^2}{c_+^2}Kt\right)\hspace{-0.08cm}+\hspace{-0.06cm}\frac{c_+^3\log c_+'}{\pi^4t^2}F_2 \hspace{-0.07cm} \left(\frac{\pi^2}{c_+^2}Kt\right) \hspace{-0.05cm} -\hspace{-0.05cm} F_3 \hspace{-0.07cm} \left(\frac{\pi^2}{c_+^2}Kt\right) \hspace{-0.05cm}= \hspace{-0.05cm}
\frac{K^3}{12t} 
\hspace{-0.05cm} - \hspace{-0.05cm} \left(\frac{2c_+ j_\sigma}\pi+\frac{\log^2c'_+}{2\pi^2}+\frac{1}{24}\right)\log t + \bigO(1),
\end{align*}
and that yields the result also in this case.

\paragraph{Acknowledgements.}
The authors are grateful to Pierre Le Doussal for important remarks about related literature.
CC was supported by the European Research Council, Grant Agreement No. 682537, the Ruth and Nils-Erik Stenb\"ack Foundation, and the Novo Nordisk Fonden Project Grant 0064428.
TC and GR were supported by the Fonds de la Recherche Scientifique-FNRS under EOS project O013018F.

\appendix

\section{Airy model RH problem}
\label{app:airy}
Introduce the following $2\times 2$-matrix valued entire functions:
\begin{align*}
\Phi_\Ai^{[1]}(z):={}&
\sqrt{2 \pi}\, e^{\frac{\pi i}{6}} \begin{pmatrix}
1 & 0 \\ 0 & -i
\end{pmatrix}
\begin{pmatrix}
\Ai(z) & \Ai(\omega^2 z) \\ \Ai'(z) & \omega^2\Ai'(\omega^2 z)
\end{pmatrix}
e^{-i\pi\sigma_3/6},
\\
\Phi_\Ai^{[2]}(z):={}&
\sqrt{2 \pi}\, e^{\frac{\pi i}{6}} \begin{pmatrix}
1 & 0 \\ 0 & -i
\end{pmatrix}
\begin{pmatrix}
-\omega\Ai(\omega z) & \Ai(\omega^2 z) \\ -\omega^2\Ai'(\omega z) & \omega^2\Ai'(\omega^2 z)
\end{pmatrix}
e^{-i\pi\sigma_3/6},
\\
\Phi_\Ai^{[3]}(z):={}&
\sqrt{2 \pi}\, e^{\frac{\pi i}{6}} \begin{pmatrix}
1 & 0 \\ 0 & -i
\end{pmatrix}
\begin{pmatrix}
-\omega^{2}\Ai(\omega^{2} z) & -\omega^2\Ai(\omega z) \\ -\omega \Ai'(\omega^{2} z) & -\Ai'(\omega z)
\end{pmatrix}
e^{-i\pi\sigma_3/6},
\\
\Phi_\Ai^{[4]}(z):={}&
\sqrt{2 \pi}\, e^{\frac{\pi i}{6}} \begin{pmatrix}
1 & 0 \\ 0 & -i
\end{pmatrix}
\begin{pmatrix}
\Ai(z) & -\omega^2\Ai(\omega z) \\ \Ai'(z) & -\Ai'(\omega z)
\end{pmatrix}e^{-i\pi\sigma_3/6},
\end{align*}
where $\Ai$ is the Airy function and $\omega:=e^{2\pi i/3}$. 
Let the contour $\Sigma_{\Ai}$ be the union of four smooth arcs $\gamma_1,\dots,\gamma_4$ meeting at the origin only, each extending to infinity so that $\C\setminus\Sigma_{\Ai}$ is the disjoint union of four domains $\Omega_1,\dots,\Omega_4$ (such that the $\Omega_j$ is bounded by the arcs $\gamma_j,\gamma_{j+1}$) satisfying
\begin{align*}
&z\in\Omega_1\Rightarrow -\varepsilon<\arg z<\frac\pi 2+\varepsilon,
& &z\in\Omega_2\Rightarrow \frac\pi 2-\varepsilon<\arg z<\pi+\varepsilon,\\
&z\in\Omega_3\Rightarrow -\pi-\varepsilon<\arg z<-\frac\pi 2+\varepsilon,
& &z\in\Omega_4\Rightarrow -\frac\pi 2-\varepsilon<\arg z<\varepsilon,
\end{align*}
for some $\varepsilon>0$ sufficiently small ($\varepsilon<\pi/6$ is enough).
It can be shown by the connection formula $\Ai(z)+\omega\Ai(\omega z)+\omega^{2}\Ai(\omega^2 z)=0$ and standard asymptotic properties of the Airy function that the piecewise defined matrix
\begin{align}
\Phi_{\Ai}(z) := \Phi_{\Ai}^{[j]}(z),\qquad z\in\Omega_j,
\end{align}
satisfies the following RH model problem \cite{DKMVZ1}.
For the formulation of the latter we agree that $\gamma_2,\gamma_3,\gamma_4$ are oriented towards the origin and $\gamma_1$ towards infinity, and we denote by the subscripts $+$ and $-$ the boundary values from, respectively, the left and right sides of $\Sigma_{\Ai}$.
\subsubsection*{RH problem for $\Phi_\Ai$}
\begin{itemize}
\item[(a)] $\Phi_{\mathrm{Ai}} : \mathbb{C} \setminus \Sigma_{\Ai} \rightarrow \mathbb{C}^{2 \times 2}$ is analytic
\item[(b)] $\Phi_{\mathrm{Ai}}$ has the jump relations
\begin{align*}
& \Phi_{\mathrm{Ai},+}(z) = \Phi_{\mathrm{Ai},-}(z) \begin{pmatrix}
 1 & 1 \\
 0 & 1
\end{pmatrix}, & & z \in \gamma_1,
\\
& \Phi_{\mathrm{Ai},+}(z) = \Phi_{\mathrm{Ai},-}(z) \begin{pmatrix}
 1 & 0  \\ 1 & 1
\end{pmatrix}, & & z \in \gamma_2 \cup \gamma_3,
\\
& \Phi_{\mathrm{Ai},+}(z) = \Phi_{\mathrm{Ai},-}(z) \begin{pmatrix}
0 & 1 \\ -1 & 0
\end{pmatrix}, & & z \in \gamma_4.
\end{align*}
\item[(c)] As $z \to \infty$, $z \notin \Sigma_{\Ai}$, we have
\begin{equation}\label{Asymptotics Airy}
\Phi_{\mathrm{Ai}}(z) = z^{-\frac{\sigma_{3}}{4}}A \left( I + \frac 1{z^{3/2}}\Phi_{\mathrm{Ai},1} +\bigO(z^{-3})\right) e^{-\frac{2}{3}z^{3/2}\sigma_{3}},
\end{equation}
where $A$ is given in \eqref{eq:A} and $\Phi_{\mathrm{Ai},1} = \frac{1}{8}\begin{pmatrix}
\frac{1}{6} & i \\ i & -\frac{1}{6}
\end{pmatrix}$.

\item[(d)] As $z \to 0$, $\Phi_{\mathrm{Ai}}(z) = \bigO(1)$. 
\end{itemize}


\begin{thebibliography}{99}

\bibitem{AmirCorwinQuastel}
G.~Amir, I.~Corwin, \& J.~Quastel.
``Probability distribution of the free energy of the continuum directed random polymer in 1+1 dimensions''. 
\emph{Comm. Pure Appl. Math.} 64(4), 466--537 (2011).

\bibitem{BBdF}
J.~Baik, R.~Buckingham, \& J.~DiFranco.
``Asymptotics of Tracy--Widom distributions and the total integral of a Painlevé II function''.
\emph{Comm. Math. Phys.} 280(2), 463--497 (2008).

\bibitem{BaikLiuSilva}
J.~Baik, Z.~Liu, \&  G.L.F.~Silva.
``Limiting one-point distribution of periodic TASEP''.
\emph{Ann. Inst. Henri Poincar\'e Probab. Stat.} 58 (2022), no.~1, 248--302. 

\bibitem{BorodinGorin}
A.~Borodin \& V.~Gorin.
``Moments Match between the KPZ Equation and the Airy Point Process''.
\emph{SIGMA Symmetry Integrability Geom. Methods Appl.} 12 (2016).

\bibitem{BothnerCafassoTarricone}
T.~Bothner, M.~Cafasso, \& S.~Tarricone.
``Momenta spacing distributions in anharmonic oscillators and the higher order finite temperature Airy kernel''.
Preprint arXiv:2101.03557 (2021). To appear in \textit{Ann. Inst. Henri Poincar\'e Probab. Statist.}

\bibitem{CaCl2019}
M.~Cafasso \& T.~Claeys.
``A Riemann--Hilbert approach to the lower tail of the KPZ equation''.
{\em Comm. Pure Appl. Math.} 75(3), 493--540 (2022). 

\bibitem{CaClR2020}
M.~Cafasso, T.~Claeys, \& G.~Ruzza.
``Airy kernel determinant solutions to the KdV equation and integro-differential Painlev\'{e} equations''.
\emph{Comm. Math. Phys.} 386(2), 1107--1153 (2021).

\bibitem{CLDR}
P.~Calabrese, P.~Le Doussal, \& A.~Rosso.
``Free-energy distribution of the directed polymer at high temperature''.
\emph{EPL (Europhysics Letters)} 90(2), 20002 (2010).

\bibitem{Corwin}
I.~Corwin.
``The Kardar–Parisi–Zhang equation and universality class''.
\emph{Random Matrices Theory Appl.} 1(01), 1130001, (2012).

\bibitem{CorwinGhosal}
I.~Corwin \& P.~Ghosal.
``Lower tail of the KPZ equation''.
\emph{Duke Math. J.} 169, no. 7 (2020).

\bibitem{CGKLDT}
I.~Corwin, P.~Ghosal, A.~Krajenbrink, P.~Le Doussal, \& L.C.~Tsai.
``Coulomb-gas electrostatics controls large fluctuations of the Kardar--Parisi--Zhang equation''.
\emph{Phys. Rev. Lett.} 121, 060201 (2018)

\bibitem{DLMS}
D.S.~Dean, P.~Le Doussal, S.N.~Majumdar, \& G.~Schehr.
``Noninteracting fermions at finite temperature in a $d$-dimensional trap: Universal correlations''.
\emph{Phys. Rev. A} 94, 063622 (2016).

\bibitem{DIK}
P.~Deift, A.~Its, \& I.~Krasovsky.
``Asymptotics of the Airy-kernel determinant''.
\emph{Comm. Math. Phys.} 278(3), 643--678 (2008).

\bibitem{DKMVZ1}
P.~Deift, T.~Kriecherbauer, K.T.-R.~McLaughlin, S.~Venakides, \& X.~Zhou.
``Strong asymptotics of orthogonal polynomials with respect to exponential weights''.
{\em Comm. Pure Appl. Math.} 52(12), 1491--1552 (1999).

\bibitem{DZ}
P.~Deift \& X.~Zhou.
``A steepest descent method for oscillatory Riemann--Hilbert problems''.
\emph{Bull. Amer. Math. Soc. (N.S.)} 26(1), 119--123 (1992).

\bibitem{Dotsenko}
V.~Dotsenko.
``Bethe ansatz derivation of the Tracy--Widom distribution for one-dimensional directed polymers".
\emph{EPL (Europhysics Letters)} 90, 20003 (2010).

\bibitem{DubrovinMinakov}
B.~Dubrovin \& A.~Minakov.
``On a class of compact perturbations of the special pole-free joint solution of KdV and $P_2^I$''. 
Preprint arXiv:1901.07470 (2019).

\bibitem{IIKS}
A.R.~Its, A.G.~Izergin, V.E.~Korepin, \& N.A.~Slavnov.
``Differential equations for quantum correlation functions''.
\emph{Internat. J. Modern Phys. B} 4(05), 1003--1037 (1990).

\bibitem{ItsSukhanov}
A.~Its \& V.~Sukhanov.
``Large time asymptotics for the cylindrical Korteweg--de Vries equation. I.''
\emph{Nonlinearity} 33, no.10, 5215--5245 (2020).

\bibitem{Hairer}
M.~Hairer.
``Solving the KPZ equation''.
\emph{Ann. of Math. (2)} 178, no.2, 559--664. (2013). 

\bibitem{HHT}
T.~Halpin-Healy \& K.A.~Takeuchi.
``A KPZ cocktail---shaken, not stirred ... toasting 30 years of kinetically roughened surfaces''.
\emph{J. Stat. Phys.} 160(4), 794--814 (2015).
 
\bibitem{Johansson}
K.~Johansson.
``From Gumbel to Tracy--Widom".
\emph{Probab. Theory Related Fields} 138(1-2), 75--112 (2007).

\bibitem{KPZ}
M.~Kardar, G.~Parisi, \& Y.C.~Zhang.
``Dynamic scaling of growing interfaces''.
\emph{Phys. Rev. Lett.} 56, 889--892 (1986).

\bibitem{Krajenbrink}
A.~Krajenbrink.
``From Painlev\'e to Zakharov--Shabat and beyond: Fredholm determinants and integro-differential hierarchies''.
\emph{J. Phys. A} 54(3), 035001 (2021).

\bibitem{KrajenbrinkLeDoussal}
A.~Krajenbrink \& P.~Le Doussal.
``Simple derivation of the $(-\lambda H)^{5/2}$ tail for the 1D KPZ equation''.
\emph{J. Stat. Mech. Theory Exp.} (2018):  063210.

\bibitem{KLD2}
A.~Krajenbrink \& P.~Le Doussal.
``Linear statistics and pushed Coulomb gas at the edge of $\beta$-random matrices: Four paths to large deviations''.
\emph{EPL (Europhysics Letters)} 125, 20009 (2019).

\bibitem{KLDP18}
A.~Krajenbrink, P.~Le Doussal, \& S.~Prolhac.
``Systematic time expansion for the Kardar--Parisi--Zhang equation, linear statistics of the GUE at the edge and trapped fermions''.
\emph{Nuclear Phys. B} 936, 239--305 (2018).

\bibitem{LDMRS}
P.~Le Doussal, S.~Majumdar, A.~Rosso, \& G.~Schehr.
``Exact short-time height distribution in 1D KPZ equation and edge fermions at high temperature''.
\emph{Phys. Rev. Lett.} 117, 070403 (2016).

\bibitem{KPZKPDoussal}
P.~Le Doussal.
``Large deviations for the Kardar--Parisi--Zhang equation from the Kadomtsev--Petviashvili equation''.
\emph{J. Stat. Mech. Theory Exp.} 2020.4 (2020): 043201.

\bibitem{LNR}
K.~Liechty, G.B.~Nguyen, \& D.~Remenik.
``Airy process with wanderers, KPZ fluctuations, and a deformation of the Tracy--Widom GOE distribution''.
Preprint arXiv:2009.07781 (2020). To appear in \textit{Ann. Inst. Henri Poincar\'e Probab. Statist.}

\bibitem{LiechtyWang}
K.~Liechty \& D. Wang.
``Asymptotics of free fermions in a quadratic well at finite temperature and the Moshe--Neuberger--Shapiro random matrix model''.
\emph{Ann. Inst. Henri Poincar\'e Probab. Statist.} 56, no.2 (2020).

\bibitem{MNS}
M.~Moshe, H.~Neuberger, \& B.~Shapiro.
``Generalized ensemble of random matrices''.
\emph{Phys. Rev. Lett.} 73(11), 1497--1500 (1994). 

\bibitem{Poppe}
C.~P\"{o}ppe.
``The Fredholm determinant method for the KdV equations''.
{\em Phys. D} 13 (1984), no. 1--2, 137--160.

\bibitem{PoppeSattinger}
C.~P\"oppe \& D.~Sattinger.
``Fredholm Determinants and the $\tau$ Function for the Kadomtsev--Petviashvili Hierarchy".
\emph{Publ. Res. Inst. Math. Sci.} 24.4, 505--538 (1988).

\bibitem{Quastel}
J.~Quastel.
``Introduction to KPZ''.
\emph{Current developments in mathematics} 2011 (1) (2011).

\bibitem{QuastelRemenik}
J.~Quastel \& D. Remenik.
``KP governs random growth off a one dimensional substrate''.
{\it Forum Math. Pi} 10, Paper No. e10, 26 pp. (2022).

\bibitem{SasamotoSpohn}
T.~Sasamoto \& H.~Spohn.
``Exact height distributions for the KPZ equation with narrow wedge initial condition''.
\emph{Nuclear Phys. B} 834(3), 523--542 (2010).

\bibitem{SasorovMeersonProlhac}
P.~Sasorov, B.~Meerson, \& S.~Prolhac.
``Large deviations of surface height in the $1+1$-dimensional Kardar--Parisi--Zhang equation: exact long-time results for $\lambda H < 0$''. \emph{J. Stat. Mech. Theory Exp.} 2017(6), 063203 (2017).

\bibitem{TW}
C.A.~Tracy \& H.~Widom.
``Level-spacing distributions and the Airy kernel''.
\emph{Comm. Math. Phys.} 159(1), 151--174 (1994).

\bibitem{Tsai}
L.C.~Tsai.
``Exact lower tail large deviations of the KPZ equation''.
\emph{Duke Math. J.} 171(9), 1879--1922 (2022).

\end{thebibliography}
\end{document}